\documentclass{article}
\usepackage{odonnell}
\usepackage{xspace,braket}

\newcommand{\Ugp}[1]{\mathrm{U}(#1)}
\newcommand{\SUgp}[1]{\mathrm{SU}(#1)}

\newcommand{\Ogp}[1]{\mathrm{O}(#1)}
\newcommand{\SOgp}[1]{\mathrm{SO}(#1)}
\newcommand{\PSOgp}[1]{\mathrm{PSO}(#1)}
\newcommand{\JustG}{\mathrm{G}}
\newcommand{\Ggp}[1]{\mathrm{G}(#1)}

\newcommand{\opnorm}[1]{\norm*{#1}_{\mathrm{op}}}

\newcommand{\K}{\mathrm{K}}

\newcommand{\dRie}{d_{\textnormal{Rie}}}

\newcommand{\Fr}{\mathrm{Fro}}

\newcommand{\rnote}[1]{*\marginpar{\tiny *Ryan: \bf #1}}
\newcommand{\pnote}[1]{\footnote{{\bf Pedro:}  #1}}

\newcommand{\rasnote}[1]{\footnote{{\bf Rocco:}  #1}}

\DeclareMathOperator{\Img}{\textnormal{Im}}

\renewcommand{\c}{c}

\title{Explicit orthogonal and unitary designs}

\author{Ryan O'Donnell\thanks{Computer Science Department, Carnegie Mellon University.}${~}^{\tiny{\textcircled{r}}}$\and Rocco A. Servedio\thanks{Department of Computer Science, Columbia University.}${~}^{\tiny{\textcircled{r}}}$ \and Pedro Paredes\thanks{Department of Computer Science, Princeton University.}${~}^{\tiny{\textcircled{r}}}$\\  \and${~}^{\tiny{\textcircled{r}}}${\small Author ordering randomized.}}

\date{\today}

\begin{document}

\maketitle
\begin{abstract}

    We give a strongly explicit construction of $\eps$-approximate $k$-designs 
    for the orthogonal group~$\Ogp{N}$ and the unitary group~$\Ugp{N}$, for $N=2^n$.
    Our designs are of cardinality~$\poly(N^k/\eps)$ (equivalently, they have seed length $O(nk + \log(1/\eps)))$; up to the polynomial, this matches the number of design elements used by the construction consisting of completely random matrices.
    \end{abstract}


\newcommand{\Fro}{\mathrm{Fro}}
\newcommand{\Rie}{\mathrm{Rie}}
\newcommand{\aux}{\mathrm{aux}}
\newcommand{\conj}{C}

\section{Introduction}
\label{sec:overall-goal}

The main new result in our work is the following:
\begin{theorem} \label{thm:main}
    Let $N = 2^n$ and let $\Ggp{n}$ denote either the orthogonal group $\Ogp{N}$ or the unitary group~$\Ugp{N}$.  
    Then for any $k = k(n)$, there is an explicit $\eps$-approximate $k$-design for~$\Ggp{n}$ of cardinality~$\poly(N^k/\eps)$; i.e., samplable using  a seed  of just $O(nk + \log(1/\eps))$ truly random bits.
Moreover, these designs are strongly explicit in the following sense:  
    (i)~each output matrix is given by an $n$-qubit circuit consisting of $S = \poly(nk) \log(1/\eps)$ gates, each gate being either $\mathrm{CNOT}$ or one of a few fixed and explicitly specified $1$-qubit gates; (ii)~the algorithm that takes as input a seed and outputs the associated circuit runs in deterministic $\poly(S)$ time.
\end{theorem}

In the unitary case,  similar results in the literature only discuss the regime $k \leq \poly(n)$~\cite{HL09,Sen18}, or have polynomially worse seed length~\cite{BHH16,Haf22}.
In contrast, our result holds for all $k$ (even exponentially large as a function of $n$, or larger), and achieves a seed length which matches, up to constant factors, that of a random construction.
A significant motivation for our work was the orthogonal case, where the only prior works we know of are~\cite{KM15,HH21}, which we discuss below.  Our  \Cref{thm:main} provides  the efficient orthogonal designs needed for Kothari and Meka's near-optimal pseudorandom generators for spherical caps~\cite{KM15}.

Let us now discuss the general context for our result.

\paragraph{Derandomization.}  Let $\mathcal{G}$ be a class of objects, and assume informally that each object has ``size''~$N^{\Theta(1)}$ (think, e.g., of strings of length~$N$, or $N \times N$ matrices).  To choose an object from the uniform probability distribution on~$\mathcal{G}$ typically requires using~$\Omega(N)$ truly random bits.  
A broad goal in  derandomization is to identify a useful notion of ``pseudorandomness'' for probability distributions on~$\mathcal{G}$, and then to show that one can sample from such a distribution using just $r \ll N$ truly random bits.\footnote{In this introduction, we refer to sampling uniformly from a set of size~$R$ as ``using $\log_2 R$ truly random bits''.}
An additional goal is for the sampling algorithm to be \emph{efficient}; i.e., the sampled object should be produced by a deterministic $\poly(r)$-time algorithm, given the truly random seed of length~$r$.  
In this case, since the sampler has only $2^r$ possible outcomes yet the total number of objects is exponential in~$N$, it must be the case that the sampler represents the output objects in a ``succinct'' way.  Informally, if it is possible to efficiently compute with objects represented in this succinct way, the sampler is said to be ``strongly explicit''.

\paragraph{Exact $k$-wise independence.}
One of the most common and useful notions of pseudorandomness is that of \emph{bounded independence}.
For random objects with $N^{\Theta(1)}$ ``entries'' (``coordinates''/``dimensions''),  it often suffices for applications if the objects are merely ``$k$-wise independent'' for some $k \ll N$.  This means that the object looks truly random whenever only~$k$ entries are inspected.
In this case one may hope that the object can be sampled using a  random seed of length  just $O(k \log N)$ bits.

The paradigmatic example of this comes from $k$-wise independent length-$N$ Boolean strings.
Using results from coding theory~\cite{ABI85}, it has long been known that $O(k \log N)$ random bits suffice to efficiently sample a precisely $k$-wise independent string $\bx \in \{0,1\}^N$ (meaning that $(\bx_{i_1}, \dots, \bx_{i_k})$ is perfectly uniformly distributed on~$\{0,1\}^k$ for any $i_1, \dots, i_k$).  

For other kinds of random objects, obtaining \emph{exact} $k$-wise independence seems extremely difficult.  
Take the case of random permutations, where $\bpi \in S_N$ is said to be $k$-wise independent if~$(\bpi(i_1), \dots, \bpi(i_k))$ is uniformly distributed on $\binom{[N]}{k}$ for any distinct $i_1, \dots, i_k$.  
While simple efficient methods for generating $2$- and $3$-wise independent permutations using $O(\log N)$ random bits 
are known, for any constant $k \geq 4$ the best known efficient construction uses~$\Theta(N)$ random bits~\cite{FPY15}.  
The situation is similar for random unitary matrices, where $\bU \in \Ugp{N}$ is said to be drawn from a $k$-design if $\E[\bU_{i_1 j_1} \cdots \bU_{i_k j_k}]$ is equal to what it would be if $\bU$ were Haar-distributed on~$\Ugp{N}$ (and similarly if any  subset of the entries $\bU_{i_t j_t}$ in the product were replaced with their complex conjugates).
Here it is known how to efficiently construct exact $2$-designs using $O(\log N)$ bits~\cite{DLT02}, and exact $3$-designs using $O(\log^2 N)$ bits~\cite{Web16}, but good constructions of exact $k$-designs for $k \geq 4$ are lacking (see, e.g.,~\cite{BNZZ19}).

\paragraph{Approximate $k$-wise independence.}
Given these issues, it is natural to seek \emph{$\eps$-approximate} $k$-wise independence ($k$-designs).  Here it is important to carefully define the precise notion of  ``approximate'', as different natural notions are often only equivalent if one is willing to change~$\eps$ by a factor that is exponential in~$k$.  
For example, in the context of Boolean strings in $\{\pm 1\}^N$, a weak\ignore{but useful} notion of $(\eps,k)$-wise independence is that $\abs{\E[\bx_{i_1} \cdots \bx_{i_k}]} \leq \eps$ for all $k$-tuples of distinct values $i_1, \dots, i_k$. 
 Naor and Naor~\cite{NN93} showed that $O(\log (nk/\eps))$ random bits suffice to  explicitly generate such a distribution, where we write $n = \log_2 N$.
However, to get the stronger guarantee that every $k$~bit positions are $\eps$-close to the uniform distribution in statistical distance, one needs $(\eps 2^{-k},k)$-wise independence (see, e.g.,~\cite{AGHP92}), and hence the number of random bits used in known constructions is $O(k + \log(n/\eps))$.  
In general, for $q$-ary rather than $2$-ary (Boolean) strings, the  seed-length penalty becomes $O(k \log q)$.
So if, e.g., one wants a distribution on $\Z_N^N$ in which every $k$ coordinates have statistical difference~$\eps$ from uniform where $N=2^n$,\footnote{Cf.\ achieving $\eps$-approximate $k$-wise independent permutations from~$S_N$.} then the best known explicit constructions use  $O(k n + \log(1/\eps))$ random bits.\\

In this work, we give a common framework for randomness-efficient generation of  approximately $k$-wise independent distributions over \emph{groups}, particularly subgroups of the unitary group.  
Our framework applies to, e.g.,  the group of $q$-ary strings $\Z_q^N$ (realized as diagonal matrices with $q$th roots of unity as the diagonal entries), the permutation group~$S_N$ (realized as $N \times N$ permutation matrices), the orthogonal group~$\Ogp{N}$, and the unitary group~$\Ugp{N}$ (with $N = 2^n$).  We will not discuss strings further in this work, as they are already very well studied.  We first describe prior work on the other three groups, and then explain our new general method.

\paragraph{Permutations.}  Explicit approximate $k$-wise independent permutations have found a wide variety of applications; e.g., in cryptography~\cite{KNR09}, hashing/dimensionality reduction~\cite{LK10,KN14}, and explicit constructions of expanders~\cite{MOP22}.
One method for creating them was initiated by Gowers~\cite{Gow96}, who showed that a random $n$-qubit circuit composed of~$\poly(n,k)\log(1/\eps)$ ``classical'' $3$-qubit gates (i.e., permutations on $\{0,1\}^3$) yields an $\eps$-approximate $k$-wise independent permutation on $S_{2^n}$.  
(Note that since the circuit size is polynomial rather than linear in $nk$, the randomness-efficiency of \cite{Gow96} is polynomially worse than the $O(nk + \log(1/\eps))$ random bits needed by a non-explicit random construction.)
Gowers's technique was to lower-bound the spectral gap of the random walk on a related graph by $1/\poly(n,k)$. 
(See~\cite{HMMR05,BH08} for improvement of the spectral gap to $1/\wt{O}(k^2 n^2)$.)
Subsequently, using techniques related to space-bounded walks in graphs~\cite{Rei08}, Kaplan--Naor--Reingold~\cite{KNR09} derandomized this ``truly random walk'' to achieve efficient $\eps$-approximate $k$-wise independent permutations on $S_{2^n}$ with seed length $O(k n + \log(1/\eps))$, matching the (inexplicit) random bound. Around the same time, Kassabov~\cite{Kas07} got the same seed length (without requiring $N$ to be a power of~$2$) via a sophisticated construction of a constant-size generating set for any~$S_N$ that makes the resulting Cayley graph an expander.

\paragraph{Unitary matrices.} Introduced to the quantum computing literature in~\cite{DCEL09}, explicit $\eps$-approximate $k$-designs for the unitary group have had a wide variety of applications, from randomized benchmarking of quantum gate sets~\cite{ZZP17}, to efficient state and process tomography~\cite{HKOT23}, to understanding quantum state and unitary complexity~\cite{RY17,BCHKP21}.  
Previously, works on constructing approximate unitary designs have chiefly focused on achieving ``strong explicitness''  rather than on randomness-efficiency.
In particular, the goal has been to show that a \emph{truly} random $n$-qubit quantum circuit composed of~$S = \poly(n,k) \log(1/\eps)$  gates (i.e.~each gate is a Haar random unitary operator on a constant number of uniformly randomly chosen qubits) constitutes an $\eps$-approximate $k$-design for~$\Ugp{2^n}$.  The breakthrough in this area came from the work of Brand{\~a}o, Harrow, and Horodecki~\cite{BHH16}, who showed that $S = O(n^2 k^{10.5} \log(1/\eps))$ suffices for $k \le 2^{\Omega(n)}$. 
(See also \cite{HL09} for an earlier construction using $\poly(n,k)\log(1/\eps)$ gates when $k=O(n/\log n)$, and \cite{Sen18} for a construction in the $k=\poly(n)$ regime.)
 Further work has been done on improving the circuit depth and the exponent on~$k$; see~\cite{Haf22}. 
Ours is the first work to derandomize these results and achieve a seed length that is \emph{linear} rather than polynomial in $n$ and $k$, and that works for all $k$, thus matching the non-explicit random construction.
As an example application of our result for unitary matrices, by applying~\cite{BCHKP21} we get an efficient deterministic procedure for outputting~$2^{O(nk)}$ many $n$-qubit unitary circuits of $\poly(nk)$ gates such that at least $2^{\Omega(n k)}$ of them (a polynomially large fraction)   have strong quantum circuit complexity~$\Omega(\frac{n}{\log n} k)$ (provided $k \leq 2^{\Omega(n)}$).

\paragraph{Orthogonal Matrices.}  It is natural to think that  designs for $\Ogp{N}$ and $\Ugp{N}$ should be related (and indeed orthogonal designs have played a role in randomized benchmarking for quantum circuits~\cite{HFGW18}).  However there is no obvious reduction between the tasks of constructing $\eps$-approximate  $k$-designs for the two groups.  
The first paper we are aware of that attempts to explicitly construct approximate orthogonal designs is~\cite{KM15}. 
That work used explicit orthogonal designs with $O(kn + \log(1/\eps))$ seed length as the core pseudorandom object underlying its  state-of-the-art pseudorandom generator for linear threshold functions on~$\mathbb{S}^{n-1}$.
Unfortunately, there was an error in their construction of these designs.\footnote{The error is in the interpretation of the main result of \cite{BG12} that is used to establish Corollary~6.1 of \cite{KM15}. Corollary~6.1 claims that the spectral gap established by \cite{BG12} for $\SUgp{N}$ is independent of $N$, but this is in error \cite{Kothari22}; indeed, as noted in \cite{BHH16} after their Corollary~7, ``the proof [in \cite{BG12}] does not give any estimate of the dependency of the spectral gap on $N$.''}
Fixing this error was a key motivation for the present work, and indeed our \Cref{thm:main} provides the crucial ingredient needed for the pseudorandom generators of~\cite{KM15}.

Some of our technical ideas for handling the orthogonal group are drawn from the work of Haferkamp and Hunter-Jones, who showed (Theorem~9 of \cite{HH21}) that truly random local orthogonal $n$-qudit circuits of size $\poly(n,k)\log(q/\eps)$ constitute $\eps$-approximate $k$-designs for~$\Ogp{q^n}$, provided $q \geq 8k^2$. This result has suboptimal randomness complexity because of the polynomial rather than linear dependence on $n$ and $k$, and only gives approximate $k$-designs for small values of $k$.

\ignore{
}

\subsection{Our framework} \label{sec:our-framework}
As stated earlier, we are interested in $k$-wise independent distributions over groups, particularly the symmetric, orthogonal, and unitary groups.  For each such group~$\JustG$,  the notion of ``$k$-wise independence'' is defined through a certain \emph{representation}~$\rho^k$ of the group.  
Informally, we say a distribution $\calP$ on~$\JustG$ is approximately $k$-wise independent if
\begin{equation}    \label{eqn:vague-design}
    \E_{\bg \sim \calP}[\rho^k(\bg)] \approx \E_{\bg \sim \JustG}[\rho^k(\bg)],
\end{equation}
where on the right-hand side $\bg$ is drawn from the Haar distribution on~$\JustG$.\footnote{Here and throughout, whenever $\JustG$ is a compact Lie group we write $\bg \sim \JustG$ to denote that $\bg$ is drawn according to the Haar distribution; in particular, this is the uniform distribution if $\JustG$ is finite.}

Let us consider our three example groups~$\JustG$, starting with the orthogonal group~$\Ogp{N}$.  In this case, the associated representation $\rho^k$ is on $(\C^N)^{\otimes k}$, and it maps $R \in \Ogp{N}$ to $R^{\otimes k}$.  
In other words, specialized to the orthogonal group, \Cref{eqn:vague-design} asserts that~$\calP$ is an approximate $k$-design on~$\Ogp{N}$ provided
\begin{equation}   
    \E_{\bR \sim \calP}[\bR^{\otimes k}] \approx \E_{\bR \sim \Ogp{N}}[\bR^{\otimes k}].
\end{equation}
As matrices, the entries of~$\rho^k(\bR) = \bR^{\otimes k}$ are degree-$k$ monomials in the entries of~$\bR$, and thus \Cref{eqn:vague-design} (qualitatively) implies that any degree-$k$ polynomial in the entries of~$\bR$ has approximately the same expectation under~$\calP$ as it has under the Haar distribution.  
This is the usual meaning of approximate $k$-wise independence in theoretical computer science, and is often how the notion is used in applications.  For the unitary matrices~$U$ we wish to consider polynomials in both the entries of~$U$ and their complex conjugates; thus the appropriate representation of $\Ugp{N}$ is $\rho^{k,k}$ on $(\C^N)^{\otimes 2k}$ defined by 
\begin{equation}
    \rho^{k,k}(U) = U^{\otimes k} \otimes \overline{U}^{\otimes k}.
\end{equation}
Actually, to unify notation we will work with $\rho^{k,k}$ even when studying the orthogonal group $\Ogp{N} \leq \Ugp{N}$; in this case of course $\rho^{k,k}$ is equivalent to $\rho^{2k}$, and we won't be concerned with the difference between $k$~and~$2k$.  
(Note that if $k$ is odd then the expectation of any degree-$k$ monomial in the entries of $\bR$, $\bR \sim \Ogp{N}$, is trivially 0.)
Finally, for the symmetric group $S_N \leq \Ugp{N}$ we could again use $\rho^{k,k}$, but previous work has (implicitly) used an alternative representation, which we'll call~$\calW^k$.  
To define it, let $[N]_{(k)}$ denote the set of sequences of distinct indices $i_1, \dots, i_k \in [N]$ and let $\C^{[N]_{(k)}}$ denote the (complex) vector space with orthonormal basis vectors $\ket{i_1\cdots i_k}$.  Then the representation~$\calW^k$ is defined on $\pi \in S_N$ via $\calW^k(\pi)\ket{i_1 \cdots i_k} = \ket{\pi(i_1) \cdots \pi(i_k)}$.
This representation~$\calW^k$ is the one usually associated to $k$-wise independence on~$S_N$, with the analogue of \Cref{eqn:vague-design} asserting that $\E_{\bpi \sim \calP}[(\bpi(i_1), \dots, \bpi(i_k))]$ is close to being uniformly distributed on~$[N]_{(k)}$ for each $(i_1, \dots, i_k) \in [N]_{(k)}$.\\

A first way to try to achieve approximate $k$-wise independence on~$\JustG \in \{S_N, \Ogp{N}, \Ugp{N}\}$ is through a Markov chain.  Suppose $P \subset \JustG$ is a set (closed under inverses) of size~$\poly(n)$, where $n = \log_2 N$.  Consider the random walk on~$\JustG$ that starts at~$\Id$ and multiplies by a uniformly random element of~$P$ at each step.  We may hope that after, say, $\poly(nk)\log(1/\eps)$ steps, the resulting distribution~$\calP$ on~$\JustG$ is close enough to the Haar distribution on~$\JustG$ that \Cref{eqn:vague-design} holds.  As alluded to earlier, results of this form were previously shown for $\JustG = S_{2^n}$ (starting with~\cite{Gow96}) and for $\JustG = \Ugp{2^n}$ (starting with~\cite{BHH16}).  One significant contribution of the present work is to generalize the latter to apply also to~$\Ogp{2^n}$ (or, more precisely and essentially equivalently, its connected subgroup~$\SOgp{2^n}$).
Specifically, in \Cref{sec:ini-gap,sec:large-m,sec:small-m}, our goal will essentially be to show the following:
\begin{theorem} \label{thm:2}
Fix $n \geq 4$ and let $P_n \subset \SOgp{2^n}$ denote the $O(n^2$)-sized multiset of all $n$-qubit, $1$-gates circuits consisting of either $\mathrm{CNOT}$ (on some $2$~qubits) or $\mathrm{Q} = \begin{bmatrix} 
        3/5 & -4/5 \\
        4/5 & \phantom{-}3/5 \\
    \end{bmatrix}$ on some $1$~qubit, and then closed under negation and  inverses.
    Then for any $k \geq 1$, 
    \begin{equation}
        \opnorm{\E_{\bg \sim P_n}[\rho^{k,k}(\bg)] - \E_{\bg \sim \SOgp{2^n}}[\rho^{k,k}(\bg)]} \leq 1 - \frac{1}{n \cdot \poly(k)}.
    \end{equation}
    A similar statement holds for $\SUgp{2^n}$ with the $1$-qubit $\mathrm{H}$,~$\mathrm{S}$, and~$\mathrm{T}$ gates replacing~$\mathrm{Q}$.
\end{theorem}
\noindent (See \Cref{thm:XXX} for more details.  In \Cref{sec:BHH-extraction-stuff} we will pass from $\SOgp{2^n}$ and $\SUgp{2^n}$ to $\Ogp{2^n}$ and $\Ugp{2^n}$; our analysis in \Cref{sec:ini-gap,sec:large-m,sec:small-m} is carried out in the ``special'' versions of these groups for technical reasons that will become clear in \Cref{sec:ini-gap}, specifically \Cref{sec:overview}.) As we discuss in \Cref{sec:overview}, the high-level approach we take to establish \Cref{thm:2} extends an approach from \cite{HH21}.

Given the above theorem, we could improve its right-hand side to~$\eps/2^{nk}$ by forming~$\bg$ as a product of $n \cdot \poly(k) \cdot \log(2^{nk}/\eps)=\poly(nk) \cdot \log(1/\eps)$ uniformly randomly elements from~$P_n$. (See \Cref{def:orthogonal-design}, where we give a precise definition of ``$\eps$-approximate $k$-design'', for a discussion of why $\eps/2^{nk}$ is the correct bound for the right-hand side.)
The resulting distribution on~$\Ogp{2^n}$  would be an $\eps$-approximate $k$-design, but unfortunately,
drawing from this distribution would require a seed of $\poly(nk) \cdot \log(1/\eps)$  truly random bits, which leaves something to be desired from the standpoint of randomness-efficiency.

To improve this and match the randomness-efficiency of the random construction, one may attempt to apply the method of ``pseudorandom walks on consistently labeled graphs'' from~\cite{Rei08,RTV06}, or ``derandomized squaring'' from~\cite{RV05}.  
This is the approach taken in~\cite{KNR09} for the symmetric group, where the evolving value of~$\calW^k(\bpi)\ket{i_1 \cdots i_k}$ can be thought of as a random walk on a graph with vertex set~$[N]_{(k)}$.  In the setting of \Cref{thm:2} there is no graph.
Nevertheless, in \Cref{sec:pseudorandom-walks} we will show how derandomized squaring can be slightly generalized to obtain the following result (a similar generalization appeared recently in~\cite{JMRW22}):
\begin{theorem} \label{thm:abbrev-derando-walks}
    (Abbreviated version of \Cref{thm:derando-walks}.)
    Given $c, \delta, \eps$, there is a strongly explicit deterministic  algorithm that outputs a sequence~$\calP$ of $O(c/\poly(\delta \eps))$ ``monomials'' over the symbols $u_1, \dots, u_c, u_1^\dagger, \dots, u_c^\dagger$, each of length $O(\log(1/\eps)/\poly(\delta))$, such that $\opnorm{\avg_{\bm \in \calP} \{ \bm(\calU) \}} \leq \eps$ whenever $\calU = (U_1, \dots, U_c)$ is a sequence of unitaries with $\opnorm{\avg_{i \in [c]}\{U_i\}} \leq 1- \delta$.  
    (Here $m(\calU)$ denotes the product of $U_i$'s and $U_i^\dagger$'s obtained by substituting $u_i = U_i$ in~$m$.)
\end{theorem}

Taking the $\delta$ of \Cref{thm:abbrev-derando-walks} to be the $1/\poly(nk)$ of \Cref{thm:2}, and the unitaries $\calU = (U_1, \dots, U_c)$ to correspond to the 1-gate circuits of \Cref{thm:2}, we obtain strongly explicit $\eps$-approximate $k$-designs as described in \Cref{thm:main} for the special orthogonal and special unitary groups.  A simple modification gives corresponding designs for the unitary and orthogonal groups, thus yielding \Cref{thm:main}.

\subsection{Organization of this paper}

In \Cref{sec:BHH-extraction-stuff} we give the detailed argument explaining how an initial spectral gap of the sort given by \Cref{thm:2} and the generalized ``derandomized squaring'' result given by \Cref{thm:abbrev-derando-walks} together yield efficient explicit approximate designs for the orthogonal and unitary groups. The rest of the paper is devoted to establishing the two necessary ingredients \Cref{thm:2} and \Cref{thm:abbrev-derando-walks}.  \Cref{sec:ini-gap} gives our general framework for establishing the initial spectral gap for the special unitary and special orthogonal groups; as we explain there, a crucial step in this framework is establishing a spectral gap for a certain  ``auxiliary'' $m$-qubit random walk which was inspired by the analysis of \cite{HH21}. Similar to \cite{HH21}, it turns out that to analyze this auxiliary random walk, two quite different technical arguments are required depending on whether the tensor power $k$ is ``large'' or ``small'' compared to the number of qubits $m$; we give these two arguments in \Cref{sec:large-m} and \Cref{sec:small-m} respectively. Finally, we provide the necessary analysis of the generalized ``derandomized pseudorandom walks'' in \Cref{sec:pseudorandom-walks}.

\subsection{Notation and preliminaries} \label{sec:prelim}

%

To give our constructions it is convenient to use  the language of quantum computing, even when the group involved is the orthogonal or symmetric group.
We will generally consider operators on $\C^N$, where $N = 2^n$ for some $n \in \N^+$.
We identify $\C^N = (\C^2)^{\otimes n}$ and think of the tensor factors as corresponding to~$n$ qubits.
\begin{notation}    \label{not:gatepositioning}
Let $g \in \Ugp{2^\ell}$, 
    thought of as an $\ell$-qubit ``gate'' and let $e =(i_1,\dots,i_\ell)$ be a sequence of $\ell$ distinct elements of $[n]$, i.e.\ $e \in [n]_\ell$.   (Here $\ell$ should be thought of as ``much less than $n$'', in particular we will be interested in constant $\ell$.) We use the notation $g_e$ for the operator in 
$\Ugp{N}$ 
    defined by applying~$g$ on qubits (i.e., tensor factors) $i_1, \dots, i_\ell$ (in that order) and applying the identity operator on the remaining $n-\ell$ qubits.
    When $e \in \binom{[n]}{\ell}$ is a set rather than a sequence, we assume the increasing order on its elements.
\end{notation}

We write $A^\dagger$ to denote the conjugate transpose of a complex matrix $A$, $\opnorm{A}$ to denote the operator norm, and $\norm{A}_1$ to denote its Schatten $1$-norm.
We use bold font to denote random variables.


\section{A general framework:  Explicit $k$-wise independent permutations, orthogonal designs, and unitary designs} \label{sec:BHH-extraction-stuff}

Let $\Ggp{n}$ be a subgroup of $\Ugp{n}$ (the key examples to keep
in mind are the group of permutations on $2^n$ elements, the
$2^n$-dimensional orthogonal group, the $2^n$-dimensional unitary
group itself, and the ``special'' versions of the latter two).  In
light of \Cref{thm:abbrev-derando-walks}, given a probability
distribution on a subset of $\Ggp{n}$, we would like to understand
how fast the associated random walk mixes vis-a-vis a particular
representation, namely the $k$-wise tensor product representation
(since that representation corresponds to $k$-wise independence).

Let $\calP$ be a probability distribution 
on $\Ggp{n}$ that is
symmetric (meaning that $\bg^{-1} = \bg^\dagger$ is distributed
as~$\calP$ when $\bg$ is), and let~$\rho$ be a unitary representation
of $\Ggp{n}$.  Note that since $\rho$ is unitary, $\E_{\bg \sim
  {\cal P}}[\rho(\bg)]$ is a Hermitian operator with real eigenvalues
lying in $[-1,1]$. Since our goal is $k$-wise independence, the
representations that are of interest to us are $k$-wise tensor product
representations:
\begin{notation} [$k$-wise tensor product representations]
  For any $k \in \N^+$, we will write $\rho^{k,k}_{2^n}$ for the
  (complex) representation of $\Ggp{n}$ defined by
  \begin{equation}
    \rho^{k,k}_{2^n}(g) = g^{\otimes k,k} \coloneqq g^{\otimes k} \otimes \overline{g}^{\otimes k},
  \end{equation}
  where $\overline{g}$ denotes the complex conjugation of matrix~$g$.
\end{notation}

There are several different definitions of $\eps$-approximate $k$-designs in the literature, all of which are equivalent if one is willing to lose factors of $2^{nk}$ on~$\eps$.  For definiteness, we choose the $1$-norm definition from \cite{HL09}.  (One could also equivalently use the notion from Kothari--Meka~\cite{KM15}, again up to $2^{nk}$ factors.)
\begin{definition} \label{def:orthogonal-design}
  A distribution ${\cal P}$ on a finite subset of matrices from
  $\Ggp{n}$ is an \emph{$\eps$-approximate $k$-design for
    $\Ggp{n}$} if

  \begin{equation}
    \norm*{\E_{\bg \sim \cal P}[\rho_{2^n}^{k,k}(\bg)] - \E_{\bg \sim \Ggp{n}}[\rho_{2^n}^{k,k}(\bg)]}_1 \leq \eps
  \end{equation}
  (where $\norm*{\cdot}_1$ denotes the Schatten $1$-norm).
  We remark that the above condition is implied by the following operator-norm bound:
  \begin{equation}
    \opnorm{\E_{\bg \sim \cal P}[\rho_{2^n}^{k,k}(\bg)] - \E_{\bg \sim \Ggp{n}}[\rho_{2^n}^{k,k}(\bg)]} \leq \eps/2^{nk},
  \end{equation}
  and indeed we will establish our approximate design results by going through the operator norm.
\end{definition}

\ignore{Like mentioned before, the outcome of our analysis will provide
approximate designs for $\SUgp{2^n}$ and $\SOgp{2^n}$ instead of
$\Ugp{2^n}$ and $\Ogp{2^n}$. Later in this section we will show how to
pass from $\SUgp{2^n}$ and $\SOgp{2^n}$ to $\Ugp{2^n}$ and
$\Ogp{2^n}$, but first we note a related fact that will be relevant
later.

\begin{fact} \label{fact:su-to-u}
  For any $k < QQQ$ we have: \pnote{What is the right dependency on $k$? Also, who do we cite for this?}

  \begin{equation}
    \E_{\bg \sim \Ugp{2^n}}[\rho_{2^n}^{k,k}(\bg)] = \E_{\bg \sim \SUgp{2^n}}[\rho_{2^n}^{k,k}(\bg)].
  \end{equation}
\end{fact}}

Often we will study the operator $\E_{\bg \sim \calP}[\rho^k(\bg)]$
through its ``Laplacian'', which we define as follows:

\begin{definition} \label{def:LW}
  We define the ``Laplacian''
    
  \begin{equation}
    L_{\calP}(\rho) = \Id - \E_{\bg \sim \calP}[\rho(\bg)],
  \end{equation}
  a self-adjoint (since~$\calP$ is symmetric) operator satisfying the
  following inequalities (in the PSD order):
  \begin{equation}    \label[ineq]{ineq:2}
    0 \leq L_{\calP}(\rho) \leq 2 \cdot \Id.
  \end{equation}
\end{definition}

\begin{notation}    \label{not:P}
  In the preceding definition, we abuse notation as follows: In place
  of~$\calP$ we may write a finite (multi)set~$P \subset \Ggp{n}$,
  in which case the uniform distribution on~$P$ is understood.  We may
  also write ``$\Ggp{n}$'' in place of~$\calP$, in which case the
  uniform (Haar) distribution is understood.  Finally, if $\calP$ now
  denotes a distribution on~$\Ggp{\ell}$, and $E \subseteq
  [n]_\ell$, we write $\calP \times E$ for the distribution
  on~$\Ggp{n}$ given by choosing $\bg \sim \calP$, independently
  choosing $\be \sim E$ (uniformly), and finally forming~$\bg_{\be}$.
\end{notation}

\begin{definition} \label{def:twiddle}
  Given a symmetric probability distribution $\calP$ as in
  \Cref{def:LW}, we define its ``lazy'' version, $\widetilde{\cal P}$,
  to be the distribution which is an equal mixture of ${\cal P}$ and
  the point distribution supported on the identity element $\Id$ (note
  that $\widetilde{\cal P}$ is also a symmetric distribution).
  Similar to \Cref{def:LW}, we have that $\E_{\bg \sim \widetilde{\cal
      P}}[\rho(\bg)]$ is a Hermitian operator but now with real
  eigenvalues lying in $[0,1]$, and we have the PSD inequalities
  \begin{equation}    \label[ineq]{ineq:2twiddle}
    0 \leq L_{\widetilde{\calP}}(\rho) \leq \Id.
  \end{equation}
\end{definition}

\begin{fact}
  In the setting of \Cref{def:LW,def:twiddle}, $L_{\Ggp{n}}(\rho)$
  is an orthogonal projection operator, and for any symmetric $\calP$
  we have that
  \begin{equation}
    \ker L_{\Ggp{n}}(\rho) \subseteq \ker L_{\calP}(\rho)
  \end{equation}
  always holds (because for \emph{every}~$g_0$ in the support
  of~$\calP$ we have $\rho(g_0) \Pi = \Pi$, where $\Pi = \E_{\bg \sim
    \Ggp{n}}[\rho(\bg)]$).  From this, and
  \Cref{ineq:2,ineq:2twiddle}, we also get
    \begin{equation}    \label[ineq]{ineq:fantasy}
      L_{\Ggp{n}}(\rho) \geq \tfrac12 \cdot L_{\calP}(\rho),
    \end{equation}
    \begin{equation} \label[ineq]{ineq:fantasytwiddle}
      L_{\Ggp{n}}(\rho) \geq L_{\widetilde{\calP}}(\rho).
    \end{equation}
\end{fact}

As \Cref{ineq:fantasy,ineq:fantasytwiddle} contain a surfeit of
symbols, one may wish to read them respectively as
\begin{equation}    \label[ineq]{ineq:rando}
  \text{``(randomizing $n$ qubits)} \geq \tfrac12 \cdot\text{($\calP$-pseudorandomizing $n$ qubits)} \quad \text{[vis-a-vis~$\rho$]''},
\end{equation}
\begin{equation} \label[ineq]{ineq:randotwiddle}
  \text{``(randomizing $n$ qubits)} \geq \text{($\widetilde{\calP}$-pseudorandomizing $n$ qubits)} \quad \text{[vis-a-vis~$\rho$]''},
\end{equation}
with the ``''$\geq \frac12 \cdot$'' part pronounced ``is at least
$\frac12$ as good as''.

It will be convenient to use the Laplacian operator in some of the
steps in the following sections, even though we ultimately want
statements about the expectation operator. To convert between the two
we will use the following:

\begin{fact}
  For any unitary representation $\rho$, $L_{\calP}(\rho) \leq \eps \cdot
  L_{\Ggp{n}}(\rho)$ is equivalent to
  
  \begin{equation}
    \opnorm{\E_{\bg \sim \cal P}[\rho(\bg)] - \E_{\bg \sim \Ggp{n}}[\rho(\bg)]} \leq \eps.
  \end{equation}
\end{fact}

\subsection{Initial spectral gaps for $S_N$, $\SOgp{N}$ and $\SUgp{N}$} \label{ssec:ini-gaps}

Here we summarize all of the non-trivial spectral gaps that we will
amplify using \Cref{thm:abbrev-derando-walks}.

\begin{theorem}[\cite{BH08}] \label{thm:ini-gap-perm}
  For any $k \geq 1$, there is a (multi)set $P_{S_{2^n}}$ of cardinality
  $O(n^3)$ such that
  \begin{equation}
    \opnorm{\E_{\bg \sim P_{S_{2^n}}}[\calW^k_{2^n}(\bg)] - \E_{\bg \sim S_{2^n}}[\calW^k_{2^n}(\bg)]} \leq 1 - \frac{1}{\wt{O}(k^2 n^2)}.
  \end{equation}
\end{theorem}

Recall that the representation~$\calW^k_{2^n}$ is defined on $g \in
S_{2^n}$ via $\calW^k_{2^n}(g)\ket{i_1 \cdots i_k} = \ket{g(i_1)
  \cdots g(i_k)}$. The set $P_{S_{2^n}}$ mentioned above is the set of
``simple 3-bit permutations''. This is the set of permutations
$f_{i,j_1,j_2,h}$, where $i, j_1, j_2 \in [n]$ are all distinct, and
$h$ is a Boolean function on $\{0, 1\}^2$, which maps $(x_1, \ldots,
x_n) \in \{0, 1\}^n$ to $(x_1, \ldots, x_{i-1}, x_i \oplus h(x_{j_1},
x_{j_2}), x_{i+1}\ldots, x_n)$.
  
  We establish the following in
\Cref{sec:ini-gap,sec:large-m,sec:small-m}.

\begin{theorem}[\Cref{thm:XXX} restated] \label{thm:ini-gap-sou}
  For $\Ggp{n} \in \{\SOgp{2^n}, \SUgp{2^n}\}$, and any $k \geq 1$, there is a (multi)set $P_G$ of cardinality
  $O(n^2)$ such that
  \begin{equation}
    \opnorm{\E_{\bg \sim P_G}[\rho_{2^n}^{k,k}(\bg)] - \E_{\bg \sim \Ggp{n}}[\rho_{2^n}^{k,k}(\bg)]} \leq 1 - \frac{1}{n \cdot \poly(k)}.
  \end{equation}
\end{theorem}

The sets $P_G$ for $\SOgp{2^n}$ and $\SUgp{2^n}$ are described in \Cref{sec:gateset}.

Our proof of \Cref{thm:ini-gap-sou} is itself a general framework that
could potentially be used to obtain similar results for other
subgroups of the unitary group (for example, the sympletic group),
even though we only carry out the calculations for $\SOgp{2^n}$ and
$\SUgp{2^n}$. 

\subsection{Explicit $k$-wise independent permutations, orthogonal designs, and unitary designs}

We can finally apply \Cref{thm:abbrev-derando-walks}, so let's write
our above results in the notation of this theorem. Fix $k \geq 1$ and
consider any of the $P$ (multi)sets described in
\Cref{thm:ini-gap-perm,thm:ini-gap-sou}.

Let ${\cal U} = (\rho(g) - \E_{\bg \sim \Ggp{n}}[\rho(\bg)] : g \in
P)$ be a sequence of unitaries, where $\rho$ is the appropriate
unitary representation ($\calW^k_{2^n}$ for $S_{2^n}$ and $\rho^{k,
  k}_{2^n}$ for $\SOgp{2^n}$ and $\SUgp{2^n}$). For this choice of
${\cal U}$ we have $c = |P| = \poly(n)$. Notice that $\opnorm{\avg_{i
    \in [c]}\{U_i\}}$ is exactly the left hand side of the equations
in \Cref{thm:ini-gap-perm,thm:ini-gap-sou}, so we know that this
average is at most $1 - \delta$ for $\delta = 1 / \poly(n, k)$ (as
observed, this is actually $1 /\wt{O}(k^2 n^2)$ for $S_{2^n}$ and $1 /
(n \poly(k))$ for $\SOgp{2^n}$ and $\SUgp{2^n}$). Given $\eps > 0$,
applying \Cref{thm:abbrev-derando-walks} (with its ``$\eps$'' parameter set to $\eps/2^{nk}$) we obtain a sequence ${\cal
  P}$ of cardinality $\poly(2^{nk} / \eps)$ that satisfies
$\opnorm{\avg_{U \in {\cal P}}\{U\}} \leq \eps$. Additionally, $U \in
{\cal P}$ is a product of at most $\poly(nk)\log(1/\eps)$ elements of
${\cal U}$, and so it can be written as $\rho(g) - \E_{\bg \sim
  \Ggp{n}}[\rho(\bg)]$, where $g$ is a product of at most
$\poly(nk)\log(1/\eps)$ elements of $P$. This follows since for any
$g, g' \in P$,

\begin{equation}
  \parens*{\rho(g) - \E_{\bg \sim \Ggp{n}}[\rho(\bg)]} \parens*{\rho(g') - \E_{\bg \sim \Ggp{n}}[\rho(\bg)]} = \parens*{\rho(g \cdot g') - \E_{\bg \sim \Ggp{n}}[\rho(\bg)]},
\end{equation}

\noindent where we use the fact that $\E_{\bg \sim
  \Ggp{n}}[\rho(\bg)]$ is an orthogonal projection operator, and that
$\rho$ is a representation. We combine all of this in the following
theorems:

\begin{theorem} \label{thm:main-perm}
  Let $\eps > 0$. Then for any $k = k(n)$, there is a set ${\cal P}_{S_{2^n}}$
  that satisfies:

  \begin{equation}
    \opnorm{\E_{\bg \sim {\cal P}_{S_{2^n}}}[\calW^k(\bg)] - \E_{\bg \sim {S_{2^n}}}[\calW^k(\bg)]} \leq \eps/2^{nk}.
  \end{equation}

  Additionally, this set satisfies the following properties:

  \begin{itemize}
  \item Its cardinality is $\poly(2^{nk} / \eps)$.
  \item Each element of ${\cal P}_{S_{2^n}}$ is given by an $n$-qubit
    circuit consisting of $S = \poly(nk) \log(1/\eps)$ gates, which
    are elements of $P_G$.
  \item The algorithm that takes as input a seed and outputs the
    associated circuit runs in deterministic $\poly(S)$ time.
  \end{itemize}
\end{theorem}

\begin{theorem} \label{thm:main-sou}
  Let $\Ggp{n} \in \{\SOgp{2^n}, \SUgp{2^n}\}$ and $\eps > 0$. Then
  for any $k = k(n)$, there is a set ${\cal P}_G$ that satisfies:

  \begin{equation}
    \opnorm{\E_{\bg \sim {\cal P}_G}[\rho_{2^n}^{k,k}(\bg)] - \E_{\bg \sim \Ggp{n}}[\rho_{2^n}^{k,k}(\bg)]} \leq \eps/2^{nk}.
  \end{equation}

  Additionally, this set satisfies the following properties:

  \begin{itemize}
  \item Its cardinality is $\poly(2^{nk} / \eps)$.
  \item Each element of ${\cal P}_G$ is given by an $n$-qubit circuit
    consisting of $S = \poly(nk) \log(1/\eps)$ gates, which are
    elements of $P_G$.
  \item The algorithm that takes as input a seed and outputs the
    associated circuit runs in deterministic $\poly(S)$ time.
  \end{itemize}
\end{theorem}

Ultimately we want designs for $\Ogp{2^n}$ and $\Ugp{2^n}$; we 
obtain them from the above via the following simple corollary.

\begin{corollary}[\Cref{thm:main} restated]
  Let $\Ggp{n} \in \{\Ogp{2^n}, \Ugp{2^n}\}$ and $\eps > 0$. Then for
  any $k = k(n)$, there is a set ${\cal P}_G$ that satisfies the
  conditions of \Cref{thm:main-sou}.
\end{corollary}
\begin{proof}
  To obtain the result for $\Ogp{2^n}$, after sampling from ${\cal
    P}_\mathrm{SO}$ one samples $\bb$ as a uniformly random $\pm 1$ and
  multiplies the first column of the sampled matrix by $\bb$ (which only
  changes the cardinality of the resulting ${\cal P}_\mathrm{O}$ by a
  factor of $2$).  
  
For $\Ugp{2^n}$,  no augmentation of ${\cal P}_\mathrm{SU}$ is required; i.e.,~we can simply take 
${\cal P}_\mathrm{U}={\cal P}_\mathrm{SU}$.\footnote{We thank an anonymous reviewer for this observation.}  
To see this, first recall that the representations $\rho^{k,k}_{2^n}$ of $\Ugp{2^n}$ and of $\SUgp{2^n}$ respectively have kernels $K_\mathrm{U} = \{e^{i \alpha} \Id \ : \ \alpha \in [0,2\pi)\}$ and $K_\mathrm{SU} \{-\Id,\Id\}$.  
Next, recall that $\Ugp{2^n}/K_1 \cong \SUgp{2^n}/K_2$, since both of these are isomorphic to the projective unitary group $\mathrm{PU}(2^n)$ (see e.g. \cite{Wiki:projectiveunitarygroup}).  
Now, since the push-forward of the Haar measure of a compact group $G$ to a factor group $G/H$ is exactly the Haar measure on $G/H$, it follows that the set ${\cal P}_\mathrm{SU}$ of unitaries in $\SUgp{2^n}$ that forms an $\epsilon$-approximate $k$-design of $\SUgp{2^n}$ is also an $\epsilon$-approximate $k$-design of $\Ugp{2^n}$.
\end{proof}



It is of note that we can apply the framework of this section, through
\Cref{thm:abbrev-derando-walks}, to obtain explicit designs of any
subgroup of the unitary group using any unitary representation, as
long as one establishes an initial gap first, like in
\Cref{ssec:ini-gaps}.

\ignore{
\subsection{old}

\emph{Maybe let's have a section here where we put the pieces together, i.e. we}

\begin{enumerate}

\item \emph{invoke the initial-spectral-gap results of Gowers / Brodsky-Hoory and combine them with \Cref{thm:derando-walks} to get $k$-wise independent permutations }

\item \emph{state the initial-spectral-gap theorems for unitary and special orthogonal that we will establish in the rest of the paper, and combine those results with \Cref{thm:derando-walks} to get our beloved explicit $k$-wise unitary designs and $k$-wise special orthogonal designs. (This might also be a good place to handle the minor task of going from special orthogonal designs to non-special orthogonal designs.)}

\end{enumerate}

\emph{Even though (1) will be brief, I feel like actually doing (1) in the intro may require a little too much forward pointering for it to fit nicely into an introduction - we can give a high level overview in the intro, but my feeling is that it'll be easier to actually do it after \Cref{sec:pseudorandom-walks}.} 

\emph{Regarding (2), hopefully this approach will let us make things relatively clean and modular. We can make portentious sounding statements in this section about how our framework can potentially be generalized - proving analogues of the initial-spectral-gap theorems for other groups (symplectic!  octahedrawhatever!) should lead to $k$-wise independence for those settings - but in this work we confine our attention to special orthogonal and unitary.}

\emph{If we go with an organization like this, below is some stuff that I think should be adapted/included in this section:}

\medskip

Let $\Ggp{2^n}$ be a subgroup of $\Ugp{2^n}$ (key examples to keep in mind are the group of permutations on $2^n$ elements, the $2^n$-dimensional orthogonal group, the $2^n$-dimensional unitary group itself, and the ``special'' versions of the latter two).
In light of  \Cref{thm:derando-walks}, given a probability distribution on a subset of $\Ggp{2^n}$, 
we would like to understand how fast the associated random walk mixes vis-a-vis a particular representation, namely the $k$-wise tensor product representation (since that representation corresponds to $k$-wise independence). We will specialize to the $k$-wise tensor product representation soon, but for now we consider a general unitary representation:
\begin{definition} \label{def:LW}
    Let $\calP$ be a probability distribution\rnote{perhaps have to say some blah blah measure-theoretic whatnot} on 
    $\Ggp{2^n}$ 
    that is symmetric (meaning that $\bg^{-1} = \bg^\dagger$ is distributed as~$\calP$ when $\bg$ is), and let~$\rho$ be a unitary representation of
$\Ggp{2^n}$.
    Note that since $\rho$ is unitary, $\E_{\bg \sim {\cal P}}[\rho(\bg)]$ is a Hermitian operator with real eigenvalues lying in $[-1,1]$.
    We define the ``Laplacian''
    \begin{equation}
        L_{\calP}(\rho) = \Id - \E_{\bg \sim \calP}[\rho(\bg)],
    \end{equation}
    a self-adjoint (since~$\calP$ is symmetric) operator satisfying  the following inequalities (in the PSD order):
    \begin{equation}    \label[ineq]{ineq:2}
        0 \leq L_{\calP}(\rho) \leq 2 \cdot \Id.
    \end{equation}
\end{definition}
\begin{notation}    \label{not:P}
    In the preceding definition, we abuse notation as follows:
    In place of~$\calP$ we may write a finite 
    (multi)set~$P \subset \Ggp{2^n}$, 
    in which case the uniform distribution on~$P$ is understood.
    We may also write 
    ``$\Ggp{2^n}$'' 
    in place of~$\calP$, in which case the uniform (Haar) distribution is understood.
    Finally, if $\calP$ now denotes a distribution 
    on~$\Ggp{2^\ell}$, 
    and $E \subseteq [n]_\ell$, we write $\calP \times E$ for the distribution 
    on~$\Ggp{2^n}$ 
    given by choosing $\bg \sim \calP$, independently choosing $\be \sim E$ (uniformly), and finally forming~$\bg_{\be}$.
\end{notation}

\begin{definition} \label{def:twiddle}
Given a symmetric probability distribution $\calP$ as in \Cref{def:LW}, we define its ``lazy'' version, $\widetilde{\cal P}$, to be the distribution which is an equal mixture of ${\cal P}$ and the point distribution supported on the identity element $\Id$ (note that $\widetilde{\cal P}$ is also a symmetric distribution).
Similar to \Cref{def:LW}, we have that $\E_{\bg \sim \widetilde{\cal P}}[\rho(\bg)]$ is a Hermitian operator but now with real eigenvalues lying in $[0,1]$, and we have the PSD inequalities
    \begin{equation}    \label[ineq]{ineq:2twiddle}
        0 \leq L_{\widetilde{\calP}}(\rho) \leq \Id.
    \end{equation}

\end{definition}

\begin{fact}
    In the setting of \Cref{def:LW,def:twiddle}, 
    $L_{\Ggp{2^n}}(\rho)$ 
    is an orthogonal projection operator, and for any symmetric $\calP$ we have that
    \begin{equation}
        \ker L_{\Ggp{2^n}}(\rho) \subseteq \ker L_{\calP}(\rho)
    \end{equation}
    always holds (because for \emph{every}~$g_0$ in the support of~$\calP$ we have $\rho(g_0) \Pi = \Pi$, where 
    $\Pi = \E_{\bg \sim  \Ggp{2^n}}[\rho(\bg)]$).
        From this, and \Cref{ineq:2,ineq:2twiddle}, we also get
    \begin{equation}    \label[ineq]{ineq:fantasy}
        L_{\Ggp{2^n}}(\rho) \geq \tfrac12 \cdot L_{\calP}(\rho),
\end{equation}
\begin{equation} \label[ineq]{ineq:fantasytwiddle}
	L_{\Ggp{2^n}}(\rho) \geq L_{\widetilde{\calP}}(\rho).
    \end{equation}
\end{fact}

As \Cref{ineq:fantasy,ineq:fantasytwiddle} contain a surfeit of symbols, one may wish to read them respectively as
\begin{equation}    \label[ineq]{ineq:rando}
    \text{``(randomizing $n$ qubits)} \geq \tfrac12 \cdot\text{($\calP$-pseudorandomizing $n$ qubits)} \quad \text{[vis-a-vis~$\rho$]''},
\end{equation}
\begin{equation} \label[ineq]{ineq:randotwiddle}
\text{``(randomizing $n$ qubits)} \geq \text{($\widetilde{\calP}$-pseudorandomizing $n$ qubits)} \quad \text{[vis-a-vis~$\rho$]''},
\end{equation}
with the ``''$\geq \frac12 \cdot$'' part pronounced ``is at least $\frac12$  as good as''. 

Since our goal is $k$-wise independence, the representations that are of interest to us are $k$-wise tensor product representations:
\begin{notation} [$k$-wise tensor product representations]\rnote{Do we want to have consistency on whether there are parentheticals after notations and definitions?}
    For any $k \in \N^+$, we will write $\rho^{k,k}_{2^n}$ for the  (complex) representation of 
    $\Ggp{2^n}$ 
    defined by
    \begin{equation}
        \rho^{k,k}_{2^n}(g) = g^{\otimes k,k} \coloneqq g^{\otimes k} \otimes \overline{g}^{\otimes k},
    \end{equation}
    where $\overline{g}$ denotes the complex conjugation of matrix~$g$.
\end{notation}

\subsection{Explicit randomness-efficient $k$-wise independent permutations over $\zo^n$} \label{sec:perms}

To get the desired initial spectral gap for $\Ggp{2^n}$, we require inequalities that go in the opposite direction to \Cref{ineq:rando} and \Cref{ineq:randotwiddle}, with $\rho=\rho^{k,k}_{2^n}$ being the $k$-wise tensor product representation.

etc etc\rnote{to be written?}

\subsection{Explicit randomness-efficient orthogonal and unitary $k$-designs} \label{sec:orthogonal-unitary}

\rnote{I don't think this is the definition of ``design'' that we want}

\begin{definition} \label{def:orthogonal-design}
Let $\Ggp{2^n}$ be a subgroup of $\Ugp{2^n}$ that is equipped with the
Haar measure.  A distribution ${\cal D}$ on a finite subset of
matrices from $\Ugp{2^n}$ is an \emph{$\eps$-approximate $k$-design
  for $\Ggp{2^n}$} if for every polynomial $p: \C^{2^n \times 2^n} \to
\R$ with complex coefficients satisfying $\deg(p) \leq k$ and $\|p\|_1
= 1$ (where $\|p\|_1$ denotes sum of norms of coefficients), we have

\[
\left| \E_{\bM \sim \cal D}[p(\bM)] - \E_{\bM \sim \Ggp{2^n}}[p(\bM)]\right| \leq \eps,
\]

We say ${\cal D}$ is an \emph{explicit $\eps$-approximate $k$-design
  for $\Ggp{2^n}$ with seed-length $s$} if there is a
$\poly(2^n)$-time randomized procedure that generates a draw from
${\cal D}$, using $s$ bits of randomness.
\end{definition}

Using some natural known observations we can simplify the above into
the following:

\begin{lemma}
Suppose ${\cal D}$ is a distribution that satisfies the following:

\begin{equation} \label{eq:tensor-design}
\left| \E_{\bM \sim \cal D}[\rho_{2^n}^{k,k}(\bM)] - \E_{\bM \sim \Ggp{2^n}}[\rho_{2^n}^{k,k}(\bM)]\right| \leq \eps,
\end{equation}

Then ${\cal D}$ is an explicit $\eps$-approximate $k$-design for
$\Ggp{2^n}$.
\end{lemma}
\begin{proof}
  Coming soon...
\end{proof}

We will build these distributions for the orthogonal and unitary
groups by constructing multisets of matrices such that their
cardinality is small and the distribution given by sampling an element
uniformly at random satisfies the condition from
\Cref{eq:tensor-design}. We build these distributions in two steps:
first, establish a really big but non-trivial bound of the form of
\Cref{eq:tensor-design}, with really small seed-length; then, use
\Cref{thm:derando-walks} to reduce this bound without making the
seed-length much bigger. In \Cref{sec:ini-gap} we will carry out the
first step for the orthogonal and unitary groups, but we summarize the
result of the section here:

\begin{theorem} (Consequence of \Cref{thm:XXX})
For the orthogonal group ($\Ggp{2^n} = \Ogp{2^n}$) and the unitary
group ($\Ggp{2^n} = \Ugp{2^n})$, there exist multisets
$\calM_\mathrm{O}$ and $\calM_\mathrm{U}$ such that

\begin{equation}
\left| \E_{\bM \sim \calM_\mathrm{G}}[\rho_{2^n}^{k,k}(\bM)] - \E_{\bM \sim \Ggp{2^n}}[\rho_{2^n}^{k,k}(\bM)]\right| \leq 1 - \frac{1}{\poly(n, k)}, 
\end{equation}

\noindent and $|\calM_\mathrm{G}| \leq \poly(n)$ (so the seed-length is $O(\log n)$).
\end{theorem}

We finally can show the main result of this paper:

\begin{theorem}
For the orthogonal group ($\Ggp{2^n} = \Ogp{2^n}$) and the unitary
group ($\Ggp{2^n} = \Ugp{2^n})$, there exist multisets
$\calM_\mathrm{O}$ and $\calM_\mathrm{U}$ such that
\begin{itemize}
\item $\calM_\mathrm{G}$ is an explicit $\eps$-approximate $k$-design
  for $\Ggp{2^n}$.
\item $|\calM_\mathrm{G}| \leq \poly(n^k / \epsilon)$, and so the
  seed-length is $O(nk + \log(1 / \eps))$.
\end{itemize}
\end{theorem}
\begin{proof}
  
\end{proof}
} 


\section{Establishing an initial spectral gap for special orthogonal and unitary groups}
\label{sec:ini-gap}

In the rest of the paper we consider a sequence of groups
$(\Ggp{n})_{n \geq 1}$ which is either $(\SOgp{2^n})_{n \geq 1}$ or $(\SUgp{2^n})_{n \geq 1}$.
We recall (see e.g. \cite[Section~1.3]{Mec19}) that these groups have associated Lie algebras $\frak{g}_n$, where
\begin{align}
 \text{for~} \Ggp{n} = \SOgp{2^n},\ \frak{g}_n &= \{H \in \R^{2^n \times 2^n} : H \text{ skew-symmetric}\}, \label{eq:algySO} \\
 \text{for~}  \Ggp{n}= \SUgp{2^n},\ \frak{g}_n &= \{H \in \C^{2^n \times 2^n} : H \text{ skew-Hermitian},  \tr(H)=0\}. \label{eq:algyU}
\end{align}
When we need to specialize our discussion to a particular one of these two cases, we will do so explicitly; most of our arguments go through for both settings (and many go through for the more general setting in which $\Ggp{n}$ is any compact connected Lie group).

As discussed in \Cref{sec:BHH-extraction-stuff}, given \Cref{thm:derando-walks}, in order to construct an explicit  $k$-design for $\Ggp{n}$ it suffices to construct an explicit sequence  $\calU = (U_1, \dots, U_{\c})$ of $2^n \times 2^n$ matrices from $\Ggp{n}$ satisfying $\opnorm{\rho^{k,k}_{2^n}(U_i)} \leq 1$ for all~$i$ and $\opnorm{\avg(\rho^{k,k}_{2^n}(\calU))} \leq 1-{\frac 1 {n \cdot \poly(k)}}$ (in fact, a spectral gap of $\poly(n,k)$ would also be sufficient). Constructing such a sequence for $\Ggp{n}$ as described above is the main goal of this section and is accomplished in the following theorem:

\begin{theorem} 
\label{thm:XXX}
Let $(\Ggp{n})_{n \geq 1} \in \{(\SOgp{2^n})_{n \geq 1},(\SUgp{2^n})_{n \geq 1}\}$. There is a fixed positive integer $n_0=4$\ignore{$n_0 \geq 2$ \rasnote{check: is $n_0$ just  $=3$ in both cases?}} and a finite multiset $P_{n_0} \subset \Ggp{n_0}$ such that
for all sufficiently large $n$ and all $k \geq 1$, we have
 \begin{equation}    \label[ineq]{ineq:well-mixing-single-pseudo-step}
        \forall k \in \N^+, \quad L_{\widetilde{P}_{n_0} \times \binom{[n]}{n_0}}(\rho^{k,k}_{2^n}) \geq {\frac 1 {n \cdot \poly(k)}} \cdot
        L_{\Ggp{n}}(\rho^{k,k}_{2^n}).
    \end{equation}
\end{theorem}
(We note that even without using the ``pseudorandom walks'' machinery of \Cref{sec:pseudorandom-walks}, as discussed in \Cref{sec:our-framework}, since \Cref{thm:XXX} establishes an initial spectral gap of $1 - {\frac 1 {n \cdot \poly(k)}}$, simply taking a product of $n \cdot \poly(k) \cdot \log(2^{nk}/\eps)$ uniform random draws from $\widetilde{P}_{n_0}$ would yield an $\eps$-approximate $k$-design for $\Ggp{n}$ with seed length $\poly(n,k) \cdot \log(1/\eps)$.  By combining \Cref{thm:XXX} with \Cref{thm:derando-walks} (i.e.~using pseudorandom walks) we are able to improve this to seed length $O(nk + \log(1/\eps))$, thus matching the random construction.)

\subsection{Overview of the proof of \Cref{thm:XXX}} \label{sec:overview}

Our proof of \Cref{thm:XXX} refines and extends an approach from \cite{HH21}, and combines it with arguments from \cite{BHH16}. In this subsection we give a high-level overview of the structure of the proof, and in the next subsection we give (a modular version of) the actual proof. Establishing the various modular pieces will comprise the rest of the paper after \Cref{sec:proof-of-XXX}.

In Theorem~4 of \cite{HH21}, Haferkamp and Hunter-Jones establish a spectral gap for non-local random quantum circuits with truly (Haar) random two-qudit unitary gates over the unitary group. 
This is done by analyzing Haar random unitary gates over $m-1$ randomly chosen qubits from an $m$-qubit system; this enables them to establish a recurrence relation which lets them bound the spectral gap of $k$-qudit Haar random unitary gates in terms of the spectral gap of $(k+1)$-qudit Haar random unitary gates.  
Our \Cref{lem:iter} below is a generalization and rephrasing of their recurrence relation for the special\footnote{At the end of this subsection we explain why, even though our ultimate goal is to obtain results for the orthogonal and unitary groups, we need to work with the special versions of these groups at this point in the argument.} unitary and special orthogonal groups; it essentially says that if truly randomizing (a randomly chosen) $m-1$ out of $m$ qubits is ``not too much worse'' than truly randomizing all $m$ qubits, then truly randomizing only a constant number of (randomly chosen) qubits out of $m$ qubits is also not too much worse than truly randomizing all $m$ qubits. 
Given this, the remaining tasks are (1) to show that indeed truly randomizing (a randomly chosen) $m-1$ out of $m$ qubits is ``not too much worse'' than truly randomizing all $m$ qubits; and (2) to show that at the bottom level of the argument, it suffices to \emph{pseudorandomize} a constant number of (randomly chosen) qubits out of $m$ qubits.

Task (1) requires a significant amount of technical work and is the subject of \Cref{sec:large-m} and \Cref{sec:small-m}. We follow the high-level approach of \cite{HH21} by breaking the analysis into two sub-cases (\Cref{thm:large-m} and \Cref{thm:small-m}) depending on the relative sizes of $k$ and $m$. In each of these sub-cases we adapt and generalize the analysis of \cite{HH21} (we note that the ``small-$m$'' case of \cite{HH21}, for the unitary group, was based in turn on \cite{BHH16}) in a way which permits a unified treatment of both the special orthogonal group and the special unitary group.

Task (2) is necessary because our ultimate goal statement, \Cref{thm:XXX}, requires the randomly chosen non-local gates to be drawn from a \emph{finite} ensemble of gates rather than being Haar random (``truly random'') gates over $n_0$ qubits. 
For this step (made formal in \Cref{lem:statement-G-new}), following \cite{BHH16} we use a deep result of Bourgain and Gamburd (subsequently generalized by Benoiste and de~Saxc\'{e}~\cite{BS16}) to pass from the Haar distribution over $\Ggp{n_0}$ to a uniform distribution over an explicit finite ensemble of $n_0$-qubit gates; see \Cref{sec:Pn0}.
The \cite{BS16} results require that the Lie groups in question be compact and simple; this requirement is why we need to work with the special versions of the unitary and orthogonal groups (indeed, in the special orthogonal case we need to further pass to the projective special orthogonal group; see the proof of \Cref{cor:bs}).

%
%
%
%
%

\subsection{Proof of \Cref{thm:XXX}} \label{sec:proof-of-XXX}

In order to establish the lower bound of \Cref{ineq:well-mixing-single-pseudo-step} we will need to chain together some statements that go in the opposite direction from \Cref{ineq:rando} and \Cref{ineq:randotwiddle}.  We do this via the following lemma, which we prove in \Cref{sec:proof-of-lem-iter}.

\begin{lemma} \label{lem:iter}
    Fix a positive integer constant $n_0 \geq 4$.  Suppose that for $n_0 < m \leq n$ we have 
    \begin{equation}    \label[ineq]{ineq:mess}
        \forall k \in \N^+, \quad L_{\Ggp{m-1} \times \binom{[m]}{m-1}}(\rho^{k,k}_{2^m}) \geq \tau_{k,m} \cdot L_{\Ggp{m}}(\rho^{k,k}_{2^m}).
    \end{equation}
    Then
    \begin{equation}    \label[ineq]{ineq:concl}
 \forall k \in \N^+, \quad L_{\Ggp{n_0} \times \binom{[n]}{n_0}}(\rho^{k,k}_{2^n}) \geq \parens*{\prod_{n_0 < m \leq n} \tau_{k,{m}}}
        \cdot L_{\Ggp{n}}(\rho^{k,k}_{2^n}).
    \end{equation}
\end{lemma}

We remark that \Cref{lem:iter} only deals with ``truly" (Haar) random gates; later we will move from $n_0$-arity ``truly random'' gates to ``pseudorandom'' gates, which are drawn uniformly at random from a finite multiset.  It may be helpful to think of the lemma's conclusion (\Cref{ineq:concl}) as intuitively saying that truly randomizing only constantly many (randomly chosen) qubits is ``not too much worse'' than truly randomizing all $n$ qubits, vis-a-vis the $k$-wise tensor product representation.

With the above lemma in hand, proving \Cref{thm:XXX} breaks down naturally into two steps.

\medskip

\noindent {\bf First step: Passing from truly random $m$-qubit gates to truly random $(m-1)$-qubit gates.}  In other words, lower-bounding $\tau_{k,m}$ for $m=n_0+1,\dots,n$.  This is the main technical task where the bulk of our work is required.  The analysis is done separately for  ``large $m$'' and ``small $m$'' cases, similar to Lemmas~6 and~7 of \cite{HH21}, respectively.  

 \Cref{sec:large-m} lower bounds $\tau_{k,m}$ for ``large $m$'':

\begin{theorem} \label{thm:large-m}
For all   $k \leq \frac{1}{\sqrt{10}m^2} 2^{m/2}$ we have that \Cref{ineq:mess} holds with $\tau_{k,m} \geq 1 - {\frac 1 m} - {\frac {\sqrt{10} km}{2^{m/2}}}$.
\end{theorem}

 \Cref{sec:small-m} gives a  lower bound on $\tau_{k,m}$ which will be useful for ``small $m$'':

\begin{theorem} \label{thm:small-m}
For all $m \geq 4$ and all $k \in \N^+$, we have that \Cref{ineq:mess} holds with 
$\tau_{k,m} \geq .04.$
%
\end{theorem}
(We note that \Cref{thm:small-m}'s requirement that $m \geq 4$ is why we take $n_0=4$ in \Cref{thm:XXX}.)

Given \Cref{thm:large-m} and \Cref{thm:small-m}, we get the desired lower bound on $\tau_{k,n_0+1} \cdots \tau_{k,n}$ from a routine computation:

\begin{lemma}  \label{lem:technical}
For any constant $n_0 \geq 4$, for all $n$ and all $k \in \N^+$
we have 
$\tau_{k,n_0+1} \cdots \tau_{k,n} \geq {\frac 1 {n \cdot \poly(k)}}.$


\end{lemma}

\begin{proof}
Fix $n_0 \geq 4$ and take any $n,k \geq 1$.
Defining $\ell = \lfloor 4 \log_2(60 k)\rfloor \geq 20$,  by \Cref{thm:small-m} we have
\begin{equation} \label{eq:mama}
\tau_{k,n_0+1} \cdots \tau_{k,\ell} \geq
 (.04)^{\ell}  = (.04)^{O(\log k)} \geq \frac{1}{\poly(k)}.
\end{equation}
This proves the result if $n \leq \ell$.  Otherwise, it remains to show that 
\begin{equation} \label[ineq]{ineq:mama}
    \tau_{k,\ell+1} \cdots \tau_{k,n} \geq 1/n.
\end{equation}
For $m \geq \ell+1$ we have $k \leq \frac{1}{60} 2^{m/4} \leq \frac{1}{\sqrt{10} m^2}2^{m/2}$, so we are eligible to use the bound from \Cref{thm:large-m}.  Then using 
\begin{equation}
    \frac{\sqrt{10} k m}{2^{m/2}} \leq \frac{\sqrt{10}  m}{60 \cdot 2^{m/4}} \leq 2^{-m/5}, \quad 1-\frac1m - 2^{-m/5} \geq \parens*{1-\frac1m}\exp(-2^{1-m/5})
\end{equation}
(the last inequality using $m \geq \ell \geq 20$),
we conclude
\begin{equation}
    \tau_{k,\ell+1} \cdots \tau_{k,n} \geq \prod_{m=\ell+1}^n \parens*{1 - \frac1m} \exp(- 2^{1-m/5}) = \frac{\ell}{n} \exp\parens*{-\sum_{m=\ell+1}^n 2^{1-m/5}} \geq \frac{1}{n}
\end{equation}    
(using $\ell \geq 20$), confirming \Cref{ineq:mama}.
\end{proof}

\medskip

\noindent {\bf Second step: From ``truly random'' non-local $n_0$-qubit gates to ``pseudorandom'' non-local $n_0$-qubit gates.}
The next lemma, proved in \Cref{sec:Pn0}, may be viewed as saying that (suitably) \emph{pseudo}-randomizing constantly many randomly chosen qubits is ``not much worse'' than \emph{truly} randomizing those qubits.  

\begin{lemma} \label{lem:statement-G-new}
There is an absolute constant
$n_0=4$ such that
for $n \geq n_0+1$, we have
$$
\forall k \in \N^+, \quad L_{\widetilde{P}_{n_0}\times \binom{[n]}{n_0}}(\rho^{k,k}_{2^n}) \geq \kappa_{n_0} \cdot 
L_{\Ggp{n_0} \times \binom{[n]}{n_0}}(\rho^{k,k}_{2^n}),
$$ where $\kappa_{n_0}$ is an absolute constant (depending only on $n_0$).
\end{lemma}

 \Cref{thm:XXX} follows from \Cref{lem:iter}, \Cref{lem:technical} and \Cref{lem:statement-G-new}.

\subsection{Proof of \Cref{lem:iter}} \label{sec:proof-of-lem-iter}

\begin{lemma} [Restatement of \Cref{lem:iter}]
    Fix a positive integer $n_0 \geq 4$.  Suppose that for $n_0 < m \leq n$ we have 
    \begin{equation}    \label[ineq]{ineq:mess-restated}
        \forall k \in \N^+, \quad L_{\Ggp{m-1} \times \binom{[m]}{m-1}}(\rho^{k,k}_{2^m}) \geq \tau_{k,m} \cdot L_{\Ggp{m}}(\rho^{k,k}_{2^m}).
    \end{equation}
    Then
    \begin{equation}    \label[ineq]{ineq:concl-restated}
 \forall k \in \N^+, \quad L_{\Ggp{n_0} \times \binom{[n]}{n_0}}(\rho^{k,k}_{2^n}) \geq \parens*{\prod_{n_0 < m \leq n} \tau_{k,{m}}}
        \cdot L_{\Ggp{n}}(\rho^{k,k}_{2^n}).
    \end{equation}
\end{lemma}

\begin{proof}
For readability we simply write $\tau_i$ in this proof to stand for $\tau_{k,i}$.
Also for readability we express the lemma as 
    \begin{align}
         \textnormal{(randomizing $m-1$ out of $m$ qubits)} &\geq \tau_{m} \cdot \textnormal{(randomizing all $m$ qubits)}\quad \forall \ n_0 < m \leq n \label[ineq]{ineq:cazh-n0}\\
        \implies  \quad
        \textnormal{(randomizing $n_0$ out of $n$ qubits)} &\geq \tau_{n_0+1} \cdots \tau_{n} \cdot \textnormal{(randomizing all $n$ qubits),} \nonumber
    \end{align}
    with the modifier ``vis-as-vis all $\rho^{k,k}_{2^m}$'' being implied.
    The $m = n_0+1$ case of \Cref{ineq:cazh-n0} is
    \begin{equation}   \label[ineq]{ineq:f1}
        \textnormal{(randomizing $n_0$ out of $n_0+1$ qubits)} \geq \tau_{n_0+1} \cdot \textnormal{(randomizing all $n_0+1$ qubits)}.
    \end{equation}
    From this, by adding an ignored $(n_0+2)$th qubit, we are able to conclude
    \begin{multline} \label[ineq]{ineq:f2}
        \textnormal{(randomizing $n_0$ out of the first $n_0+1$ of $n_0+2$ qubits)} \\
        \geq \tau_{n_0+1} \cdot \textnormal{(randomizing the first $n_0+1$ of  $n_0+2$ qubits)}.
    \end{multline}
    To derive this implication more formally, start with \Cref{ineq:f1}, which says that for all $k \in \N^+$,
    \begin{equation}    \label[ineq]{ineq:16}
        \E_{\substack{\bg \sim \Ggp{n_0} \\ \be \sim \binom{[n_0+1]}{n_0}}}[\Id - 
        \bg_{\be}^{\otimes k,k}
        ]
        \geq \tau_{n_0+1} \cdot 
        \E_{\bh \sim \Ggp{n_0+1}}[\Id - \bh^{\otimes k,k}
        ].
    \end{equation}
    We now consider tacking on a $(n_0+2)$th tensor factor that is ignored by both~$\bg_e$ and by~$\bh$.
    Since $A \geq B \implies A \otimes \Id \geq B \otimes \Id$,  we can tensor-product both sides of \Cref{ineq:16} by $\Id^{\otimes k,k}$ (where $\Id$ denotes the $2\times2$ identity matrix) to conclude
    \begin{equation}    \label[ineq]{ineq:17}
        \E_{\substack{\bg \sim 
        \Ggp{n_0} 
        \\ \be \sim \binom{[n_0+1]}{n_0} \in [n_0+2]_{n_0}}}[\Id - 
        \bg_{\be}^{\otimes k,k}
        ]
        \geq \tau_{n_0+1} \cdot 
        \E_{\substack{\bh \sim 
        \Ggp{n_0+1}\\ 
        f \coloneqq [n_0+1] \in [n_0+2]_{n_0+1}}}[\Id - \bh_f^{\otimes k,k} 
        ],
    \end{equation}
    and this is the meaning of \Cref{ineq:f2}.
    Indeed, we can insert the ignored $(n_0+2)$th qubit at any position, not just the last one; i.e., for any $j \in [n_0+2]$,
    \begin{equation}    \label[ineq]{ineq:18}
        \E_{\substack{\bg \sim 
        \Ggp{n_0} \\ 
        \be \sim \binom{[n_0+2]\setminus j}{n_0}}}[\Id -  \bg_{\be}^{\otimes k,k}
        ]
        \geq \tau_{n_0+1} \cdot 
        \E_{\substack{\bh \sim 
        \Ggp{n_0+1}\\ 
        f \coloneqq [n_0+2] \setminus j}}[\Id - \bh_{f}^{\otimes k,k}. 
        ]
    \end{equation}
        If we now average the above (PSD-order) inequality over $\bj \sim [n_0+2]$ we get
    \begin{equation}    \label[ineq]{ineq:19}
        \E_{\substack{\bg \sim 
        \Ggp{n_0} \\
         \be \sim \binom{[n_0+2]}{n_0}}}[\Id -  \bg_{\be}^{\otimes k,k}
        ]
        \geq \tau_{n_0+1} \cdot 
        \E_{\substack{\bh \sim 
        \Ggp{n_0+1}\\ 
        \boldf \sim \binom{[n_0+2]}{n_0+1}}}[\Id - \bh_{\boldf}^{\otimes k,k} 
        ],
    \end{equation}
    which we would express as
    \begin{equation}\label[ineq]{ineq:f3}
        \textnormal{(randomizing $n_0$ out of $n_0+2$ qubits)} \\
        \geq \tau_{n_0+1} \cdot \textnormal{(randomizing $n_0+1$ out of $n_0+2$ qubits)}.        
    \end{equation}
    But the $m = n_0+2$ case of our hypothesis \Cref{ineq:cazh-n0} is
    \begin{equation}   \label[ineq]{ineq:f10}
        \textnormal{(randomizing $n_0+1$ out of $n_0+2$ qubits)} \geq \tau_{n_0+2} \cdot \textnormal{(randomizing all $n_0+2$ qubits)},
    \end{equation}
    so chaining this together with \Cref{ineq:f3} (using the PSD-ordering fact $A \geq B$, $B \geq C \implies A \geq C$) gives
    \begin{equation}\label[ineq]{ineq:ffff}
        \textnormal{(randomizing $n_0$ out of $n_0+2$ qubits)} \\
        \geq \tau_{n_0+1} \cdot \tau_{n_0+2} \cdot \textnormal{(randomizing all $n_0+2$ qubits)}.
    \end{equation}
    Iterating this argument completes the proof of the lemma.
\end{proof}

\subsection{Proof of \Cref{lem:statement-G-new}} \label{sec:Pn0}

An ingredient we need for \Cref{lem:statement-G-new} is the existence of a suitable finite ``gate set'' with useful properties. This is provided by the following lemma, which follows from known universality results in quantum computing (see \Cref{sec:gateset}):

\begin{lemma} \label{lem:gateset}
There is an absolute constant $n_0=4$ for which there is a finite multiset $P_{n_0} \subset \SOgp{2^{n_0}}$, closed under negations and inverses, with two properties:
    \begin{itemize}
        \item [(A)] (There is a basis in which) every matrix in~$P_{n_0}$ has algebraic entries.
        \item [(B)] Finite products of elements of~$P_{n_0}$ are dense in~$\SOgp{2^{n_0}}$.
\end{itemize}
The same statement is true for $\SUgp{2^{n_0}}$ (also with $n_0=4$).
\end{lemma}

\Cref{lem:gateset} allows us to use a deep result of Benoist and de~Saxc\'{e}~\cite{BS16}, which extended earlier work of Bourgain--Gamburd~\cite{BG12} (for the case of the special unitary group) to a broader range of groups.
The main result of \cite{BS16} is as follows:

\begin{theorem} \label{thm:bs}
    (\cite[Consequence of Theorem~1.2]{BS16}.) 
    For $n \geq 1$ let $\Ggp{n} \subseteq \SUgp{2^n}$ be a connected compact simple Lie group.
   Fix a positive integer $n_0$ and suppose that $P_{n_0} \subset \Ggp{n_0}$ satisfies properties (A) and (B) of \Cref{lem:gateset}.
   Then there exists a constant $\kappa > 0$ such that  
     \begin{equation}
         \opnorm{\E_{\bg \sim \widetilde{P_{n_0}}}[\mathrm{reg}(\bg)] - \E_{\bg \sim \Ggp{n_0}} [\mathrm{reg}(\bg)]} \leq 1 - \kappa, \label{eq:BdS}
     \end{equation}
     where $\mathrm{reg}$ denotes the regular representation of~$\Ggp{n_0}$.
     Equivalently, 
     $
         L_{\widetilde{P_{n_0}}}(\mathrm{reg}) \geq \kappa  \cdot L_{\Ggp{n_0}}(\mathrm{reg}),
    $
    or
    \begin{equation}    \label[ineq]{ineq:bdsrando}
        \textnormal{($\widetilde{P_{n_0}}$-pseudorandomizing in $2^{n_0}$ dimensions)} \geq \kappa 
        \cdot\textnormal{(randomizing $2^{n_0}$ dimensions)} \quad \textnormal{[vis-a-vis~$\mathrm{reg}$]}.
    \end{equation}    
\end{theorem}

We remark that (as noted by \cite{BHH16}) a weaker form of \Cref{eq:BdS}, with the $k$-wise tensor product representation in place of the regular representation and $\kappa$ depending on $k$, has been known at least since \cite{ArnoldKrylov62}; however, the stronger quantitative bound of \Cref{eq:BdS} is essential for our purposes.

\Cref{thm:bs} yields the following useful corollary:

\begin{corollary} \label{cor:bs}
For $n_0=4$, $\Ggp{n_0}=\SOgp{2^{n_0}}$, and $P_{n_0} \subset \Ggp{n_0}$ satisfying properties (A) and (B) of \Cref{lem:gateset}, 
    there is a constant $\kappa > 0$ such that for all $k \in \N^+$ we have     
     $
         L_{\widetilde{P_{n_0}}}(\mathrm{\rho}^{k,k}_{2^{n_0}}) \geq \kappa \cdot L_{\Ggp{n_0}}(\mathrm{\rho}^{k,k}_{2^{n_0}}).
    $
    That is, vis-a-vis any $\rho^{k,k}_{2^{n_0}}$, we have 
    \begin{equation}    \label[ineq]{ineq:bdsrand}
        \textnormal{($\widetilde{P_{n_0}}$-pseudorandomizing $n_0$ qubits)} \geq \kappa \cdot\textnormal{($\Ggp{n_0}$-randomizing $n_0$ qubits)}.
    \end{equation}    
The same is true for $n_0=4$, $\Ggp{n_0}=\SUgp{2^{n_0}}$.
\end{corollary}

\begin{proof}
We first note that since all irreducible representations appear in the regular representation\footnote{For a concrete proof in the case of $\Ggp{\ell} = \SOgp{2^\ell}$, see e.g.~\cite[Lem.~6.1]{KM15}.}, the conclusion of \Cref{thm:bs} also holds for any $\rho^{k,k}_{2^\ell}$ representation. Since the special unitary group is connected, compact, and simple\footnote{Recall that the Lie algebra of the special unitary group is simple, and that Benoiste and de~Saxc\'{e} remark, following their Theorem~1.2 in \cite{BS16}, that ``\emph{For us, a compact simple Lie group will be a compact real Lie group whose Lie algebra is simple.}''}, this immediately gives \Cref{cor:bs} in the case $\Ggp{n_0} = \SUgp{2^{n_0}}$.

For the special orthogonal case,
while $\Ggp{n_0} = \SOgp{2^{n_0}}$ is not simple, the projective special orthogonal group $\PSOgp{2^{n_0}} = \SOgp{2^{n_0}}/\{\pm 1\}$ is a connected compact simple Lie group.
Writing $P'_{n_0}$ to denote the multiset of elements of $\PSOgp{2^{n_0}}$ corresponding to $P_{n_0}$,  \Cref{thm:bs} gives us that
     \begin{equation}
         \opnorm{\E_{\bg \sim \widetilde{P'_{n_0}}}[\rho^{k,k}_{2^{n_0}}(\bg)] - \E_{\bg \sim \PSOgp{2^{n_0}}} [\rho^{k,k}_{2^{n_0}}(\bg)]} \leq 1 - \kappa. \label{eq:BdS-PSO}
     \end{equation}
Now recalling that $\rho^{k,k}_{2^{n_0}}(g)=g^{\otimes k} \otimes g^{\otimes k}$, since $P_{n_0}$ is closed under negation there is no need to distinguish between $\PSOgp{2^{n_0}}$ and $\SOgp{2^{n_0}}$ in either of the expectations appearing in \Cref{eq:BdS-PSO}, i.e. we have
\begin{equation}
\E_{\bg \sim \widetilde{P'_{n_0}}}[\rho^{k,k}_{2^{n_0}}(\bg)] = 
\E_{\bg \sim \widetilde{P_{n_0}}}[\rho^{k,k}_{2^{n_0}}(\bg)],
\quad \quad \quad
\E_{\bg \sim \PSOgp{2^{n_0}}} [\rho^{k,k}_{2^{n_0}}(\bg)]=
\E_{\bg \sim \SOgp{2^{n_0}}} [\rho^{k,k}_{2^{n_0}}(\bg)],
\label{eq:doesnt-matter-projective-or-not}
\end{equation}
which gives \Cref{cor:bs} for the case $\Ggp{n_0} = \SOgp{2^{n_0}}$.
\end{proof}

With \Cref{cor:bs} in hand, now we are ready to prove \Cref{lem:statement-G-new}:

\medskip

\emph{Proof of \Cref{lem:statement-G-new}.}
By \Cref{cor:bs}, we have
$L_{\widetilde{P_{n_0}}}(\mathrm{\rho}^{k,k}_{2^{n_0}}) \geq \kappa_{n_0} \cdot L_{\Ggp{n_0}}(\mathrm{\rho}^{k,k}_{2^{n_0}})$, i.e.
\begin{equation} \label{eq:sesame4}
\Id - \E_{\bh \sim\widetilde{P_{n_0}}}[\rho(\bh)] \geq
\kappa_{n_0} \left(
\Id - \E_{\bg \sim \Ggp{n_0}}[\rho(\bg)]
\right).
\end{equation}
    We consider tacking on $n-{n_0}$ tensor factors that are ignored by both~$\bg$ and by~$\bh$.
    Since $A \geq B \implies A \otimes \Id \geq B \otimes \Id$,  we can tensor-product both sides of \Cref{eq:sesame4} by the identity  to conclude
\begin{equation} \label{eq:sesame5}
\Id - \E_{\bh \sim\widetilde{P_{n_0}}}[\rho(\bh_{[n_0]})] \geq
\kappa_{n_0} \left(
\Id - \E_{\bg \sim \Ggp{n_0}}[\rho(\bg_{[n_0]})]\right).
\end{equation}
We can insert the ignored $n-n_0$ qubits at any positions, not just the last one; averaging the resulting inequalities, we get
\begin{align}
{\frac 1 {{n \choose n-n_0}}} \sum_{1 \leq i_1 < \cdots < i_{n_0} \leq n} 
\left(
\Id - \E_{\bh \sim \widetilde{P_{n_0}}}[\rho(\bh_{(i_1,\dots,i_{n_0})})]
\right)
&\geq 
\kappa_{n_0} \cdot
{\frac 1 {{n \choose n-n_0}}} \sum_{1 \leq i_1 < \cdots < i_{n_0} \leq n} 
\left(
\Id - \E_{\bg \sim \Ggp{n_0}}[\rho(\bg_{(i_1,\dots,i_{n_0})})] 
\right),
\label{eq:sesame6}
\end{align}
which is what \Cref{lem:statement-G-new} asserts.
\qed

\subsubsection{Proof of \Cref{lem:gateset}} \label{sec:gateset}

We first consider $\SOgp{2^4}$; so we must show that there is a finite multiset $P_{4} \subset \SOgp{2^4}$, closed under inverses, that satisfies conditions (A) and (B) of \Cref{lem:gateset}.

Define the 1- and 2-qubit gates
    \begin{equation} \label{eq:Q-and-CNOT}
    \mathrm{Q} \coloneqq \begin{bmatrix} 
        3/5 & -4/5 \\
        4/5 & \phantom{-}3/5 \\
    \end{bmatrix}
    \quad\text{and}\quad
    \mathrm{CNOT} \coloneqq  \begin{bmatrix} 
        1 & 0 \\
        0 & 0 
    \end{bmatrix} \otimes 
    \begin{bmatrix} 
        1 & 0 \\
        0 & 1 
    \end{bmatrix}  
    +
    \begin{bmatrix} 
        0 & 0 \\
        0 & 1 
    \end{bmatrix} \otimes 
    \begin{bmatrix} 
        0 & 1 \\
        1 & 0 
    \end{bmatrix},
    \end{equation}
and let $P_{4}$ be the following finite subset\footnote{Recall that $\mathrm{CNOT} \not \in \SOgp{2^2}$, but $\mathrm{CNOT} \otimes \Id_{4 \times 4} \in \SOgp{16}$.} of $\SOgp{2^4}$:
\begin{equation} \label{eq:P4prime}
        P_{4} := {\text{the closure of~}\{ \mathrm{Q}_{(j)} : j \in [4] \} \cup \{ \mathrm{CNOT}_{(i,j)} : i,j \in [4], i \neq j\}\text{~under inverses and negations}}.
\end{equation}
Clearly $P_4$ satisfies (A), and (B) follows from the following result from~\cite[Thm.~3.1]{Shi02}:
\begin{fact}    \label{fact:shi}
    The $1$- and $2$-qubit gates 
    \begin{equation}
    \mathrm{Q} \coloneqq \begin{bmatrix} 
        3/5 & -4/5 \\
        4/5 & \phantom{-}3/5 \\
    \end{bmatrix}
    \quad\text{and}\quad
    \mathrm{CNOT} \coloneqq  \begin{bmatrix} 
        1 & 0 \\
        0 & 0 
    \end{bmatrix} \otimes 
    \begin{bmatrix} 
        1 & 0 \\
        0 & 1 
    \end{bmatrix}  
    +
    \begin{bmatrix} 
        0 & 0 \\
        0 & 1 
    \end{bmatrix} \otimes 
    \begin{bmatrix} 
        0 & 1 \\
        1 & 0 
    \end{bmatrix}
    \end{equation}
    are together universal for quantum computing with real amplitudes.
    More precisely, recalling \Cref{eq:P4prime},
    we have that finite products of elements of $P_{4}$ are dense in~$\SOgp{2^{4}}$.
\end{fact}

Next we turn to $\SUgp{2^4}$. Define the $1$-qubit Hadamard gate (denoted $\mathrm{H}$), phase gate (denoted $\mathrm{S}$), and ``$\pi/8$ gate'' (denoted $\mathrm{T}$) respectively as
\begin{equation} \label{eq:universal-unitary}
   \mathrm{H} \coloneqq {\frac 1 {\sqrt{2}}} \begin{bmatrix} 
        1 & \phantom{+}1 \\
        1 & -1 \\
    \end{bmatrix}, \quad 
   \mathrm{S} \coloneqq {\frac 1 {\sqrt{2}}} \begin{bmatrix} 
        1 & \phantom{+}1 \\
        1 & -1 \\
    \end{bmatrix}, \quad 
    \quad \text{and} \quad
   \mathrm{T} \coloneqq \begin{bmatrix} 
        1 & 0 \\
        0 & e^{i \pi / 4} \\
    \end{bmatrix},    
\end{equation}
and recall the definition of $\mathrm{CNOT}$ from \Cref{eq:Q-and-CNOT}.  Now let $P'_{4}$ be the closure of $\{\mathrm{H}_{(j)}: j \in [4]\} \cup \{\mathrm{S}_{(j)}: j \in [4]\} \cup \{\mathrm{T}_{(j)}: j \in [4]\} \cup  \{ \mathrm{CNOT}_{(i,j)} : i,j \in [4], i \neq j\}$ under inverses and negations.
It is clear that $P'_4$ is a finite set of elements of $\Ugp{2^4}$, closed under inverses, satisfying (A).  The fact that $P'_4$ satisfies (B) follows from the well-known fact (see, e.g.,~\cite[Sec.~4.5.3]{NC10}) that $\mathrm{H}$, $\mathrm{S},$ $\mathrm{T}$ and $\mathrm{CNOT}$ together are universal for quantum computing. Finally, we obtain the desired set of elements $P_4 \subset \SUgp{2^4}$ by multiplying elements of $P'_4$ by suitable complex values of unit norm to have determinant one.




\section{Lower bounding $\tau_m$ for large $m$
} \label{sec:large-m}

In this section we prove \Cref{thm:large-m}, restated below, using simplifications of techniques introduced in~\cite{HH21}:

\begin{theorem} [Restatement of \Cref{thm:large-m}] \label{thm:large-m-restatement}
Let the sequence of groups $(\Ggp{n})_{n \geq 1}$ be either $(\SOgp{2^n})_{n \geq 1}$ or $(\SUgp{2^n})_{n \geq 1}$.
Define the following operators on $(\C^{2^m})^{\otimes 2k}$:
\begin{equation}
  \Pi^{(m)}  = \E_{\bg \sim \Ggp{m}}[\rho^{k,k}_{2^m}(\bg)], \quad \Pi_{[m]\setminus i} \otimes \Id_i = \textnormal{($\Id_{2k \times 2k}$ on the $i$th tensor factor, $\Pi^{(m-1)}$ on the remainder)}.
\end{equation}
Then for all  $k \leq \frac{1}{\sqrt{10}m^2} 2^{m/2}$ we have
\begin{equation} \label[ineq]{ineq:large-m-tau-lower-bound}
        \opnorm{\mathop{\avg}_{i=1}^m \{\Pi_{[m] \setminus i} \otimes \Id_i\} - \Pi^{(m)}} \leq \frac1m + \frac{\sqrt{10} km}{2^{m/2}};
\end{equation}
equivalently, in the notation of \Cref{thm:large-m}, $\tau_m \geq 1 - (\frac1m + \frac{\sqrt{10} km}{2^{m/2}})$.
\end{theorem}
We observe:
\begin{fact}  \label{fact:ac}
  $\Img \Pi^{(m)}$ is a subspace of $\Img(\Pi_{[m] \setminus i} \otimes \Id_i)$ for all~$i$.
\end{fact}

\subsection{Identifying the projectors}
To prove \Cref{thm:large-m-restatement}, we will need to have a description of the projection operator~$\Pi^{(m)}$; luckily, this is provided by known representation theory.  To state the results we need some notation.
\begin{notation} 
   If $X \in \C^{r \times r}$ is a matrix, we write $\mathrm{vec}(X) \in \C^{r} \otimes \C^{r}$ for its vectorization; here $\mathrm{vec}$ is the linear map that takes $\ket{i}\!\bra{j}$ to $\ket{ij}$.
\end{notation}
\begin{fact}  \label{fact:vectorization}
  For matrices $R_0, R_1, S \in \C^{r \times r}$ it holds that $(R_0 \otimes R_1)\mathrm{vec}(S) = \mathrm{vec}(R_0 S R_1^\top)$.
\end{fact}
\begin{notation}  \label{not:1}
    Having fixed some $D=2^m \in \N^+$, we write \begin{equation}
      \ket{\Phi} = D^{-1/2} \sum_{a=1}^D \ket{a} \otimes \ket{a} = D^{-1/2} \mathrm{vec}(\Id_{D \times D})
    \end{equation} for the maximally entangled state on $\C^D \otimes \C^D$.
\end{notation}
\begin{notation}
    For $k \in \N^+$, let $\calM_{2k}$ denote the set of all perfect matchings on~$[2k]$, and let $\calM_{2k}^{\textrm{bip}}$ denote the subset of all ``bipartite'' perfect matchings, meaning that each pair in the matching can be written as $\{i,j\}$ with $i \leq k$ and $j > k$.
\end{notation}
\begin{notation}  \label{not:2}
  For $M \in \calM_{2k}$, we introduce the unit vector
   \begin{equation}
    \ket{\Phi_M} = \bigotimes_{\{i,j\} \in M} \ket{\Phi}_{ij} \in (\C^D)^{\otimes 2k},
  \end{equation}
  where we abuse notation slightly by writing $\ket{\Phi}_{ij}$ for the maximally entangled state on the $i$th and $j$th tensor components.  
\end{notation}
Let us give two examples.  First, with $k = 3$:
\begin{equation}  \label{eqn:color}
  M = \{\{1,2\}, \{3,6\}, \{4,5\}\} \implies \ket{\Phi_{M}} = D^{-k/2} \sum_{a,b,c = 1}^D \ket{aabccb} = D^{-k/2}  \cdot \sum_{\substack{\chi : [2k] \to [D] \\ \text{all edges of $M$ monochromatic} \\ \text{for vertex-coloring $\chi$}}} \ket{\chi}.
\end{equation}
As a second example, with general~$k$:
\begin{equation}    \label{eqn:karate}
  M_0 = \{\{1,k+1\}, \{2,k+2\}, \dots, \{k, 2k\}\} \implies \ket{\Phi_{M_0}} = D^{-k/2} \mathrm{vec}(\Id_{D^k \times D^k}).
\end{equation}

It is not hard to show that every $\ket{\Phi_M}$ with $M \in \calM_{2k}$ (respectively, $M \in \calM^{\text{bip}}_{2k}$) is fixed by every $\rho_{D}^{k,k}(g)$ for $g \in \SOgp{D}$ (respectively, $g \in \SUgp{D}$).  To illustrate this for the particular~$M_0 \in \calM^{\text{bip}}_{2k} \subseteq \calM_{2k}$ from \Cref{eqn:karate}, we have that for $g \in \SOgp{D} \leq \SUgp{D}$,
\begin{equation}
  g^{\otimes k} \otimes \ol{g}^{\otimes k} \ket{\Phi_{M_0}}  = {\frac {g^{\otimes k} \otimes \ol{g}^{\otimes k} \mathrm{vec}(\Id_{D^k \times D^k})}{D^{k/2}}} = {\frac {
  \mathrm{vec}(g^{\otimes k} \mathrm{vec}(\Id_{D^k \times D^k}) (\ol{g}^{\otimes k})^\top )}{D^{k/2}}} = {\frac {\mathrm{vec}(\Id_{D^k \times D^k})}{D^{k/2}}} = \ket{\Phi_{M_0}},
\end{equation}
where we used \Cref{fact:vectorization} and $\ol{g}^\top = g^\dagger = g^{-1}$.
Given this fact, each $\ket{\Phi_M}$ must be fixed by the average representation~$\Pi^{(m)}$, and thus be in $\mathop{\mathrm{Im}} \Pi^{(m)}$.  On the other hand, it is elementary to show (e.g.,~\cite[Prop.~1]{BC20}) that $\mathop{\mathrm{Im}} \Pi^{(m)}$ is \emph{precisely} the set of vectors fixed by every operator in $\{\rho_{D}^{k,k}(g) : g \in \Ggp{m}\}$ (recall that $D=2^m$).  In turn, these are precisely the vectorizations of all matrices in the \emph{commutant} (centralizer) of
$
  \calA = \{g^{\otimes k} : g \in \Ggp{m}\}.
$
Finally, the commutants of tensor product representations of our groups have been identified under the umbrella of \emph{Schur--Weyl duality}.  
\begin{theorem} \label{thm:sherman}
  By Schur--Weyl duality for $\Ugp{D}$~\cite{Sch01,Wey39,Yua12}, $D = 2^m$, when $\Ggp{m} = \Ugp{2^m}$ the projector $\Pi^{(m)}$ has image equal to the span of $\ket{\Phi_{M}}$ for $M \in \calM_{2k}^{\mathrm{bip}}$.  The same is true when $\Ggp{m} = \SUgp{2^m}$, since $\Pi^{(m)}$ is unchanged in this case.\footnote{Observe that because of the conjugation in the definition of $\rho^{k,k}_{2^n}$, the expectation $\Pi^{(m)}$ is the same whether the expectation is taken over $\bg \sim \Ggp{n}=\SUgp{2^n}$ or $\bg \sim \Ggp{n}=\Ugp{2^n}$.}

  By Schur--Weyl duality for $\SOgp{D}$~\cite{Bra37,Gro99}, $D = 2^m$, when $\Ggp{m} = \SOgp{2^m}$ and $k < 2^{m-1}$ the projector $\Pi^{(m)}$ has  image equal to the span of $\ket{\Phi_{M}}$ for $M \in \calM_{2k}$.  
\end{theorem}
\begin{remark}
    The condition $k < 2^{m-1}$ in the previous theorem cannot be dropped.  For example,
    \begin{equation}
      \E_{\bg \sim \Ogp{2}}[\rho^{1,1}_{2}(\bg)] = \text{projection onto } \tfrac{1}{\sqrt{2}}(\ket{00} + \ket{11}),
    \end{equation}
    but 
    \begin{equation}
      \E_{\bg \sim \SOgp{2}}[\rho^{1,1}_{2}(\bg)] = \text{projection onto } \spn\{\tfrac{1}{\sqrt{2}}(\ket{00} + \ket{11}), \tfrac{1}{\sqrt{2}}(\ket{01} - \ket{10})\}.
    \end{equation}
\end{remark}

We have now identified a spanning set for $\mathop{\text{Im}} \Pi^{(m)}$, but working with it is complicated by the fact that it is not an orthonormal basis.  It is, however, relatively ``close'' to being so, as we now show (following and simplifying some arguments from~\cite[Lem.~17]{BHH16} and \cite[Lem.~9]{HH21}).  First, an elementary lemma in linear algebra:
\begin{lemma} \label{lem:gram}
  Let $W \in \C^{d \times t}$ have unit vector columns
  $\ket{w_1}, \dots, \ket{w_t}$, and suppose their Gram matrix~$W^\dagger W \in \C^{t \times t}$ is close to the identity, in the sense that $E = W^\dagger W - \Id$ has $\opnorm{E} \leq \kappa < 1$.
  (For example, this would hold if
  \begin{equation}  \label[ineq]{ineq:11}
    \norm{E}_{1 \mapsto 1} = \max_{j \in [t]} \sum_{i \neq j} \abs{\braket{w_i|w_j}} \leq \kappa,
  \end{equation}
  since generally $\norm{E}_{1 \mapsto 1} \geq \rho(E)  = \opnorm{E}$, as $E$ is Hermitian.)
  Then $WW^\dagger = \sum_i \ket{w_i}\!\bra{w_i}$ satisfies
  \begin{equation}  \label[ineq]{eqn:pproj}
    WW^\dagger \mathop{\approx}^\kappa \Pi_T,
  \end{equation}
  where $\Pi_T$ is the projector onto $T = \spn\{\ket{w_1}, \dots, \ket{w_t}\}$, and $X \mathop{\approx}^\kappa Y$ denotes $\opnorm{X-Y} \leq \kappa$.
\end{lemma}
\begin{proof}
  By hypothesis, all eigenvalues~$\lambda$ of~$W^\dagger W$ satisfy $|\lambda - 1| \leq \kappa < 1$.
  Hence $WW^\dagger$ also has these~$t$ (nonzero) $\lambda$'s within~$\kappa $ of~$1$ as eigenvalues (associated to eigenvectors in~$T$), plus possibly additional eigenvalues of~$0$ (outside~$T$). 
  This confirms \Cref{eqn:pproj}.
\end{proof}
\begin{theorem} \label{thm:kappa}
    In the setting of $\Ggp{m} = \SOgp{D}$, $D = 2^m$ and provided $k^2 \leq \frac19 D$, we have
    \begin{equation}
        \sum_{M \in \calM_{2k}} \ket{\Phi_M}\!\bra{\Phi_M} \ \mathop{\approx}^{\kappa_m} \ \Pi^{(m)},
    \end{equation}
    where $\kappa_m \coloneqq \tfrac{10}{9} \frac{k^2}{D}$.   
    In the setting of $\Ggp{m} = \SUgp{2^m}$,  the same is true with $\calM_{2k}$ replaced by $\calM_{2k}^{\textnormal{bip}}$ (and one could replace $\kappa_m$ by $\tfrac{5}{9} \frac{k^2}{D}$, but we won't).
\end{theorem}
\begin{proof}
  The result for $\Ugp{2^m}$ (hence $\SUgp{2^m}$) appears in~\cite{BHH16}, and for $\Ogp{2^m}$ in~\cite{HH21}, but we present here a representation theory-free proof, focusing on the $\SOgp{2^m}$ case.

  We will employ  \Cref{lem:gram}, with the $\ket{w_i}$'s being the $\ket{\Phi_M}$'s, $M \in \calM_{2k}$.  In particular, we will establish the premise in   \Cref{ineq:11} with $\kappa = \kappa_m$.  By symmetry of all matchings in~$\calM$, the quantity inside the maximum is the same for every ``$\ket{w_j}$''; thus, we need only bound it for one particular choice, say the~$M_0$ from \Cref{eqn:karate}.  Thus we need to establish
  \begin{equation}  \label[ineq]{ineq:verfme}
    \sum_{M \in \calM_{2k}} \abs{\braket{\Phi_M|\Phi_{M_0}}} = 1 + \sum_{M \neq M_0} \abs{\braket{\Phi_M|\Phi_{M_0}}} \leq 1+\kappa_m.
  \end{equation}
  In computing $\braket{\Phi_M|\Phi_{M_0}}$, it is easy to see (e.g., from \Cref{eqn:color}) we get a contribution of~$D^{-k}$ from every vertex-coloring $\chi : [2k] \to [D]$ that makes all edges of $M$~and~$M_0$ monochromatic.  Since $M \cup M_0$ is a union of cycles, this is equivalent to a contribution of $D^{\text{cc}(M \cup M_0)}$, where $\text{cc}(\cdot)$ denotes the number of connected components.  Thus (cf.~\cite[(B10)]{HH21})
\begin{equation}
  D^k \cdot \sum_{\text{matchings } M} \abs{\braket{\Phi_M|\Phi_{M_0}}} =  D^k \cdot \sum_{\text{matchings } M} \braket{\Phi_M|\Phi_{M_0}} = \sum_{M} D^{\text{cc}(M \cup M_0)}.
\end{equation}
The summation on the right is just the generating function (with ``indeterminate''~$D$) for the number of connected components obtained when placing a matching (initially:~$M$) onto the endpoints of~$k$ labeled paths (initially:~$M_0$).  But this is a very simple exercise.  Take the first labeled path, with endpoints~$x,y$, and consider the vertex~$z$ to which~$x$ is matched.  There are~$2k-1$ possibilities for~$z$, with one of them ($z = y$) increasing the component count by~$1$, and the other $2k-2$ increasing the count by~$0$.  Thus the generating function picks up a factor of $(D^1 + (2k-2) \cdot D^0)$, and we reduce~$k$ to $k-2$.  We conclude that (cf.~\cite[(B12)]{HH21})
\begin{equation}
    \sum_{M} D^{\text{cc}(M \cup M_0)} = (D + (2k-2)) (D + (2k-4)) \cdots (D+2) D 
\end{equation}
and hence 
\begin{equation}
  \sum_{M} \abs{\braket{\Phi_M|\Phi_{M_0}}} = (1)\parens*{1 + \tfrac{2}{D}}\parens*{1 + \tfrac{4}{D}}\cdots \parens*{1 + \tfrac{2k-2}{D}} \leq \exp(\tfrac{k(k-1)}{D}) \leq 1+\tfrac{10}{9}\tfrac{k^2}{D} = 1 + \kappa_m,
\end{equation}
the last inequality holding because we have assumed $k^2 \leq  \frac{1}{9} D$.
Thus we have indeed verified \Cref{ineq:verfme}.

The case of $\Ggp{m} = \Ugp{2^m}$ is similar; we just need to compute the generating function for bipartite matchings, meaning $\calM_{2k}^{\text{bip}}$ replaces~$\calM$.  The bound for $\kappa_m$ becomes $(1)(1+\frac{1}{D})(1+\frac{2}{D}) \cdots (1+\frac{k-1}{D}) - 1$, which is only smaller (by a factor of about~$\frac12$).
\end{proof}

\subsection{Proof of \Cref{thm:large-m-restatement}}

In this section we establish \Cref{thm:large-m-restatement}.
We begin by proving some general facts about projectors that are nearly orthogonal to each other.
\begin{lemma} \label{lem:projs}
  Let $P_1, \dots, P_m$ be orthogonal projections, and write $A = \mathop{\avg}_{i=1}^m \{P_i\}$. Then
  \begin{equation}  \label[ineq]{ineq:proo}
     \opnorm{P_i P_j} \leq \eps\ \ \forall\ i \neq j \quad \implies \quad 
   \opnorm{A} \leq \tfrac{1}{m} + \min\{\sqrt{\eps}, m \eps\}.
  \end{equation}
\end{lemma}
\begin{proof}
  We have
  \begin{equation}
    A^2 = \frac1m A + \frac{1}{m^2} \sum_{i \neq j} P_i P_j;
    \quad\implies\quad
    \opnorm{A}^2 \leq \frac1m \opnorm{A} + \frac{m(m-1)}{m^2} \eps \leq \frac1m \opnorm{A} + \eps.
  \end{equation}
  Solving the quadratic inequality yields $\opnorm{A} \leq \frac{1}{2m} + \sqrt{\frac{1}{4m^2} + \eps}$, from which the result follows.
\end{proof}
\begin{corollary} \label{cor:projs}
  In the setting of \Cref{lem:projs}, let $P$ be an orthogonal projection with $\Img P \leq \Img P_i$ for all~$i$.
  Then \Cref{ineq:proo} holds with each instance of $P_i$ replaced by~$\wt{P}_i = P_i - P$.
\end{corollary}
\begin{proof}
  It suffices to note that $\wt{P}_i^2 = \wt{P}_i$, since $P_i \cdot P = P \cdot P_i = P$.
\end{proof}
\begin{remark} \label{rem:sq}
  The identity used in the proof easily extends to $\wt{P}_{i_1} \wt{P}_{i_2} \cdots \wt{P}_{i_k} = P_{i_1} P_{i_2} \cdots P_{i_k} - P$.   Also, this identity remains true if any set of tildes is removed from the LHS (except for the set of all~$k$).
\end{remark}

Let us now study the particular orthogonal projectors involved in \Cref{thm:large-m-restatement}.  We wish to employ \Cref{cor:projs} with
\begin{equation}
  P_i \coloneqq \Pi_{[m] \setminus i} \otimes \Id_i, \quad i = 1 \dots m, \qquad P \coloneqq \Pi^{(m)}.
\end{equation}
\Cref{fact:ac} tells us \Cref{cor:projs}'s hypothesis is satisfied. We thus obtain
\begin{equation}  \label[ineq]{ineq:epsy}
  \opnorm{\mathop{\avg}_{i=1}^m \{P_i\} -  \Pi^{(m)}} \leq \frac1m + \min\{\sqrt{\eps}, m \eps\}, \quad \text{for } \eps = \max_{i \neq j} \braces*{\opnorm{\wt{P}_i \wt{P}_j}}.
\end{equation}
By symmetry of the $m$ tensor factors, we have $\eps = \opnorm{\wt{P}_1 \wt{P}_m}$, and hence
\begin{equation}  \label{eqn:puy}
  \eps^2 = \opnorm{(\wt{P}_1 \wt{P}_m)^\dagger (\wt{P}_1 \wt{P}_m)} = \opnorm{P_m \wt{P}_1 P_m},
\end{equation}
where we used \Cref{rem:sq} to get $\wt{P}_m \wt{P}_1 \wt{P}_1
\wt{P}_m = P_m \wt{P}_1 P_m$. 

Our goal will be to use \Cref{thm:kappa} (recall its~$\kappa_m$ notation) to establish the following:
\begin{align}
  \textbf{Claim: }  \quad \eps^2 = 
  \opnorm{P_m \wt{P}_1 P_m} &\leq \kappa_{m-2} + 2 \kappa_{m-1} + \kappa_m \label[ineq]{ineq:theclaim}\\
  &= \tfrac{10}{9} k^2(2^{2-m} + 2 \cdot 2^{1-m} \cdot 2^{-m}) = 10k^22^{-m} \eqqcolon \delta.
\end{align}
We will apply \Cref{thm:kappa} for $m-2, m-1, m$; its hypothesis will be satisfied even for $m-2$, since we have $k^2 \leq \frac19 2^{m-2}$ by virtue of the assumption $k^2 \leq \frac{1}{10m^4} 2^m$ in the theorem we're proving.  Moreover, this assumption implies that $\delta^{1/4} \leq 1/m$, meaning that 
\Cref{ineq:epsy} gives us the bound
\begin{equation}  
  \opnorm{\mathop{\avg}_{i=1}^m \{P_i\} -  \Pi^{(m)}} \leq \frac1m + \min\{\delta^{1/4}, m \delta^{1/2}\} = \frac1m + m \delta^{1/2} = \frac1m + \frac{\sqrt{10} km}{2^{m/2}},
\end{equation}
verifying \Cref{ineq:large-m-tau-lower-bound} and completing the proof of \Cref{thm:large-m-restatement}.  Thus it remains to establish \Cref{ineq:theclaim}.\\

To establish the claim, let us write $\calM$ for either $\calM_{2k}$ or $\calM_{2k}^{\text{bip}}$ (depending on~$\Ggp{m}$); and, for $M \in \calM$ let us write
\begin{equation}
  J_M = \ket{\phi_M}\!\bra{\phi_M}, \quad \text{where $\ket{\phi_M}$ is the $D = 2$ case of $\ket{\Phi_M}$ from \Cref{not:1,not:2}}.
\end{equation}
Then (up to tensor factoring reordering) we Have $J_M^{\otimes m} = \ket{\Phi_M}$, and hence \Cref{thm:kappa} tells us

\newcommand{\JM}{J_M} 

\begin{equation}  \label[ineq]{ineq:kk}
  \sum_{M \in \calM} \JM^{\otimes m} \mathop{\approx}^{\kappa_m} \Pi^{(m)}.
\end{equation}
We will also use this to derive

\begin{equation}
  \sum_{M \in \calM} \JM^{\otimes(m-1)} \mathop{\approx}^{\kappa_{m-1}} \Pi^{(m-1)} \quad \implies \quad \sum_{M \in \calM} \Id_1 \otimes \JM^{\otimes(m-1)}  \mathop{\approx}^{\kappa_{m-1}} P_1,
\end{equation}

\noindent where the implication is by tensoring with~$\Id_1$ (which
doesn't change operator norm differences). Using \Cref{ineq:kk} again, and the
triangle inequality, we reach

\begin{equation}  \label[ineq]{ineq:yowz}
  \wt{P}_1 = P_1 - P \mathop{\approx}^{\kappa_{m-1} + \kappa_m}   \sum_{M \in \calM} \Id_1 \otimes \JM^{\otimes(m-1)}  - \sum_{M \in \calM} \JM^{\otimes m} = \sum_{M \in \calM} \ol{J}_M \otimes \JM^{\otimes(m-1)},
\end{equation}
where $\ol{J}_M \coloneqq \Id - \JM$.
Since $\opnorm{P_m} \leq 1$, we can further conclude
\begin{align}
  P_m \wt{P}_1 P_m &\mathop{\approx}^{\kappa_{m-1} + \kappa_m}  P_m \parens*{\sum_{M \in \calM} \ol{J}_M \otimes \JM^{\otimes(m-1)}} P_m\\
  &= (\Pi^{(m-1)} \otimes \Id_m) \parens*{\sum_{M \in \calM} \ol{J}_M \otimes \JM^{\otimes(m-2)} \otimes J_M} (\Pi^{(m-1)} \otimes \Id_m) \\
  &= \sum_{M \in \calM} \parens*{\Pi^{(m-1)} (\ol{J}_M \otimes \JM^{\otimes(m-2)}) \Pi^{(m-1)}}  \otimes J_M \label[ineq]{eqn:yup}.
\end{align}
Writing
\begin{equation}
  Z_M \coloneqq \Pi^{(m-1)} (\ol{J}_M \otimes \JM^{\otimes(m-2)}) \Pi^{(m-1)},
\end{equation}
we can put \Cref{eqn:yup} into \Cref{eqn:puy} to obtain
\begin{equation}  \label[ineq]{ineq:yum}
  \eps^2 \leq \kappa_{m-1} + \kappa_m + \opnorm{\sum_{M \in \calM}  Z_M \otimes J_M}.
\end{equation}
Now $Z_M$ is PSD, being a conjugation (by $\Pi^{(m-1)})$ of a PSD matrix: the tensor product of projections $J_M$ and~$\ol{J}_M$.
Since $0 \leq J_M \leq \Id$, we therefore conclude $0 \leq Z_M \otimes J_M \leq Z_M \otimes \Id_m$.
Summing this over $M$ yields 
\begin{equation}
  0 \leq \sum_{M \in \calM} Z_M \otimes J_M \leq \sum_{M \in \calM} Z_M \otimes \Id_m = \parens*{\sum_{M \in \calM} Z_M} \otimes \Id_m,
\end{equation}
and hence (from \Cref{ineq:yum})
\begin{equation}  
  \eps^2 \leq \kappa_{m-1} + \kappa_m + \opnorm{\sum_{M \in \calM} Z_M} = \kappa_{m-1} + \kappa_m + \opnorm{\Pi^{(m-1)} \parens*{\sum_{M \in \cal M} \ol{J}_M \otimes J_M^{\otimes(m-2)}} \Pi^{(m-1)}}. \label[ineq]{ineq:yoyoyo}
\end{equation}
We have effectively now reduced from $m$ tensor components to~$m-1$.
Indeed, suppose we had defined the ``$m-1$'' analogues of $P_1, P_2, \dots$ and $P$, calling them $P_1^{(m-1)}, P_2^{(m-1)}, \dots$ and $P^{(m-1)} = \Pi^{(m-1)}$.  
Then  \Cref{ineq:yowz} would tell us
\begin{equation}
  \wt{P}_1^{(m-1)} = P^{(m-1)}_1 - P^{(m-1)} \mathop{\approx}^{\kappa_{m-2} + \kappa_{m-1}} \sum_{M \in \cal M} \ol{J}_M \otimes J_M^{\otimes(m-2)},
\end{equation}
and putting this into \Cref{ineq:yoyoyo} (using $\opnorm{P^{(m-1)}} \leq 1$) yields
\begin{equation}
  \eps^2 \leq \kappa_{m-2} + 2 \kappa_{m-1} + \kappa_m + \opnorm{P^{(m-1)} \wt{P}_1^{(m-1)} P^{(m-1)}}.
\end{equation}
But $P^{(m-1)} \wt{P}_1^{(m-1)} P^{(m-1)}$ is in fact~$0$!  (In the notation of \Cref{cor:projs} this would be ``$P \cdot \wt{P}_1 \cdot P = 0$''.)
Thus we have established the claim, \Cref{ineq:theclaim}.

\ignore{
\subsection{Some stuff}
\begin{lemma} \label{lem:gram}
  Let $W \in \C^{d \times t}$ have unit vector columns
  $\ket{w_1}, \dots, \ket{w_t}$, and suppose their Gram matrix~$W^\dagger W \in \C^{t \times t}$ is close to the identity, in the sense that $E = W^\dagger W - \Id$ has $\opnorm{E} \leq \kappa < 1$.
  (For example, this would hold if
  \begin{equation}  \label[ineq]{ineq:11}
    \norm{E}_{1 \mapsto 1} = \max_{j \in [t]} \sum_{i \neq j} \abs{\braket{w_i|w_j}} \leq \kappa,
  \end{equation}
  since generally $\norm{E}_{1 \mapsto 1} \geq \rho(E)  = \opnorm{E}$, as $E$ is Hermitian.)
  Then $WW^\dagger = \sum_i \ket{w_i}\!\bra{w_i}$ satisfies
  \begin{equation}  \label[ineq]{eqn:pproj}
    \opnorm{WW^\dagger  - \Pi_T} \leq \kappa,
  \end{equation}
  where $\Pi_T$ is the projector onto $T = \spn\{\ket{w_1}, \dots, \ket{w_t}\}$.
\end{lemma}
\begin{proof}
  By hypothesis, all eigenvalues~$\lambda$ of~$W^\dagger W$ satisfy $|\lambda - 1| \leq \kappa < 1$.
  Hence $WW^\dagger$ also has these~$t$ (nonzero) $\lambda$'s within~$\kappa $ of~$1$ as eigenvalues (associated to eigenvectors in~$T$), plus possibly additional eigenvalues of~$0$ (outside~$T$). 
  This confirms \Cref{eqn:pproj}.
\end{proof}

\rnote{Even though we can directly cite the following lemma from HHJ, their proof with representation theory severely irritated me, so I wrote this\dots}
Write $\ket{\Phi} = D^{-1/2} \sum_{a=1}^D \ket{a} \otimes \ket{a}$  for the maximally entangled state on $\C^D \otimes \C^D$, and for $M$ a matching on $[2k]$, introduce the unit vector
\begin{equation}
  \ket{\Phi_M} = \bigotimes_{\{i,j\} \in M} \ket{\Phi}_{ij} \in (\C^D)^{\otimes 2k},
\end{equation}
where we abuse notation slightly by Let$\ket{\Phi}_{ij}$ for the maximally entangled state on the $i$th and $j$th tensor components.  
As an example, with $k = 3$:
\begin{equation}  \label{eqn:color}
  M = \{\{1,2\}, \{3,6\}, \{4,5\}\} \implies \ket{\Phi_{M}} = D^{-k/2} \sum_{a,b,c = 1}^D \ket{aabccb} = D^{-k/2}  \cdot \sum_{\substack{\chi : [2k] \to [D] \\ \text{all edges of $M$ monochromatic} \\ \text{for vertex-coloring $\chi$}}} \ket{\chi}
\end{equation}

XXXXremark here that all such $\ket{\Phi_M}$, and therefore also their span, are fixed by all $\rho_{\SOgp{2^m}}^{k,k}(g)$ (here $D = 2^m$), and Brauer or Groot or whoever showed this is all of them when $D \geq 2k$ or whatever.  Also remark here same story for the unitary group, but when you restrict to the matchings in which every edge crosses the cut between $\{1, \dots, k\}$ and $\{k+1, \dots, 2k\}$.XXX

We want  to employ \Cref{lem:gram} via  \Cref{ineq:11}.  By symmetry, we may as well fix one matching, say $M_0 = \{\{1, k+1\}, \cdots, \{k, 2k\}\}$, and then bound
\begin{equation}
  \sum_{\text{matchings } M} \abs{\braket{\Phi_M|\Phi_{M_0}}} = 1 + \sum_{M \neq M_0} \abs{\braket{\Phi_M|\Phi_{M_0}}}.
\end{equation}
In computing $\braket{\Phi_M|\Phi_{M_0}}$, it is easy to see (e.g., from \Cref{eqn:color})  get a contribution of~$D^{-k}$ from every string vertex-coloring $\chi : [2k] \to [D]$ that makes all edges of $M$~and~$M_0$ monochromatic.  Since $M \cup M_0$ is a union of cycles, this is equivalent to a contribution of $D^{\text{cc}(M \cup M_0)}$, where $\text{cc}(\cdot)$ denotes the number of connected components.  Thus XXXput in ``cf HHJ B10''XXX
\begin{equation}
  D^k \cdot \sum_{\text{matchings } M} \abs{\braket{\Phi_M|\Phi_{M_0}}} =  D^k \cdot \sum_{\text{matchings } M} \braket{\Phi_M|\Phi_{M_0}} = \sum_{M} D^{\text{cc}(M \cup M_0)}.
\end{equation}
The summation on the right is just the generating function (with ``indeterminate''~$D$) for the number of connected components obtained when placing a matching (initially:~$M$) onto the endpoints of~$k$ labeled paths (initially:~$M_0$).  But this is a very simple exercise.  Take the first labeled path, with endpoints~$x,y$, and consider the vertex~$z$ to which~$x$ is matched.  There are~$2k-1$ possibilities for~$z$, with one of them ($z = y$) increasing the component count by~$1$, and the other $2k-2$ increasing the count by~$0$.  Thus the generating function picks up a factor of $(D^1 + (2k-2) \cdot D^0)$, and we reduce~$k$ to $k-2$.  We conclude that 
 XXXHHJ (B12)XXX
\begin{equation}
    \sum_{M} D^{\text{cc}(M \cup M_0)} = (D + (2k-2)) (D + (2k-4)) \cdots (D+2) D 
\end{equation}
and hence 
\begin{equation}
  \sum_{M} \abs{\braket{\Phi_M|\Phi_{M_0}}} = (1)\parens*{1 + \tfrac{2}{D}}\parens*{1 + \tfrac{4}{D}}\cdots \parens*{1 + \tfrac{2k-2}{D}} \leq \exp(\tfrac{k(k-1)}{D}) \leq 1+1.1\tfrac{k^2}{D},
\end{equation}
with the last inequality holding provided, say, $k^2 \leq \frac{1}{8} D$.

\subsection{Stuffs}
\begin{lemma} \label{lem:projs2}
  Let $P_1, \dots, P_n$ be orthogonal projections. Then
  \begin{equation}  \label[ineq]{ineq:proo2}
     \opnorm{P_i P_j} \leq \eps\ \ \forall\ i \neq j \quad \implies \quad 
   \opnorm{\mathop{\avg}_{i=1}^n \{P_i\}} \leq \tfrac{1}{n} + \min\{\sqrt{\eps}, n \eps\}.
  \end{equation}
\end{lemma}
\begin{proof}
  Let$A = \mathop{\avg}_{i=1}^n \{P_i\}$, we have 
  \begin{equation}
    A^2 = \frac1n A + \frac{1}{n^2} \sum_{i \neq j} P_i P_j;
    \quad\implies\quad
    \opnorm{A}^2 \leq \frac1n \opnorm{A} + \frac{n(n-1)}{n^2} \eps \leq \frac1n \opnorm{A} + \eps.
  \end{equation}
  Solving the quadratic inequality yields $\opnorm{A} \leq \frac{1}{2n} + \sqrt{\frac{1}{4n^2} + \eps}$, from which the result follows.
\end{proof}
\begin{corollary} \label{cor:projs2}
  In the setting of \Cref{lem:projs2}, let $P$ be an orthogonal projection with $\Img P \leq \Img P_i$ for all~$i$.
  Then \Cref{ineq:proo2} holds with each instance $P_i$ replaced by~$\wt{P}_i = P_i - P$.
\end{corollary}
\begin{proof}
  It suffices to note that $\wt{P}_i^2 = \wt{P}_i$, since $P_i \cdot P = P \cdot P_i = P$.
\end{proof}
\begin{remark} \label{rem:sq2}
  The identity used in the proof easily extends to $\wt{P}_{i_1} \wt{P}_{i_2} \cdots \wt{P}_{i_k} = P_{i_1} P_{i_2} \cdots P_{i_k} - P$.   Also, this identity remains true if any set of tildes is removed from the LHS (except for the set of all~$k$).
\end{remark}

XXXinsert words hereXXXX

We now have the following scenario:
There is a set $\calM$ of matchings on $[2k]$.
For each $M \in \calM$, let $\ket{\phi_M}$ denote the $D = 2$ case ofXXX the $\ket{\Phi_M}$ notationXXX; that is, the uniform superposition over $\ket{h}$ for $h \in \{0,1\}^{2k}$ ``$M$-respecting'' (meaning $h_i = h_j$ whenever $\{i,j\} \in M$).
(Up to reordering tensor factorsXXX, the $D = 2^n$ case of the  $\ket{\Phi_M}$ notation is $\ket{\phi_M}^{\otimes n}$.)
Let us write $\JM = \ket{\phi_M}\!\bra{\phi_M}$.
Writing also $\Pi^{(n)}$ for projection onto the span of $(\ket{\phi_M}^{\otimes n} : M \in \calM)$, we showedXXX that
\begin{equation}  \label[ineq]{ineq:kk2}
  \sum_{M \in \calM} \JM^{\otimes n} \mathop{\approx}^{\kappa_n} \Pi^{(n)},
\end{equation}
where the notation $\displaystyle A \mathop{\approx}^{\delta} B$ will mean $\opnorm{A-B} \leq \delta$.

We will employ the following notation:  For $i \subset [n]$, we will write $\Pi_{[n] \setminus i} \otimes \Id_i$ for the operator on $(\C^{2k})^{\otimes n}$ that applies $\Pi^{(n-1)}$ on all tensor factors except the $i$th, and applies $\Id_{2k \times 2k}$ on the $i$th.
\begin{fact}  \label{fact:ac2}
  $\Img \Pi^{(n)}$ is a subspace of $\Img(\Pi_{[n] \setminus i} \otimes \Id_i)$ for all~$i$.\rnote{Takes 1 second of thought to see this.}
\end{fact}
We now wish to employ \Cref{cor:projs2}, with
\begin{equation}
  P_i \coloneqq \Pi_{[n] \setminus i} \otimes \Id_i, \quad i = 1 \dots n, \qquad P \coloneqq \Pi^{(n)}.
\end{equation}
\Cref{fact:ac2} tells us \Cref{cor:projs2}'s hypothesis is satisfied.  
We thus obtain
\begin{equation}  \label[ineq]{ineq:epsy2}
  \opnorm{\mathop{\avg}_{i=1}^n \{P_i\} -  \Pi^{(n)}} \leq \frac1n + \min\{\sqrt{\eps}, n \eps\}, \quad \text{for } \eps = \max_{i \neq j} \braces*{\opnorm{\wt{P}_i \wt{P}_j}}.
\end{equation}
By symmetry of the $n$ tensor factors, we have $\eps = \opnorm{\wt{P}_1 \wt{P}_n}$, and hence
\begin{equation}  \label{eqn:puy2}
  \eps^2 = \opnorm{(\wt{P}_1 \wt{P}_n)^\dagger (\wt{P}_1 \wt{P}_n)} = \opnorm{P_n \wt{P}_1 P_n},
\end{equation}
where we used \Cref{rem:sq2} to get $\wt{P}_n \wt{P}_1 \wt{P}_1 \wt{P}_n = P_n \wt{P}_1 P_n$.
From \Cref{ineq:kk2} we can derive
\begin{equation}
  \sum_{M \in \calM} \JM^{\otimes(n-1)} \mathop{\approx}^{\kappa_{n-1}} \Pi^{(n-1)} \quad \implies \quad \sum_{M \in \calM} \Id_1 \otimes \JM^{\otimes(n-1)}  \mathop{\approx}^{\kappa_{n-1}} P_1,
\end{equation}
where  the implication is by tensoring with~$\Id_1$ (which doesn't change operator norm).
Using \Cref{ineq:kk2} again, and the triangle inequality, we reach
\begin{equation}  \label[ineq]{ineq:yowz2}
  \wt{P}_1 = P_1 - P \mathop{\approx}^{\kappa_{n-1} + \kappa_n}   \sum_{M \in \calM} \Id_1 \otimes \JM^{\otimes(n-1)}  - \sum_{M \in \calM} \JM^{\otimes n} = \sum_{M \in \calM} \ol{J}_M \otimes \JM^{\otimes(n-1)},
\end{equation}
where $\ol{J}_M \coloneqq \Id - \JM$.
Since $\opnorm{P_n} \leq 1$, we can further conclude
\begin{align}
  P_n \wt{P}_1 P_n \mathop{\approx}^{\kappa_{n-1} + \kappa_n}  P_n \parens*{\sum_{M \in \calM} \ol{J}_M \otimes \JM^{\otimes(n-1)}} P_n
  &= (\Pi^{(n-1)} \otimes \Id_n) \parens*{\sum_{M \in \calM} \ol{J}_M \otimes \JM^{\otimes(n-2)} \otimes J_M} (\Pi^{(n-1)} \otimes \Id_n) \\
  &= \sum_{M \in \calM} \parens*{\Pi^{(n-1)} (\ol{J}_M \otimes \JM^{\otimes(n-2)}) \Pi^{(n-1)}}  \otimes J_M \label[ineq]{eqn:yup2}.
\end{align}
Writing
\begin{equation}
  Z_M \coloneqq \Pi^{(n-1)} (\ol{J}_M \otimes \JM^{\otimes(n-2)}) \Pi^{(n-1)},
\end{equation}
we can put \Cref{eqn:yup2} into \Cref{eqn:puy2} to obtain
\begin{equation}  \label[ineq]{ineq:yum2}
  \eps^2 \leq \kappa_{n-1} + \kappa_n + \opnorm{\sum_{M \in \calM}  Z_M \otimes J_M}.
\end{equation}
Now $Z_M$ is PSD, being a conjugation (by $\Pi^{(n-1)})$ of a PSD matrix: the tensor product of projections $J_M$ and~$\ol{J}_M$.
Since $0 \leq J_M \leq \Id$, we therefore conclude $0 \leq Z_M \otimes J_M \leq Z_M \otimes \Id_n$.  
Summing this over $M$ yields 
\begin{equation}
  0 \leq \sum_{M \in \calM} Z_M \otimes J_M \leq \sum_{M \in \calM} Z_M \otimes \Id_n = \parens*{\sum_{M \in \calM} Z_M} \otimes \Id_n,
\end{equation}
and hence (from \Cref{ineq:yum2})
\begin{equation}  
  \eps^2 \leq \kappa_{n-1} + \kappa_n + \opnorm{\sum_{M \in \calM} Z_M} = \kappa_{n-1} + \kappa_n + \opnorm{\Pi^{(n-1)} \parens*{\sum_{M \in \cal M} \ol{J}_M \otimes J_M^{\otimes(n-2)}} \Pi^{(n-1)}}. \label[ineq]{ineq:yoyoyo2}
\end{equation}
We have effectively now reduced from $n$ tensor components to~$n-1$.  
Indeed, suppose we had defined the ``$n-1$'' analogues of $P_1, P_2, \dots$ and $P$, calling them $P_1^{(n-1)}, P_2^{(n-1)}, \dots$ and $P^{(n-1)} = \Pi^{(n-1)}$.  
Then  \Cref{ineq:yowz2} would tell us
\begin{equation}
  \wt{P}_1^{(n-1)} = P^{(n-1)}_1 - P^{(n-1)} \mathop{\approx}^{\kappa_{n-2} + \kappa_{n-1}} \sum_{M \in \cal M} \ol{J}_M \otimes J_M^{\otimes(n-2)},
\end{equation}
and putting this into \Cref{ineq:yoyoyo2} (using $\opnorm{P^{(n-1)}} \leq 1$) yields
\begin{equation}
  \eps^2 \leq \kappa_{n-2} + 2 \kappa_{n-1} + \kappa_n + \opnorm{P^{(n-1)} \wt{P}_1^{(n-1)} P^{(n-1)}}.
\end{equation}
But $P^{(n-1)} \wt{P}_1^{(n-1)} P^{(n-1)} = 0$!  (In the notation of \Cref{cor:projs2} this would be ``$P \cdot \wt{P}_1 \cdot P = 0$''.)
Thus
\begin{equation}
  \eps \leq \sqrt{ \kappa_{n-2} + 2 \kappa_{n-1} + \kappa_n}
\end{equation}
XXXwhich may be combined with \Cref{ineq:epsy2} to wrap this up; I think I have a square-root different from PedroXXX

\subsection{More}
In this section we show \Cref{thm:large-m}. We will actually show a
more general result that implies \Cref{thm:large-m} as a corollary,
but before we do so we introduce some linear algebra
concepts.

\begin{definition}
  Let $\{\ket{\phi}_i\}_i$ be a set of unit vectors that forms a
  basis. We call this set a $\kappa$-\textit{frame}, for $\kappa \geq
  0$ (the \textit{frame bound}), such that

  \begin{equation}
  \sum_i |\braket{\phi_j | \phi_i}| \leq (1 + \kappa) \qquad \text{for all } j.
  \end{equation}
\end{definition}

Notice that an orthonormal basis is a frame with $\kappa = 0$. Indeed,
this definition of frame gives us an ``approximate orthonormal
basis''.

\begin{definition}
  Let $\{\ket{\phi}_i\}_i$ be a frame and let $\{\ket{i}\}_i$ be
  an orthonormal basis of $\text{span}(\{\ket{\phi}_i\})$. Then, we
  define the \textit{analysis operator} $T$, and its adjoint the
  \textit{synthesis operator} $T^\star$, as follows

  \begin{equation}
    T = \sum_i \ket{i}\bra{\phi_i} \qquad T^\star = \sum_i \ket{\phi_i}\bra{i}.
  \end{equation}

  Additionally, we define the \textit{frame operator} $S$ as

  \begin{equation}
    S = T^\star T = \sum_i \ket{\phi_i}\bra{\phi_i}.
  \end{equation}
\end{definition}

Recall the notation $G_m$ defined in \Cref{not:gatepositioning}. To apply this
framework of frames to $L_{G_m}(\rho^{k,k}_{2^m})$ we want to first
understand its zero eigenvalue space, which is the same as the one
eigenvalue space of $\E_{\bg \sim G_m}(\rho^{k,k}_{2^m}(\bg))$. This
type of analysis was done before, for example in \cite{BV10}.\rnote{I kinda feel that if this proposition is literally exactly stated and proved elsewhere (e.g.\ in that Bordenave--Collins paper), we shouldn't also prove it ourselves.  (Particularly if it forces me to remember about measure theory :-)}

\begin{proposition} \label{prop:onespace}
  The one eigenvalue space of $\E_{\bg \sim
    G_m}(\rho^{k,k}_{2^m}(\bg))$ is spanned by all the vectors that
  are invariant under multiplication by $g^{\otimes k,k}$ for all $g
  \in G_m$.
\end{proposition}
\begin{proof}
  First note that if $\ket{x}$ is such that $g^{\otimes k,k} \ket{x} =
  \ket{x}$ for all $g \in G_m$, then trivially $\E_{\bg \sim
    G_m}(\rho^{k,k}_{2^m}(\bg)) \ket{x} = \ket{x}$.

  Now, $\ket{x}$ is a unit eigenvector of eigenvalue 1.

  \begin{equation}
  1 = |\braket{x | \E_{\bg \sim G_m}(\rho^{k,k}_{2^m}(\bg)) | x}| = |\E_{\bg \sim G_m} \braket{x | \bg^{k,k} | x}| \leq \E_{\bg \sim G_m} |\braket{x | \bg^{k,k} | x}|,
  \end{equation}

  \noindent where the inequality follows from Jensen's inequality and
  convexity of the Euclidean norm squared. Equality is achieved if an
  only if for all but measure-zero of the elements $\bg$ we have
  $|\braket{x | \bg^{k,k} | x}| = 1$, which implies $\bg^{k,k} \ket{x}
  = \ket{x}$. By continuity we can conclude that this measure-zero set
  is empty.
\end{proof}

This proposition says that in order to know the one eigenvalue space
of $\E_{\bg \sim G_m}(\rho^{k,k}_{2^m}(\bg))$ we only have to know the
invariants of $\rho^{\otimes k,k}_{2^m}$. As we will see, this is
something that has been widely studied in representation theory.

\begin{lemma} \label{lem:framebd}
  Let $\{\ket{\phi}_i\}_i$ be the set of vectors that are invariant
  under multiplication by $g^{\otimes k,k}$ for all $g \in G_m$. If
  this set is a $\kappa_m$-frame then

  \begin{equation}
    \opnorm{\E_{\bg \sim G_m}[\rho_{2^m}^{k,k}(\bg)] - S} \leq \kappa_m.
  \end{equation}
\end{lemma}
\begin{proof}
  $\{\ket{i}\}_i$ be an orthonormal basis of
  $\text{span}(\{\ket{\phi}_i\})$. Using \Cref{prop:onespace} we know
  that this space is exactly the one eigenvalue space of $\E_{\bg \sim
    G_m}[\rho_{2^m}^{k,k}(\bg)]$ and since this is a projector its
  eigendecomposition gives $\sum_i \ket{i}\bra{i}$.

  Using the fact that $S = T^\star T$ and its adjoint $T T^\star$ have the same eigenvalues we obtain

  \begin{equation}
    \opnorm{\E_{\bg \sim G_m}[\rho_{2^m}^{k,k}(\bg)] - S} = \opnorm{T T^\star - \sum_i \ket{i}\bra{i}} = \opnorm{\sum_{i\neq j} \ket{i}\bra{j} \braket{\phi_i | \phi_j}} \leq \max_i{\sum_{i \neq j} \braket{\phi_i | \phi_j}} \leq \kappa,
  \end{equation}

  \noindent where the inequality follows from the fact that for a Hermitian matrix $M$ we have that $\opnorm{M} \leq \max_i \sum_j |M_{ij}|$.
\end{proof}

We can now combine these insights to prove the following
generalization of \Cref{thm:large-m}.

\begin{theorem}
  Suppose that for all dimensions $t \leq m$ the set of vectors that
  are invariant under multiplication by $g^{\otimes k,k}$ for all $g
  \in G_t$ is a $\kappa_t$-frame, where $\kappa_t$ is some function
  that depends on $t$. Then

  \begin{equation}
    \opnorm{\E_{\bg \sim G_{m - 1} \times \binom{[m]}{m - 1}}[\rho_{2^m}^{k,k}(\bg)] - \E_{\bg \sim G_m}[\rho_{2^m}^{k,k}(\bg)]} \leq \frac{1}{m} + m (\kappa_{m - 2} + 2\kappa_{m - 1} + \kappa_m).
  \end{equation}
\end{theorem}
\begin{proof}
  Let's first make the following definitions:

  \[
  \Pi^{(m)} := \E_{\bg \sim G_m}[\rho_{2^m}^{k,k}(\bg)] \qquad Q_i := \E_{\bg \sim G_{m - 1} \times ([m] \setminus \{i\})}[\rho_{2^m}^{k,k}(\bg)].
  \]

  Recall that from \Cref{fact:projop} we know that $\Pi^{(m)}$ and all
  of the $Q_i$ are orthogonal projectors. Furthermore, using
  \Cref{prop:onespace} we know that $\Pi^{(m)}$ projects to the space
  spanned by all the vectors that are invariant under multiplication
  by $g^{\otimes k, k}$ for all $g \in G_m$, and $Q_i$ projects to the
  space spanned by all the vectors that are invariant under
  multiplication by operators that apply $g^{\otimes k, k}$ on all
  qubits except the $i$th qubit, for all $g \in G_{m - 1}$.

  Using the above, our goal is to bound $\opnorm{\frac{1}{m} \sum_{i =
      1}^mQ_i - \Pi^{(m)}}$. Call this quantity $\gamma_m$ and
  observe:

  \begin{align}
    \gamma_m^2 &= \opnorm{\frac{1}{m} \sum_{i = 1}^mQ_i - \Pi^{(m)}}^2\\
    &= \opnorm{\parens*{\frac{1}{m} \sum_{i = 1}^mQ_i - \Pi^{(m)}}^\star \parens*{\frac{1}{m} \sum_{i = 1}^mQ_i - \Pi^{(m)}}}\\
    &= \opnorm{\frac{1}{m^2} \sum_{i, j = 1}^mQ_iQ_j - \frac{1}{m}\sum_{i = 1}^mQ_i\Pi^{(m)} - \Pi^{(m)}\frac{1}{m}\sum_{i = 1}^mQ_i + \parens*{\Pi^{(m)}}^2}.
  \end{align}

  Since the $Q_i$ project to a subspace of $\Pi^{(m)}$ we have that
  $\Pi^{(m)}Q_i = Q_i\Pi^{(m)} = \Pi^{(m)}$ for all $i$. Applying this to the above we get
  
  \begin{align}
    \gamma_m^2 &= \opnorm{\frac{1}{m^2} \sum_{i, j = 1}^mQ_iQ_j - \Pi^{(m)}}\\
    &= \opnorm{\frac{1}{m^2} \sum_{i}^m\parens*{Q_i - \Pi^{(m)}} + \frac{1}{m^2} \sum_{i \neq j}\parens*{Q_iQ_j - \Pi^{(m)}}}\\
    &\leq \frac{1}{m}\opnorm{\frac{1}{m} \sum_{i}^mQ_i - \Pi^{(m)}} + \frac{1}{m^2} \sum_{i \neq j}\opnorm{Q_iQ_j - \Pi^{(m)}}\\
    &\leq \frac{1}{m} \gamma_m + \opnorm{Q_1Q_m - \Pi^{(m)}} \label{eq:opbnde1},
  \end{align}

  \noindent where the last inequality follows from the symmetries of
  the $Q_i$.

  For $t \leq m$, let $\{\ket{\phi}_i\}_i$ be the set of vectors that
  are invariant under multiplication by $g^{\otimes k,k}$ for all $g
  \in G_t$. Given the structure of the $G_t$ as described in
  \Cref{not:gatepositioning}, we can assume some extra structure on
  the vectors $\ket{\phi_i}$, namely that they are of the form
  $\ket{\varphi_i}^{\otimes t}$ and $\ket{\varphi_i} \in
  (\C^2)^{\otimes 2k}$. Since these form a $\kappa_t$-frame let's
  write their frame operator as $S^{(t)} = \sum_i \varphi_i^{\otimes t}$, where
  $\varphi_i = \ket{\varphi_i} \bra{\varphi_i}$.
  
  Notice that the action of $Q_1$ is the same as applying $\Id$ to the
  first qubit and $\Pi^{(m - 1)}$ to the remaining ones, so $Q_1 = \Id
  \otimes \Pi^{(m - 1)}$. Thus, we can define $S_1 = \Id \otimes S^{(m
    - 1)}$ as the frame operator on all the qubits except the first
  one, and applying \Cref{lem:framebd} we obtain $\opnorm{Q_1 - S_1}
  \leq \kappa_{m - 1}$. We can similarly define $S_m$ and deduce
  $\opnorm{Q_m - S_m} \leq \kappa_{m - 1}$

  We will now focus on bounding the second term of \Cref{eq:opbnde1}.

  \begin{align}
    \opnorm{Q_1Q_m - \Pi^{(m)}}^2 &= \opnorm{Q_1Q_mQ_mQ_1 - \Pi^{(m)}}\\
    &= \opnorm{Q_1Q_mQ_1 - Q_1\Pi^{(m)}Q_1}\\
    &\leq \opnorm{Q_1S_mQ_1 - Q_1S^{(m)}Q_1} + \kappa_{m - 1} + \kappa_m ,
  \end{align}

  Let us focus on the first term of the above and apply the
  definitions of the frame operators and $Q_i$. Let $\varphi_i^\perp =
  \Id - \varphi_i$. We have
  
  \begin{align}
    \opnorm{Q_1S_mQ_1 - Q_1S^{(m)}Q_1} &= \opnorm{\sum_{i} \parens*{\Id \otimes \Pi^{(m-1)}} \parens*{(\varphi^{\otimes m - 1}_i \otimes \Id) - \varphi_i^{\otimes m}} \parens*{\Id \otimes \Pi^{(m-1)}}}\\
    &= \opnorm{\sum_{i} \varphi_i \otimes \parens*{\Pi^{(m - 1)} \parens*{\varphi^{\otimes m - 2}_i \otimes \varphi_i^\perp} \Pi^{(m - 1)}}} \label{eq:projeqpsd}.
  \end{align}

  Next, we will note the following claim to bound each of the terms in
  the sum above:

  \begin{claim}
    \begin{equation}
    \varphi_i \otimes \parens*{\Pi^{(m - 1)} \parens*{\varphi^{\otimes m - 2}_i \otimes \varphi_i^\perp} \Pi^{(m - 1)}} \leq \Id \otimes \parens*{\Pi^{(m - 1)} \parens*{\varphi^{\otimes m - 2}_i \otimes \varphi_i^\perp} \Pi^{(m - 1)}},
    \end{equation}

    \noindent where the $\leq$ is PSD order.
  \end{claim}
  \begin{proof}
    \renewcommand{\qedsymbol}{$\blacksquare$} The claim is equivalent
    to proving that $\varphi_i^\perp \otimes \parens*{\Pi^{(m - 1)}
      \parens*{\varphi^{\otimes m - 2}_i \otimes \varphi_i^\perp}
      \Pi^{(m - 1)}}$ is PSD, so we will do that. We use the fact that
    if $A \geq 0$ and $B \geq 0$ then $A \otimes B \geq 0$ and show
    that each of the components is PSD.

    We show that the first component is a projector, which implies it
    is PSD. Indeed, note that $\parens*{\varphi_i^\perp}^T = \Id -
    \parens*{\varphi_i}^T = \Id -
    \parens*{\ket{\varphi_i}\bra{\varphi_i}}^T = \Id - \varphi_i
    = \varphi_i^\perp$ and $\parens*{\varphi_i^\perp}^2 = \Id -
    2\varphi_i + \varphi_i^2 = \Id - \varphi_i =
    \varphi_i^\perp$, since $\varphi_i^2 =
    \ket{\varphi_i}\bra{\varphi_i}\ket{\varphi_i}\bra{\varphi_i}
    = \ket{\varphi_i}\bra{\varphi_i}$ (recall that
    $\ket{\varphi_i}$ is a unit vector).

    To show the second component is PSD we provide a decomposition
    into a product of an operator and its transpose.
    \begin{align}
      \parens*{\Pi^{(m - 1)} \parens*{\varphi^{\otimes m - 2}_i \otimes \varphi_i^\perp}}\parens*{\Pi^{(m - 1)} \parens*{\varphi^{\otimes m - 2}_i \otimes \varphi_i^\perp}}^T &= \Pi^{(m - 1)} \parens*{\varphi^{\otimes m - 2}_i \otimes \varphi_i^\perp}^2 \Pi^{(m - 1)}\\
      &= \Pi^{(m - 1)} \parens*{\varphi^{\otimes m - 2}_i \otimes \varphi_i^\perp} \Pi^{(m - 1)},
    \end{align}
    \noindent where we used the fact that both $\varphi_i$ and
    $\varphi_i^\perp$ are projectors.
  \end{proof}

  We now apply the above claim to \Cref{eq:projeqpsd}.

  \begin{align}
    \opnorm{Q_1S_mQ_1 - Q_1S^{(m)}Q_1} &\leq \opnorm{\sum_{i} \Id \otimes \parens*{\Pi^{(m - 1)} \parens*{\varphi^{\otimes m - 2}_i \otimes \varphi_i^\perp} \Pi^{(m - 1)}}}\\
    &\leq \opnorm{\sum_{i} \parens*{\Pi^{(m - 1)} \parens*{\varphi^{\otimes m - 2}_i \otimes \varphi_i^\perp} \Pi^{(m - 1)}}}.
  \end{align}

  We can now substitute the above into \Cref{eq:proj_5}.
  
  \begin{align}
    \opnorm{Q_1Q_m - \Pi^{(m)}}^2 &\leq \opnorm{\sum_{i} \parens*{\Pi^{(m - 1)} \parens*{\varphi^{\otimes m - 2}_i \otimes \varphi_i^\perp} \Pi^{(m - 1)}}} + \kappa_{m - 1} + \kappa_m\\
    &= \opnorm{\sum_{i} \parens*{\Pi^{(m - 1)} \parens*{\varphi^{\otimes m - 2}_i \otimes \Id} \Pi^{(m - 1)} - \Pi^{(m - 1)} \varphi^{\otimes m - 1}_i \Pi^{(m - 1)}}} + \kappa_{m - 1} + \kappa_m\\
    &= \opnorm{\Pi^{(m - 1)} \parens*{S^{(m-2)} \otimes \Id} \Pi^{(m - 1)} - \Pi^{(m - 1)} S^{(m-1)} \Pi^{(m - 1)}} + \kappa_{m - 1} + \kappa_m \\
    &\leq \opnorm{\Pi^{(m - 1)} \parens*{\Pi^{(m-2)} \otimes \Id} \Pi^{(m - 1)} - \Pi^{(m - 1)} \Pi^{(m-1)} \Pi^{(m - 1)}} + \kappa_{m - 2} + 2\kappa_{m - 1} + \kappa_m\\
    &= \kappa_{m - 2} + 2\kappa_{m - 1} + \kappa_m.
  \end{align}

  To finish the proof, we combine this with \Cref{eq:opbnde1} to obtain:

  \begin{equation}
    \gamma_m^2 \leq \frac{1}{m}\gamma_m + \kappa_{m - 2} + 2\kappa_{m - 1} + \kappa_m.
  \end{equation}

  We divide both sides by $\gamma_m$ and note that if $\gamma_m < 1 /
  m$ then we are done, otherwise we obtain the claimed bound. 
\end{proof}

\subsection{Old stuff}

In this section we show an analogue of \cite[Lem.~6]{HH21}, and our
proof is roughly as the original.  The following is equivalent to \Cref{thm:large-m}:\rasnote{shall we elaborate on this equivalence a bit? \Cref{sec:BHH-extraction-stuff} doesn't mention ``operator norm'' explicitly anywhere I think}

\begin{proposition} \label{prop:lgbound}
  For $2^m > k^2$ we have:
  \[
  \opnorm{\E_{\Ogp{2^{m-1}}\times\binom{[m]}{m - 1}}[\rho_{2^m}^{k,k}(\bg)] - \E_{\Ogp{2^m}}[\rho_{2^m}^{k,k}(\bg)]} \leq \frac{1}{m} + 3\frac{km}{2^{(m - 1) / 2}}.
  \]
\end{proposition}

Before giving the proof we need to introduce some known facts about these
operators. For a discussion and proofs of the facts we present here,
see \cite{CS06}. First, recall that
$\E_{\Ogp{2^m}}[\rho_{2^m}^{k,k}(\bg)]$ is an orthogonal projection
operator.\rasnote{Just a little type-checking to make sure I'm following: it orthogonally projects elements of $((\R^2)^{\otimes m})^{\otimes 2k}$ onto a subspace of $((\R^2)^{\otimes m})^{\otimes 2k}$ (right?)}

Now, let $M_{2k}$ be the set of all pair permutations on $2k$
elements, i.e. $\sigma \in M_{2k}$ partitions $[2k]$ into pairs
$((\sigma(1), \sigma(2)), \ldots, (\sigma(2k - 1), \sigma(2k)))$ where
$\sigma(2i - 1) < \sigma(2i)$ and $\sigma(1) < \sigma(3) < \ldots <
\sigma(2k - 1)$.

For $i, j \in [k]$, we define the vectors $\ket{\Omega_{ij}} \in
(\C^2)^{\otimes 2k}$ \rasnote{Minor: $(\C^2)^{\otimes 2k}$ or $(\R^2)^{\otimes 2k}$?} as maximally entagled states in the $i, j$ tensor
factors, so $\ket{\Omega_{ij}} = \frac{1}{\sqrt{2}} (\ket{0_i 0_j} +
\ket{1_i 1_j})$. We can think of these vectors as copies of a
single qubit in the $\rho^{k,k}$ representation.\rasnote{I want to understand this better} Given a pair
partition $\sigma \in M_{2k}$ we define the vector $\ket{\phi_\sigma}$
as $\bigotimes_{i = 1}^k \ket{\Omega_{\sigma(2i - 1)
    \sigma(2i)}}$.\rasnote{I think I know what we mean here but am having a little type-checking dissonance; if I think of each $\ket{\Omega_{\sigma(2i - 1)
    \sigma(2i)}}$ as an element of $(\R^2)^{\otimes 2k}$, I shouldn't be getting an element of $(\R^2)^{\otimes 2k}$ when I tensor $k$ of them together, but $\ket{\phi_\sigma}$ is an element of $(\R^2)^{\otimes 2k}$, right?}
     One can show\rasnote{do we want to show this, or give a reference?} that
$\E_{\Ogp{2^m}}[\rho_{2^m}^{k,k}(\bg)] \ket{\phi_\sigma}^{\otimes m} =
\ket{\phi_\sigma}^{\otimes m}$, for all $\sigma \in M_{2k}$\rasnote{perhaps add something like ``; i.e.~$\ket{\phi_\sigma}^{\otimes m}$ lies in the image space of $\E_{\Ogp{2^m}}[\rho_{2^m}^{k,k}(\bg)]$  for each $\sigma \in M_{2k}$.''}.

Given $\sigma \in M_{2k}$, let $\phi_\sigma$ denote $\phi_\sigma  :=
\ket{\phi_\sigma}\bra{\phi_\sigma}$.\rasnote{Type checking: this belongs to
$((\R^2)^{\otimes 2k}) \times ((\R^2)^{\otimes 2k})$, right; so an $m$-th tensor power of it belongs to
$((\R^2)^{\otimes 2k})^{\otimes m} \times ((\R^2)^{\otimes 2k})^{\otimes m}$
which synchs up with the fact that the projection operator is likewise an element of
$((\R^2)^{\otimes m})^{\otimes 2k} \times ((\R^2)^{\otimes m})^{\otimes 2k}$; right?
} The \emph{frame operator} $S$ is
defined as $S = \sum_{\sigma \in M_{2k}} \phi_\sigma^{\otimes m}$. We
have the following result from \cite{HH21}[Lemma 9]:

\begin{lemma} \label{lem:frame_bound}
  For $2^m > k^2$ we have the following bound:

  \[
  \opnorm{\E_{\bg \sim \Ogp{2^m}}[\rho_{2^m}^{k,k}(\bg)] - S} \leq \frac{2k^2}{2^m}.
  \]
\end{lemma}

We will use this fact to prove our original proposition.

\begin{proof}[Proof of \Cref{prop:lgbound}]
  Let's first make the following definitions:

  \[
  \Pi^{(m)} := \E_{\bg \sim \Ogp{2^m}}[\rho_{2^m}^{k,k}(\bg)] \qquad Q_i := \E_{\bg \sim \Ogp{2^{m - 1}} \times ([m] \setminus \{i\})}[\rho_{2^m}^{k,k}(\bg)].
  \]

  Using the above, our goal is to bound $\opnorm{\frac{1}{m} \sum_{i =
      1}^mQ_i - \Pi^{(m)}}$. Call this quantity $\gamma_m$ and observe:

  \begin{align}
    \gamma_m^2 &= \opnorm{\frac{1}{m} \sum_{i = 1}^mQ_i - \Pi^{(m)}}^2\\
    &= \opnorm{\parens*{\frac{1}{m} \sum_{i = 1}^mQ_i - \Pi^{(m)}}^T \parens*{\frac{1}{m} \sum_{i = 1}^mQ_i - \Pi^{(m)}}}\\
    &= \opnorm{\frac{1}{m^2} \sum_{i, j = 1}^mQ_iQ_j - \frac{1}{m}\sum_{i = 1}^mQ_i\Pi^{(m)} - \Pi^{(m)}\frac{1}{m}\sum_{i = 1}^mQ_i + \parens*{\Pi^{(m)}}^2} \label{eq:proj_1}\\
    &= \opnorm{\frac{1}{m^2} \sum_{i, j = 1}^mQ_iQ_j - \Pi^{(m)}} \label{eq:proj_2}\\
    &= \opnorm{\frac{1}{m^2} \sum_{i}^m\parens*{Q_i - \Pi^{(m)}} + \frac{1}{m^2} \sum_{i \neq j}\parens*{Q_iQ_j - \Pi^{(m)}}} \label{eq:proj_2.5}\\
    &\leq \frac{1}{m}\opnorm{\frac{1}{m} \sum_{i}^mQ_i - \Pi^{(m)}} + \frac{1}{m^2} \sum_{i \neq j}\opnorm{Q_iQ_j - \Pi^{(m)}}\\
    &= \frac{1}{m} \gamma_m + \opnorm{Q_1Q_m - \Pi^{(m)}} \label{eq:proj_3},
  \end{align}

  \noindent where in \Cref{eq:proj_1} we use the fact that all these
  operators are symmetric, in \Cref{eq:proj_2} we use the fact that
  $\Pi^{(m)}$ is a projector,\rasnote{I'm being dense; is it easy to see why the fact that $\Pi^{(m)}$ is a projector implies \Cref{eq:proj_2}?} \rasnote{I guess \Cref{eq:proj_2.5} (and also some of the later lines like \Cref{eq:proj_5.5}) use the fact that $Q_i$ is a projector - can we explain why that is the case?} and \Cref{eq:proj_3} follows \rasnote{Should \Cref{eq:proj_3} be $=$ or $\leq$ --- it seems to me that to have equality we would need a ${\frac {m-1} m}$ factor in front of the second term?} from the
  symmetries of the $Q_i$. We will now focus on bounding the second
  term of the above, using \Cref{lem:frame_bound}. Let $S_i$ be the
  frame operator on all the qubits except $i$, so formally $S_i =
  \sum_{\sigma \in M_{2k}} \phi_\sigma^{\otimes i-1} \otimes \Id
  \otimes \phi_\sigma^{\otimes m - i}$, where $\Id \in (\C^2)^{\otimes 2k}$.
  We have

  \begin{align}
    \opnorm{Q_1Q_m - \Pi^{(m)}}^2 &= \opnorm{Q_1Q_mQ_mQ_1 - \Pi^{(m)}} \label{eq:proj_5}\\
    &= \opnorm{Q_1Q_mQ_1 - Q_1\Pi^{(m)}Q_1} \label{eq:proj_5.5}\\
    &\leq \opnorm{Q_1S_mQ_1 - Q_1SQ_1} + \frac{2k^2}{2^{m - 1}} + \frac{2k^2}{2^m},
    \label{eq:proj_5.75}
  \end{align}
\red{where \Cref{eq:proj_5.75} is by two applications of \Cref{lem:frame_bound} and the triangle inequality}\rasnote{this okay?}.
Let us focus on the first term of \Cref{eq:proj_5.75} and apply the definitions of the
  frame operators and $Q_i$. For $\sigma \in M_{2k}$, let
  $\phi_\sigma^\perp = \Id - \phi_\sigma$. We have
  \begin{align}
    \opnorm{Q_1S_mQ_1 - Q_1SQ_1} &= \opnorm{\sum_{\sigma \in M_{2k}} \parens*{\Id \otimes \Pi^{(m-1)}} \parens*{\phi^{\otimes m - 1}_\sigma \otimes \Id - \phi_\sigma^{\otimes m}} \parens*{\Id \otimes \Pi^{(m-1)}}}\\
    &= \opnorm{\sum_{\sigma \in M_{2k}} \phi_\sigma \otimes \parens*{\Pi^{(m - 1)} \parens*{\phi^{\otimes m - 2}_\sigma \otimes \phi_\sigma^\perp} \Pi^{(m - 1)}}}. \label{eq:proj_4}
  \end{align}

  We proceed by observing that the terms in the sum are all positive
  semi-definite operators, which \red{implies the following claim:}\rasnote{This phrasing is a little confusing to me --- is the claim using the fact that the terms in the sum are PSD?}

  \begin{claim}
    \[
    \phi_\sigma^\perp \otimes \parens*{\Pi^{(m - 1)} \parens*{\phi^{\otimes m - 2}_\sigma \otimes \phi_\sigma^\perp} \Pi^{(m - 1)}} \geq 0.
    \]
  \end{claim}
  \begin{proof}
    \renewcommand{\qedsymbol}{$\blacksquare$} We use the fact that if
    $A \geq 0$ and $B \geq 0$ then $A \otimes B \geq 0$ and show that
    each of the components is PSD.

    We show that the first component is a projector, which implies it
    is PSD. Indeed, note that $\parens*{\phi_\sigma^\perp}^T = \Id -
    \parens*{\phi_\sigma}^T = \Id -
    \parens*{\ket{\phi_\sigma}\bra{\phi_\sigma}}^T = \Id - \phi_\sigma
    = \phi_\sigma^\perp$ and $\parens*{\phi_\sigma^\perp}^2 = \Id -
    2\phi_\sigma + \phi_\sigma^2 = \Id - \phi_\sigma =
    \phi_\sigma^\perp$, since $\phi_\sigma^2 =
    \ket{\phi_\sigma}\bra{\phi_\sigma}\ket{\phi_\sigma}\bra{\phi_\sigma}
    = \ket{\phi_\sigma}\bra{\phi_\sigma}$ (recall that
    $\ket{\phi_\sigma}$ is a normalized vector).

    To show the second component is PSD we provide a decomposition
    into a product of an operator and its transpose.
    \begin{align*}
      \parens*{\Pi^{(m - 1)} \parens*{\phi^{\otimes m - 2}_\sigma \otimes \phi_\sigma^\perp}}\parens*{\Pi^{(m - 1)} \parens*{\phi^{\otimes m - 2}_\sigma \otimes \phi_\sigma^\perp}}^T &= \Pi^{(m - 1)} \parens*{\phi^{\otimes m - 2}_\sigma \otimes \phi_\sigma^\perp}^2 \Pi^{(m - 1)}\\
      &= \Pi^{(m - 1)} \parens*{\phi^{\otimes m - 2}_\sigma \otimes \phi_\sigma^\perp} \Pi^{(m - 1)},
    \end{align*}
    \noindent where we used the fact that both $\phi_\sigma$ and
    $\phi_\sigma^\perp$ are projectors.
  \end{proof}

  Applying the claim to \Cref{eq:proj_4} we obtain:\rasnote{To be sure I'm following: the next line is using that $\opnorm{A-B} \leq \opnorm{A}$ when $A,B,A-B$ are all positive semidefinite?}

  \begin{align}
    \opnorm{Q_1S_mQ_1 - Q_1SQ_1} &\leq \opnorm{\sum_{\sigma \in M_{2k}} \Id \otimes \parens*{\Pi^{(m - 1)} \parens*{\phi^{\otimes m - 2}_\sigma \otimes \phi_\sigma^\perp} \Pi^{(m - 1)}}}\\
    &\leq \opnorm{\sum_{\sigma \in M_{2k}} \parens*{\Pi^{(m - 1)} \parens*{\phi^{\otimes m - 2}_\sigma \otimes \phi_\sigma^\perp} \Pi^{(m - 1)}}}. \label{eq:grain}
  \end{align}

  We can now substitute the above into \Cref{eq:proj_5}.  Let $S^{(i)}$ be the frame
  operator on $i$ qubits, so formally $S^{(i)} = \sum_{\sigma \in
    M_{2k}} \phi_\sigma^{\otimes i}$.
    Combining \Cref{eq:proj_5.75} and \Cref{eq:grain} we get that

  \begin{align}
    \opnorm{Q_1Q_m - \Pi^{(m)}}^2 &\leq \opnorm{\sum_{\sigma \in M_{2k}} \parens*{\Pi^{(m - 1)} \parens*{\phi^{\otimes m - 2}_\sigma \otimes \phi_\sigma^\perp} \Pi^{(m - 1)}}} + \frac{3k^2}{2^{m - 1}}\\
    &= \opnorm{\sum_{\sigma \in M_{2k}} \parens*{\Pi^{(m - 1)} \parens*{\phi^{\otimes m - 2}_\sigma \otimes \Id} \Pi^{(m - 1)} - \Pi^{(m - 1)} \phi^{\otimes m - 1}_\sigma \Pi^{(m - 1)}}} + \frac{3k^2}{2^{m - 1}}\\
    &= \opnorm{\Pi^{(m - 1)} \parens*{S^{(m-2)} \otimes \Id} \Pi^{(m - 1)} - \Pi^{(m - 1)} S^{(m-1)} \Pi^{(m - 1)}} + \frac{3k^2}{2^{m - 1}} \\
    &\leq \opnorm{\Pi^{(m - 1)} \parens*{\Pi^{(m-2)} \otimes \Id} \Pi^{(m - 1)} - \Pi^{(m - 1)} \Pi^{(m-1)} \Pi^{(m - 1)}} + \frac{3k^2}{2^{m - 1}} + \frac{2k^2}{2^{m - 2}} + \frac{2k^2}{2^{m - 1}} \label{eq:spirits}\\
    &= \frac{9k^2}{2^{m - 1}},
  \end{align}
  where \Cref{eq:spirits} is by two applications of 
  \Cref{lem:frame_bound}.
  To finish the proof, we combine this with \Cref{eq:proj_3} to obtain:

  \begin{equation}
    \gamma_m^2 \leq \frac{1}{m}\gamma_m + \frac{3k}{2^{(m - 1) / 2}}.
  \end{equation}

  We divide both sides by $\gamma_m$ and note that if $\gamma_m < 1 /
  m$ then we are done, otherwise we obtain the claimed bound. 
  \rasnote{I'm freaked out by the fact that here the end of proof is an unshaded box but above it is a shaded box, even though both commands say backslash end left brace proof right brace.  what is happening}
\end{proof}

\ignore{

}
}


\section{Lower bounding $\tau_m$ for small $m$}
\label{sec:small-m}

In this section we prove \Cref{thm:small-m}, restated below:

\begin{theorem} [Restatement of \Cref{thm:small-m}] \label{thm:small-m-restatement2}
Let the sequence of groups $(\Ggp{n})_{n \geq 1}$ be either $(\SOgp{2^n})_{n \geq 1}$ or $(\SUgp{2^n})_{n \geq 1}$.
For any $m \geq  4$ we have that
\begin{equation} \label{eq:small-m-tau-lower-bound}
        \forall k \in \N^+,  \quad \quad \quad  \opnorm{\E_{\bg \sim \Ggp{m-1}\times\binom{[m]}{m - 1}}[\rho_{2^m}^{k,k}(\bg)] - \E_{\bg \sim \Ggp{m}}[\rho_{2^m}^{k,k}(\bg)]} \leq \parens*{1 - (1-\tfrac{1}{m})\tfrac{1-2^{2-m}}{4-2^{3-m}}}^{1/4} \leq .96;
\end{equation}
equivalently, in the notation of \Cref{thm:small-m}, $\tau_m \geq .04$.
\end{theorem}

\subsection{Metrics}

As discussed in \Cref{sec:ini-gap}, for $\Ggp{m}=\SOgp{2^m}$ or $\Ggp{m}=\SUgp{2^m}$ we have that $\Ggp{m} \subseteq \Ugp{2^m}$ is a compact connected Lie group with associated Lie algebra~$\frak{g}_m$, where
\begin{align}
 \text{for~} \Ggp{m} = \SOgp{2^m},\ \frak{g}_m &= \{H \in \R^{2^m \times 2^m} : H \text{ skew-symmetric}\}, \label{eq:algy1} \\
 \text{for~}  \Ggp{m} = \SUgp{2^m},\ \frak{g}_m &= \{H \in \C^{2^m \times 2^m} : H \text{ skew-Hermitian},\ \tr H=0\}. \label{eq:algy2}
\end{align}

As per  \cite[Prop.~2.11.1]{Tao14}, $\Ggp{m}$ can be given the structure of a Riemannian manifold with a bi-invariant metric.
Moreover, $\Ggp{m}$~is totally geodesic within $\Ugp{2^m}$, hence the exponential map $\exp : \mathfrak{g}_m \to \Ggp{m}$ is surjective and Riemannian distance $\dRie$ within~$\Ggp{m}$ coincides with Riemannian distance within~$\Ugp{2^m}$.
This distance can be computed straightforwardly (see, e.g., \cite[within Lem.~1.3]{Mec19}), as follows:

\begin{itemize}
  \item The Riemannian distance is bi-invariant, so $\dRie(X,Y) = \dRie(\Id,Z)$ for $Z = Y X^{-1}$.
  \item Given $Z \in \Ggp{m}$, we can choose a unique $H \in \mathfrak{g}_m$ with $\exp(H) = Z$ such that the eigenvalues of $H$ are of the form $\mathrm{i} \theta_j$ for $\theta_j \in (-\pi, \pi]$. 
We write $H = \log Z$ for this choice of~$H$.
  \item Then $\dRie(\Id,Z) = \|H\|_{\Fr} = (\sum_j \theta_j^2)^{1/2}$.
\end{itemize}
In other words,
\begin{equation}
  \dRie(X,Y) = \|\log (Y X^{-1})\|_{\Fr}.
\end{equation}

For the sake of computation it will be convenient to work not just with the Riemannian distance $\dRie$ on~$\Ggp{m}$, but also the (very similar) Frobenius distance~$d_{\Fro}$, where $d_{\Fro}(X,Y)$ denotes $\|X-Y\|_{\Fro}$.
In the above setup, now using bi-invariance of $d_{\Fro}$, we evidently have
\begin{equation}
  d_{\Fro}(X,Y) = \|\Id - Z\|_{\Fr} = \parens*{{\littlesum}_j |1-\exp(\mathrm{i}\theta_j)|^2 }^{1/2} = \parens*{{\littlesum}_j (2\sin(\theta_j/2))^2}^{1/2}.
\end{equation}
For some constant $c < .4 \leq 1$ we have the following numerical  inequality (for $|\theta| \leq \pi$):
\begin{equation}
  (2 \sin(\theta/2))^2 \leq \theta^2 \leq (2 \sin(\theta/2))^2 + c (2 \sin(\theta/2))^4.
\end{equation}
Using just $c \leq 1$, we may conclude\footnote{Here we are clarifying slightly the deduction of~\cite[eq.~(112a)]{BHH16}.}
\begin{equation} \label[ineq]{ineq:sqrr}
  d_{\Fro}(X,Y)^2 \leq \dRie(X,Y)^2 \leq d_{\Fro}(X,Y)^2 + d_{\Fro}(X,Y)^4.
\end{equation}
Finally, we will also use the operator-norm distance, $d_{\mathrm{op}}(X,Y) = \opnorm{X - Y}$, which satisfies $d_{\mathrm{op}}(X,Y) \leq d_{\Fro}(X,Y)$.

\bigskip

We now move on to considering (Borel) probability measures on  metric spaces (always assumed to be complete and separable).
First we recall some basic definitions:
\begin{definition} \label{def:coupling}
  A pair of jointly distributed random variables $(\bX,\bY)$ is a \emph{coupling} of probability distributions $\nu_1,\nu_2$ if $\bX$ (respectively,~$\bY$) has marginal distribution $\nu_1$ (respectively, $\nu_2$).
\end{definition}
  
\begin{definition} \label{def:wasserstein}
  On the metric space $(M,d)$, the  \emph{$L^p$-Wasserstein distance} between two measures $\nu_1$ and~$\nu_2$ is 
  \begin{equation} \label{eq:lpwasserstein}
    W_{d,p}(\nu_1,\nu_2) = \inf \braces*{
  \E[d(\bX,\bY)^p]^{1/p} \ : \ (\bX,\bY)\text{~is a coupling of~}(\nu_1,\nu_2)
    }.
  \end{equation}
\end{definition}

\begin{notation}
    If $\nu$ is a probability measure on metric space~$M$ and $K$ is a Markov transition kernel on~$M$, we write $K^\ell \nu$ for the probability measure on~$M$ resulting from starting with probability measure~$\nu$ and taking $\ell \in \N$ steps according to~$K$.
\end{notation}

\subsection{Oliveira's theorem and its consequences}
We now state a key result of Oliveira~\cite{Oliveira09} that says that on any \emph{length space} (see e.g.~\cite{Bridson1999}),  $L^2$-Wasserstein local contraction implies global contraction.  
As we only need the result in the particular case of compact, connected Lie groups (which are finite-diameter complete Riemannian manifolds), we state it only in this simpler context:
\begin{theorem} (Implied by \cite[Thm.~3]{Oliveira09}.) \label{thm:oliveira}
  Let $(M,d)$ be a finite-diameter complete Riemannian manifold, and let $K$ be a Markov transition kernel on~$M$ satisfying the following:
  \begin{equation}  \label[ineq]{ineq:oliveira-assumption}
    W_{d,2}(K \delta_X, K\delta_Y) \leq (\eta + o(1)) d(X,Y), \quad \text{with respect to $d(X,Y) \to 0$.}
  \end{equation}
  (Here $\delta_Z$ denotes the measure that puts all of its probability mass on $Z \in M$.)
  Then for all probability measures $\nu_1, \nu_2$ on~$M$ it holds that 
  \begin{equation} 
    W_{d,2}(K\nu_1,K\nu_2) \leq \eta \cdot W_{d,2}(\nu_1,\nu_2).
  \end{equation}
\end{theorem}
Iterating this yields the following:
\begin{corollary}
  In the setting of \Cref{thm:oliveira}, for any $\ell \in \N^+$ we have
  \begin{equation}
    W_{d,2}(K^\ell \nu_1,K^\ell \nu_2) \leq \eta^\ell \cdot W_{d,2}(\nu_1,\nu_2) \leq D \eta^\ell,
  \end{equation}
  where $D$ is an upper bound on the diameter of~$M$.
\end{corollary}
We now specialize this corollary to the case where $(M,d)$ is $(\Ggp{m}, \dRie)$; combining it with \Cref{def:wasserstein} and using also $W_{{d_{\mathrm{op}}},1} \leq W_{d_{\Fro},1} \leq W_{d_{\Fro},2} \leq W_{\dRie,2}$, we may conclude:
\begin{corollary} \label{cor:puppy}
    Let $\Ggp{m}$ be a compact connected Lie group, and let $K$ be a Markov transition kernel on~$\Ggp{m}$ such that \Cref{ineq:oliveira-assumption} holds for $\dRie$ with constant~$\eta$.
    Then for any probability measures $\nu_1, \nu_2$ on~$\Ggp{m}$, and any $\ell \in \N^+$,  there is a coupling $(\bX,\bY)$ of the measures $K^\ell \nu_1, K^\ell \nu_2$ under which
    \begin{equation}
      \E[\opnorm{\bX - \bY}] \leq 2 D \eta^\ell
    \end{equation}
    (where $D$ is a bound on the $\dRie$-diameter of~$\Ggp{m}$, and the factor~$2$ accounts for the~$\inf$).
\end{corollary}

Our next step is to get rid of the coupling in \Cref{cor:puppy}. To do this, we first observe that
the representation
$\rho^{k,k}_{2^m}$ is uniformly continuous on~$\Ggp{m}$ with respect to the operator-norm distance.
Concretely, from the identity
\begin{equation}
  g_1 \otimes \cdots \otimes g_{K} - 
  h_1 \otimes \cdots \otimes h_{K} 
  = \sum_{i = 1}^K g_1 \otimes \cdots \otimes g_{i-1} \otimes (g_i - h_i) \otimes h_{i+1} \otimes \cdots \otimes h_K
\end{equation}
and $\opnorm{X}, \opnorm{\overline{X}} = 1$ for $X \in \Ggp{m}$, as well as multiplicativity 
of $d_{\mathrm{op}}$ with respect to tensor products, we may conclude that
\begin{equation}
  \opnorm{\rho^{k,k}_{2^m}(X) - \rho^{k,k}_{2^m}(Y)} \leq 2k \opnorm{X-Y}
\end{equation}
for any $X,Y \in \Ggp{m}$.
Using this, as well as the triangle inequality for $d_{\mathrm{op}}$, in \Cref{cor:puppy} yields:
\begin{corollary} \label{cor:dog}
  In the setting of \Cref{cor:puppy}, 
  \begin{equation}
    \opnorm{\E_{\bX \sim K^\ell \nu_1}[\rho^{k,k}_{2^m}(\bX)] - \E_{\bY \sim K^\ell \nu_2}[\rho^{k,k}_{2^m}(\bY)]} \leq 4k D \eta^\ell.
  \end{equation}
\end{corollary}
\noindent (Note that in contrast with \Cref{cor:puppy}, here \Cref{cor:dog} does not feature any coupling between $K^\ell \nu_1$ and $K^\ell \nu_2$.)

Now we further specialize by taking $\nu_1 = \delta_{\Id}$ (the measure with all probability on the identity element $\Id \in \Ggp{m}$), taking $\nu_2$ to be Haar measure, and specifying that
\begin{equation}  \label{eqn:myK}
  K \text{ arises from left-multiplying by a random $\bg \sim \calP$},
\end{equation}
where $\calP$ is some symmetric probability distribution on~$\Ggp{m}$ as in \Cref{def:LW}.
Note that, whatever~$\calP$ is, we have $K^\ell \nu_2 = \nu_2$ (Haar measure), and
\begin{equation}
  \E_{\bX \sim K^\ell \nu_1}[\rho^{k,k}_{2^m}(\bX)] 
  = \E_{\substack{\bg_1, \dots, \bg_\ell \sim \calP \\ \text{independent}}}[\rho^{k,k}_{2^m}(\bg_\ell \cdots \bg_1)] = \E[\rho^{k,k}_{2^m}(\bg_\ell) \cdots \rho^{k,k}_{2^m}(\bg_1)] = \E_{\bg \sim \calP}[\rho^{k,k}_{2^m}(\bg)]^\ell.
\end{equation}
From this and \Cref{cor:dog} we conclude the following:
\begin{corollary} \label{cor:kitten}
  Let $\calP$ be a symmetric probability distribution on~$\Ggp{m}$.
  Given $X,Y \in \Ggp{m}$, write $\calP^{(X)}$ (respectively,~$\calP^{(Y)}$) for the distribution of $\bg X$ (respectively, $\bg Y$) when $\bg \sim \calP$.
  Then supposing
  \begin{equation}
    W_{\dRie,2}(\calP^{(X)},\calP^{(Y)}) \leq (\eta + o(1)) \dRie(X,Y) \quad \text{with respect to $\dRie(X,Y) \to 0$,}
  \end{equation}
  it follows that 
  for any $\ell, k \in \N^+$ we have
  \begin{equation}
    \opnorm{\E_{\bg \sim \calP}[\rho^{k,k}_{2^m}(\bg)]^\ell - \E_{\bg \sim \Ggp{m}}[\rho^{k,k}_{2^m}(\bg)]} \leq 4kD \cdot \eta^\ell.
  \end{equation}
\end{corollary}

Our goal for the next section will be to establish the following: 
\begin{theorem} \label{thm:couple}
  Let $\nu_m$ denote the distribution $\Ggp{m-1} \times \binom{[m]}{m-1}$ on~$\Ggp{m}$, thought of as inducing a Markov chain on~$\Ggp{m}$ via left-multiplication.
  Fix any $X,Y \in \Ggp{m}$ with $\dRie(X,Y) = \eps \leq 1$, and let $\bX''$ (respectively,~$\bY''$) denote the result of taking \emph{two} independent steps from~$X$ (respectively,~$Y$) according to~$\nu_m$.
  Then there is a coupling of $\bX'',\bY''$ under which
  \begin{equation}
    \E[\dRie(\bX'', \bY'')^2] \leq (1 - \gamma_m) \eps^2 + O_m(\eps^3),
  \end{equation}
  where $\gamma_m = (1-\frac{1}{m})\gamma'_m$ with $\gamma'_m = \frac{1-2^{2-m}}{4 - 2^{3-m}}$, and the $O_m(\cdot)$ hides a constant depending only on~$m$.
\end{theorem}
This theorem establishes the hypothesis of \Cref{cor:kitten} with $\calP = \nu_m \ast \nu_m$ and $\eta = \sqrt{1-\gamma_m}$.  We can therefore easily derive the following (where the equality uses the fact that
$\E_{\bg \sim \Ggp{m}}[\rho^{k,k}_{2^m}(\bg)]$ is a projection operator):
\begin{equation}
  \opnorm{\E_{\bg \sim \nu_m}[\rho^{k,k}_{2^m}(\bg)]^{2\ell} - \E_{\bg \sim \Ggp{m}}[\rho^{k,k}_{2^m}(\bg)]} = \opnorm{\E_{\bg \sim \nu_m}[\rho^{k,k}_{2^m}(\bg)] - \E_{\bg \sim \Ggp{m}}[\rho^{k,k}_{2^m}(\bg)]}^{2\ell} \leq 4kD \cdot (1-\gamma_m)^{\ell/2}. 
\end{equation}
Taking $(2\ell)$th roots and then $\ell \to \infty$ thus yields \Cref{thm:small-m-restatement2}.


\subsection{Proof of \Cref{thm:couple}}

%
%
%
%

We begin by describing the needed coupling. 
First, we use the same randomness to take one step from each of~$X,Y$; that is, we define
\begin{equation} \label{eq:XprimeYprime}
  \bX' = \bg_{[m]  \setminus \bi} \cdot X, \qquad \bY' = \bg_{[m]  \setminus \bi} \cdot Y, 
\end{equation}
where $\bi \sim [m]$, $\bg \sim \Ggp{m-1}$ are uniformly random and independent. 
To take the second steps, we first draw~$\bj \sim [m]$.  
Then, based on the outcomes $\bi, \bj, \bg$, we will deterministically define some 
\begin{equation}
  \bh = h(\bi, \bj, \bg) \in \Ggp{m-1}
\end{equation}
and then take
\begin{equation}
  \bX'' = (\wt{\bg} \bh)_{[m] \setminus \bj} \cdot \bX', \qquad 
  \bY'' = \wt{\bg}_{[m] \setminus \bj} \cdot \bY',
\end{equation}
where $\wt{\bg} \sim \Ggp{m-1}$ is drawn uniformly and independently of all other random variables.
This is a valid coupling, since for every outcome of $\bi, \bj, \bg$ the distributions of $\wt{\bg} \bh$ and $\wt{\bg}$ are identical.
Then
\begin{equation}  \label{eqn:minme}
  \dRie(\bX'', \bY'') = \dRie(\bh_{[m] \setminus \bj} \cdot  \bX',  \bY'),
\end{equation}
since $\dRie({\cdot},{\cdot})$ is unitarily invariant.
In case $\bi = \bj$, we will ``give up'' and simply define~$\bh = \Id$, in which case we get $\dRie(\bX'', \bY'') = \dRie(\bX', \bY') = \dRie(X,Y) = \eps$.  
Thus we have 
\begin{equation}
  \E[\dRie(\bX'', \bY'')^2] = \frac1m \eps^2 + \parens*{1 - \frac{1}{m}}\mathop{\avg}_{\bi \neq \bj} \braces*{ \E_{\bg \sim \Ggp{m-1}}\bracks*{\dRie\parens*{\bh_{[m] \setminus \bj} \cdot \bX', \bY'}^2}}.
\end{equation}
To complete the definition of~$\bh$, we specify the function $h$:
\begin{equation}
  \text{for $i \neq j$, we define } h=h(i,j,g) \text{ to minimize } d_{\Fro}\parens*{h_{[m] \setminus j} \cdot g_{[m]  \setminus i} \cdot X, g_{[m]  \setminus i} \cdot Y}^2;
\end{equation}
in other words, $\bh = h(\bi,\bj,\bg)$ minimizes $d_{\Fro}\parens*{\bh_{[m] \setminus \bj} \cdot \bX',\bY'}$.
With this choice of $h$, note that we have $d_{\Fro}\parens*{\bh_{[m] \setminus \bj} \cdot \bX', \bY'} \leq \eps \leq 1$ for every outcome of $\bi,\bj,\bg$, since $\bh = \Id$ is always an option (and $d_{\Fro}(\bX',\bY') = d_{\Fro}(X,Y) \leq \dRie(X,Y) = \eps$).  
Thus employing \Cref{ineq:sqrr} we may conclude
\begin{equation}
  \E[\dRie(\bX'', \bY'')^2] \leq \frac1m \eps^2 + \parens*{1 - \frac{1}{m}} \mathop{\avg}_{\bi \neq \bj} \braces*{\E_{\bg \sim \Ggp{m-1}}\bracks*{d_{\Fro}\parens*{\bh_{[m] \setminus \bj} \cdot \bX',\bY'}^2}} + \eps^4.
\end{equation}
Thus to complete the proof of \Cref{thm:couple}, it suffices to establish the following:
\begin{equation}  \label[ineq]{ineq:estme}
  \forall i \neq j, \qquad \E_{\bg \sim \Ggp{m-1}}\bracks*{\min_{h \in \Ggp{m-1}} \braces*{d_{\Fro}\parens*{h_{[m] \setminus j}, \bY' \cdot (\bX')^{-1}}^2}} \leq (1-\gamma'_m) \eps^2 + O_m(\eps^3).
\end{equation}
(Here we used $d_{\Fro}\parens*{h_{[m] \setminus j} \cdot \bX',  \bY'} = d_{\Fro}\parens*{h_{[m] \setminus j}, \bY' \cdot (\bX')^{-1}}$.)
Our proof of this will not have any particular dependence on $i,j$, so without loss of generality let us fix $i = 1$ and $j = m$.  
We establish \Cref{ineq:estme} via the below two lemmas.  (Here and subsequently the notation ``$\tr_i X$'' below denotes the partial trace corresponding to tracing out the $i$th qubit of $X$.)

\begin{lemma} \label{lem:yummy1}
  Fix any $Z \in \Ggp{m}$ with $\dRie(\Id,Z) = \eps$.
  Then
  \begin{equation}  \label[ineq]{ineq:yummy1}
    \min_{h \in \Ggp{m-1}} \braces*{d_{\Fro}\parens*{h \otimes \Id, Z}^2} \leq (1 - \tfrac12 \norm{\tr_m B}_{\Fr}^2) \eps^2 + O_m(\eps^3),
  \end{equation}
  where  $B = \tfrac{1}{\eps} \log Z \in \frak{g}_m$ satisfies $\|B\|_{\Fro} = 1$.
\end{lemma}
\begin{lemma} \label{lem:yummy2}
  For $m \geq 2$ and any $A \in \frak{g}_m$, writing $\delta = 2^{2-m}$, we have 
  \begin{equation} \label[ineq]{ineq:yummy2}
    \E_{\bg \sim \Ggp{m-1}}[\norm{\tr_m((\Id \otimes \bg_{[m]\setminus 1})A(\Id \otimes \bg_{[m]\setminus 1}^{\dagger}))}_{\Fr}^2] \geq \frac{1 - \delta}{2 - \delta}\|A\|_{\Fr}^2.
  \end{equation}
\end{lemma}
To see how the above two lemmas imply \Cref{ineq:estme} (in the case $i = 1$, $j = m$), we first apply \Cref{lem:yummy1} with $Z$~being the outcome of $\bY' \cdot (\bX')^{-1}$. 
Writing $\bB = \frac1\eps \log(\bY' (\bX')^{-1})$  (recall that from \Cref{eq:XprimeYprime} this is  a random matrix depending on~$\bg$),  \Cref{lem:yummy1} tells us that
\begin{equation}
  \E_{\bg \sim \Ggp{m-1}}\bracks*{\min_{h \in \Ggp{m-1}} \braces*{d_{\Fro}\parens*{h_{[m] \setminus j} \otimes \Id_j, \bY' \cdot (\bX')^{-1}}^2}} \leq (1-\tfrac12 \E[\|\tr_m \bB\|_{\Fro}^2]) \eps^2 + O_m(\eps^3).
\end{equation}
But, for $\bg \sim \Ggp{m-1}$, we have
\begin{equation}
  \bB =  \tfrac{1}{\eps} \log \parens*{(\Id \otimes \bg_{[m]\setminus 1}) Y X^{-1}  (\Id \otimes  \bg_{[m]\setminus 1}^{\dagger})} = (\Id \otimes  \bg_{[m]\setminus 1})\parens*{ \tfrac{1}{\eps} \log (Y X^{-1})  } (\Id \otimes  \bg_{[m]\setminus 1}^{\dagger}).
\end{equation}
The result now follows by applying \Cref{lem:yummy2} with $A = \frac{1}{\eps} \log(Y X^{-1})$, which has $\|A\|_{\Fr} = 1$ since $\dRie(X,Y) = \eps$.

\subsubsection{Proof of \Cref{lem:yummy1}}

To prove \Cref{lem:yummy1}, it suffices to show that the particular choice
\newcommand{\cc}{-\tfrac12}
\begin{equation}
   h \coloneqq \exp(\cc \eps \tr_m B)
\end{equation}
satisfies \Cref{ineq:yummy1}.
We observe that since $\Ggp{m}$ is either $\SOgp{2^m}$ or $\SUgp{2^m}$, recalling \Cref{eq:algy1,eq:algy2} we have that $\tr_m B \in \mathfrak{g}_{m-1}$ since $B \in \mathfrak{g}_m$ (note that $\tr B=0$ implies $\tr(\tr_m B)=0$),
and hence indeed $h \in \Ggp{m-1}$ as required.
Now we must bound 
%
\begin{equation}  \label{eqn:whatami}
  d_{\Fro}\parens*{h \otimes \Id, Z}^2 = \la h \otimes \Id - Z, h \otimes \Id - Z\ra = 2 \tr \Id - \la h\otimes \Id,Z \ra - \la Z, h\otimes \Id \ra = 2 \tr \Id - 2\Re \la h\otimes \Id,Z \ra.%
  %
\end{equation}
Recalling $Z = \exp(\eps B)$ where $\|B\|_{\Fr} = 1$, we abuse notation slightly by writing 
\begin{equation}
  Z = 1 + \eps B + \eps^2 B^2/2 + O_m(\eps^3),
\end{equation}
where ``$O_m(\eps^3)$'' stands for some matrix~$E$ satisfying $\|E\|_{\Fr} \leq C \eps^3$, with $C$ a constant depending only on~$m$ that may change from line to line.

We may similarly expand $h \otimes \Id = \exp(\cc \eps \tr_m B) \otimes \Id$, and upon substituting into~\Cref{eqn:whatami} and simplifying, we obtain
\begin{equation}  \label{eqn:putty}
  \eqref{eqn:whatami} = \Re\tr(T - 2B) \eps  + \Re\tr\parens*{TB  - \tfrac14 T^2 - B^2} \eps^2 + O_m(\eps^3), \qquad T \coloneqq (\tr_m B) \otimes \Id.
\end{equation}
Now $B, T$ are both in the Lie algebra for $\Ggp{m}$; i.e., they are skew-symmetric in the case $\Ggp{m} = \SOgp{2^m}$, and traceless skew-Hermitian  in the case $\Ggp{m} = \SUgp{2^m}$.
Thus both have purely imaginary trace, meaning $\Re \tr(T-2B) = 0$.
Moreover, $B$~skew-Hermitian implies $\tr B^2 = \tr B(-B^\dagger) = -\la B,B\ra = -\|B\|_{\Fr}^2 = -1$, and similarly $\tr T^2 =  -\|T\|_{\Fr}^2 = -2 \norm{\tr_m B}_{\Fr}^2$.
Finally, 
\begin{equation}
  \tr(TB) = -\la T, B\ra = -\la (\tr_m B) \otimes \Id, B\ra = -\la \tr_m B, \tr_m B \ra = -\norm{\tr_m B}_{\Fr}^2.
\end{equation}
Putting these deductions into \Cref{eqn:putty} yields
\begin{equation}
  \eqref{eqn:whatami} = (1-\tfrac12\norm{\tr_m B}_{\Fr}^2)\eps^2 + O_m(\eps^3),
\end{equation}
completing the proof of \Cref{lem:yummy1}.

\subsubsection{Proof of \Cref{lem:yummy2}}  \label{subsubsec:yummy2}

Given $m$ qubits, we'll write $L = \{1, \dots, m-1\}$ for the system defined by the first~$m-1$ of them, and (slightly abusing notation) write~$m$ for the system defined by the $m$th one.
We will also write $L'$~and~$m'$ for duplicate copies of these systems, and given a subset $S \subseteq [m]$, we write $\textrm{SWAP}_{S,S'}$ to denote the operator that swaps the qubits in $S$ with the corresponding subset of qubits $1',\dots,m'.$ 
Now for any $(m-1)$-qubit operator~$C$ we have
\begin{equation}  \label{eqn:s1}
  \norm{C}_{\Fro}^2 = \tr(C^\dagger C) = \tr((C_L^\dagger \otimes C_{L'}) \cdot \textrm{SWAP}_{L,L'}).
\end{equation}
In turn, if $C = \tr_m B$ for some operator $B$ on~$m$ qubits, we conclude
\begin{equation}  \label{eqn:s2}
  \norm{\tr_m B}_{\Fro}^2 
  = \tr(\tr_{m,m'}(B^\dagger \otimes B)\cdot \textrm{SWAP}_{L,L'}) 
  = \tr((B^\dagger \otimes B)\cdot (\textrm{SWAP}_{L,L'} \otimes \Id_{m,m'})).
\end{equation}
Next, if $B = H A H^\dagger$ for unitary~$H$, we may use the cyclic property of trace to conclude
\begin{equation}  \label{eqn:cocoa}
    \norm{\tr_m B}_{\Fro}^2 
    = \tr((A^\dagger \otimes A)\cdot W) = \la A \otimes A^\dagger, W \ra, \qquad W \coloneqq (H \otimes H) (\textrm{SWAP}_{L,L'} \otimes \Id_{m,m'}) (H^\dagger \otimes H^\dagger).
\end{equation}
(The above formula, specialized to $m=3$, essentially appears as~\cite[Eqn.~(103)]{BHH16}.)
Finally, suppose $H = \Id \otimes g$ for some $(m-1)$-qubit unitary~$g$.  
For notational clarity we break up the system $L$ into subsystems ``$1$'' and $K = \{2, \dots, m-1\}$,  writing $H = \Id_1 \otimes g_{K,m}$ and 
\begin{equation}
  H \otimes H = \Id_{1,1'} \otimes (g_{K,m} \otimes g_{K',m'}).
\end{equation}
Putting this into the definition of~$W$, we see that the two qubits labeled $1$~and~$1'$ are simply swapped by~$W$, and we have
\begin{equation}
  W = \mathrm{SWAP}_{1,1'} \cdot \wh{W}, \qquad \wh{W} \coloneqq (g_{K,m} \otimes g_{K',m'}) S (g_{K,m}^\dagger \otimes g_{K',m'}^\dagger), \qquad S \coloneqq (\mathrm{SWAP}_{K,K'} \otimes \Id_{m,m'}). \label{eq:WS}
\end{equation}
Recalling \Cref{fact:vectorization},
we see that
\begin{equation}
  \mathrm{vec}(\wh{W}) = \rho^{2,2}_{2^{m-1}}(g) \cdot \mathrm{vec}(S).
\end{equation}
In other words, $\wh{W}$ is the action of $g$ on $S$ under representation $\rho^{2,2}_{2^{m-1}}$, when we suitably use the ``matricized'' interpretation of this representation.
Finally, we are interested in the case that $\bg \sim \Ggp{m-1}$ is chosen ``uniformly'' (Haar measure on $\Ggp{m-1}$); then we conclude from the above equations that
\begin{equation}  \label{eqn:putitin}
  \E_{\bg \sim \Ggp{m-1}}[\norm{\tr_m((\Id \otimes \bg)A(\Id \otimes \bg^{\dagger}))}_{\Fr}^2]
  = \langle A \otimes A^\dagger, \mathrm{SWAP}_{1,1'} \cdot S_0 \rangle, \qquad \mathrm{vec}(S_0) \coloneqq \E_{\bg \sim \Ggp{m-1}} [\rho^{2,2}_{2^{m-1}}(\bg)] \cdot \mathrm{vec}(S).
\end{equation}
We now compute $S_0$ (we note that a similar calculation for $\Ggp{m-1}=\Ugp{2^{m-1}}$ is given in~\cite[Eqn.~(61)]{HH21}):

\begin{proposition} \label{prop:reptheory}
  Let  $D = 2^{m-1}$, and define the following operators acting across systems $K \cup \{m\}$, $K' \cup \{m'\}$:
  \begin{equation}
      Q_2 = D \cdot \ket{\Phi}\!\bra{\Phi}, \quad
      Q_3 = \Id, \quad
      Q_4 = \mathrm{SWAP}
  \end{equation}
  (where   $\ket{\Phi} = D^{-1/2} \sum_{a \in \zo^{m-1}} \ket{a} \otimes \ket{a}$ is the maximally entangled state).
  Then 
  \begin{align}
    \Ggp{m-1} &= \SOgp{D} &\implies \phantom{Dc_2 \cdot Q_2 + {}} S_0 &= c_2 \cdot Q_2+ c_3 \cdot Q_3 + c_4 \cdot Q_4, \\
    \Ggp{m-1} &= \SUgp{D} &\implies  \phantom{Dc_2 \cdot Q_2 + {}} 
    S_0 &=c'_3 \cdot Q_3 + c'_4 \cdot Q_4,
  \end{align}  
  where the non-negative constants $c_2,c_3,c_4,c'_3,c'_4$ are given by
  \begin{gather}
    c_2 =   \frac{D/2-1}{(D-1)(D+2)}, 
  \quad c_3 = \frac{3D/2+1}{(D-1)(D+2)},
  \quad c_4 = \frac{(D/2-1)(D+3)}{(D-1)(D+2)}, \\
  \quad c'_3 = \frac{3D/2}{(D-1)(D+1)},
  \quad c'_4 = \frac{D^2/2 - 2}{(D-1)(D+1)} \geq c_4.
  \end{gather}
\end{proposition}
\begin{proof}
  We recall from \Cref{thm:sherman} that\footnote{Note that here we are using $m \geq 4$.}
\begin{equation}
  \E_{\bg \sim \Ggp{m-1}}[\rho^{2,2}_{2^{m-1}}(\bg)] = \text{projection onto the span of } \{\ket{\varphi_M} : M \in \calM\},
\end{equation}
where we use the following notation:
\begin{gather}
  M_{12} = \{\{1,2\}, \{3,4\}\}, \quad
  M_{13} = \{\{1,3\}, \{2,4\}\}, \quad
  M_{14}= \{\{1,4\}, \{2,3\}\};\\
   \ket{\varphi_{M_{12}}} = \mathrm{vec}(Q_2) = \sum_{x,y \in \{0,1\}^{m-1}} \ket{x,x,y,y}, \label{eq:food}\\
    \ket{\varphi_{M_{13}}} = \mathrm{vec}(Q_3) = \sum_{x,y} \ket{x,y,x,y},  \quad
    \ket{\varphi_{M_{14}}} = \mathrm{vec}(Q_4) = \sum_{x,y} \ket{x,y,y,x};\\
    \Ggp{m-1} = \SOgp{2^{m-1}} \implies \calM = \{M_{12}, M_{13}, M_{14}\}, \qquad \Ggp{m-1} = \SUgp{2^{m-1}}  \implies \calM = \{M_{13}, M_{14}\}. \label{eq:pinky}
\end{gather}
Let us further define
\begin{equation}
  \ket{\psi_{10}} = \sum_{x \in \{0,1\}^{m-1}} \ket{x,x,x,x} \quad \text{and} \quad  \ket{\psi_{1j}} = \ket{\varphi_{M_{1j}}} - \ket{\psi_{10}},
\end{equation}
so that the $\ket{\psi_{1j}}$'s are pairwise orthogonal, with $\braket{\psi_{10}|\psi_{10}} = D$ and $\braket{\psi_{1j}|\psi_{1j}} = D(D-1)$  for $j > 1$. 
Then, since from \Cref{eq:WS} we have
\begin{equation}
  \mathrm{vec}(S) = \sum_{\substack{x = (x',a) \in \{0,1\}^{m-2} \times \{0,1\}\\y = (y',b) \in \{0,1\}^{m-2} \times \{0,1\}}} \ket{(x',a),(y',b),(y',a),(x',b)},
\end{equation}
we can easily compute
\begin{equation}
  \bra{\psi_{10}}\mathrm{vec}(S) = D, \quad \bra{\psi_{12}}\mathrm{vec}(S) = 0, \quad 
  \bra{\psi_{13}}\mathrm{vec}(S) = D, \quad 
  \bra{\psi_{14}}\mathrm{vec}(S) = D(D/2-1).
\end{equation}
From this we conclude that the projection of $\mathrm{vec}(S)$ onto the span of the four $\ket{\psi_{1j}}$'s (which is also the span of $\ket{\psi_{10}}$ and the three $\ket{\varphi_{M_{1j}}}$'s) 
is
\begin{equation}
  \ket{\sigma} \coloneqq \ket{\psi_{10}} 
  + \frac{1}{D-1} \ket{\psi_{13}} + \frac{D/2-1}{D-1} \ket{\psi_{14}}.
\end{equation}
Now one may easily verify that the following vector $\ket{\tau}$ is orthogonal to each $\ket{\varphi_{M_{1j}}} = \ket{\psi_{10}} + \ket{\psi_{1j}}$:
\begin{equation}
  \ket{\tau} = -(D-1)\ket{\psi_{10}} + \ket{\psi_{12}} + \ket{\psi_{13}} + \ket{\psi_{14}}.
\end{equation}
Thus we can bring $\ket{\sigma}$ into the span of the three $\ket{\varphi_{M_{1j}}}$'s by adding a suitable multiple of~$\ket{\tau}$ to zero out the $\ket{\psi_{10}}$ component as follows:
\begin{align}
  \ket{\sigma} + c\ket{\tau} &= (1-(D-1)c)\ket{\psi_{10}} + c\ket{\psi_{12}} + \parens*{\frac{1}{D-1} + c} \ket{\psi_{13}}
  + \parens*{\frac{D/2-1}{D-1} + c} \ket{\psi_{14}}\\
  &= \parens*{\frac{D/2-1}{D-1} - (D+2)c}\ket{\psi_{10}} + c \ket{\varphi_{M_{12}}} + \parens*{\frac{1}{D-1} + c}\ket{\varphi_{M_{13}}} + \parens*{\frac{D/2-1}{D-1} + c}\ket{\varphi_{M_{14}}},
\end{align}
and taking $c = \frac{D/2-1}{(D-1)(D+2)}$ we finally get that 
\begin{equation}
  \text{the projection of $\mathrm{vec}(S)$ onto the span of } \ket{\varphi_{M_{12}}}, \ket{\varphi_{M_{13 }}}, \ket{\varphi_{M_{14}}} \text{ is }  
  c_2 \ket{\varphi_{M_{12}}} + c_3 \ket{\varphi_{M_{13}}} + c_4 \ket{\varphi_{M_{14}}}.
\end{equation}
One can repeat the above using $\ket{\tau'} = -(D-1) \ket{\psi_{10}} + \ket{\psi_{13}} + \ket{\psi_{14}}$ in place of~$\ket{\tau}$ to similarly deduce
\begin{equation}
  \text{the projection of $\mathrm{vec}(S)$ onto the span of } \ket{\varphi_{M_{13 }}}, \ket{\varphi_{M_{14}}} \text{ is }  c'_3 \ket{\varphi_{M_{13}}} + c'_4 \ket{\varphi_{M_{14}}}.
\end{equation}
  The proof is complete. 
\end{proof}

Now we compute:
\begin{equation}
  \la A \otimes A^\dagger, \mathrm{SWAP}_{1,1'} \cdot Q_4 \ra = 
  \la A \otimes A^\dagger, \mathrm{SWAP}_{[m],[m]'} \ra = \|A\|_{\Fro}^2
\end{equation}
(similar to \Cref{eqn:s1}), and
\begin{equation}
  \la A \otimes A^\dagger, \mathrm{SWAP}_{1,1'} \cdot Q_3 \ra = 
  \la A \otimes A^\dagger, \mathrm{SWAP}_{1,1'} \cdot \Id_{[m]\setminus 1, [m]' \setminus 1'} \ra = \|\tr_{[m] \setminus 1} A\|_{\Fro}^2 \geq 0
\end{equation}
(similar to \Cref{eqn:s2}).
Finally, since (as can be easily verified)
  \begin{equation}
    \mathrm{SWAP}_{1,1'} \cdot Q_2 = \sum_{\substack{a,b \in \{0,1\} \\ x,y \in \{0,1\}^{m-1}}} \ket{(a,x), (b,x)}\!\bra{(b,y),(a,y)},
  \end{equation}  
we may conclude that
  \begin{equation}
    \la A \otimes A^\dagger, \mathrm{SWAP}_{1,1'} \cdot Q_2 \ra
    = \sum_{a,b,x,y} \braket{(b,y)|A^\dagger|(a,x)} \braket{(a,y)|A|(b,x)} \geq -\|A\|_{\Fr}^2
  \end{equation}
by Cauchy--Schwarz.
Putting these conclusions together with \Cref{eqn:putitin} and \Cref{prop:reptheory}, we get that for both $\Ggp{m-1} = \SOgp{2^{m-1}}$ and $\Ggp{m-1} = \SUgp{2^{m-1}}$ it holds that
\begin{equation}
  \E_{\bg \sim \Ggp{m-1}}[\norm{\tr_m((\Id \otimes \bg)A(\Id \otimes \bg^{-1}))}_{\Fr}^2] \geq (c_4 - c_2) \|A\|_{\Fr}^2 = \frac{D/2-1}{D-1} \|A\|_{\Fr}^2 = \frac{1 - 2^{2-m}}{2 - 2^{2-m}}\|A\|_{\Fr}^2,
\end{equation}
completing the proof of \Cref{lem:yummy2}.

\ignore{

\section{Older version of previous section}

In this section we prove \Cref{thm:small-m}, restated below:

\begin{theorem} [Restatement of \Cref{thm:small-m}] \label{thm:small-m-restatement}
For any $m \geq  3,$ we have that
\begin{equation} \label{eq:small-m-tau-lower-bound}
        \forall k \in \N^+,  \quad \quad \quad  \opnorm{\E_{\SOgp{2^{m-1}}\times\binom{[m]}{m - 1}}[\rho_{2^m}^{k,k}(\bg)] - \E_{\SOgp{2^m}}[\rho_{2^m}^{k,k}(\bg)]} \leq 
1 - \left(1 - {\frac 1 m}\right){\frac {2^{2m-3}+2^{m-2} - 3}{2^{2m-1} + 2^m - 4} }\end{equation}
(equivalently, in the notation of \Cref{thm:small-m}, $
\tau_m \geq \left(1 - {\frac 1 m}\right){\frac {2^{2m-3}+2^{m-2} - 3}{2^{2m-1} + 2^m - 4} }$).
\end{theorem}

\subsection{Setup}

\begin{definition} \label{def:coupling}
A pair of jointly distributed random variables $(\bX,\bY)$ is a \emph{coupling} of probability distributions $\nu_1,\nu_2$ if $\bX$ (respectively,~$\bY$) has marginal distribution $\nu_1$ (respectively, $\nu_2$).
\end{definition}

\begin{definition} \label{def:wasserstein-old}
The \emph{$L^p$-Wasserstein distance}, with respect to distance function~$d$, between two distributions $\nu_1$ and $\nu_2$  is 
\begin{equation} \label{eq:lpwasserstein}
W_{d,p}(\nu_1,\nu_2) := \inf \left\{
\E[d(\bX,\bY)^p]^{1/p} \ : \ (\bX,\bY)\text{~is a coupling of~}(\nu_1,\nu_2)
\right\}.
\end{equation}
\end{definition}

We will be concerned with distributions over $\SOgp{2^m}$, and we will work with two different distance measures over $\SOgp{2^m}$:

\begin{definition} [Frobenius distance] \label{def:frob}
Given two matrices $A,B \in \SOgp{N}$ the \emph{Frobenius distance}, denoted $d_{\Fro}(A,B)$ or $\|A-B\|_F$,  is 
$\left(\sum_{i,j} |A_{ij}-B_{ij}|^2\right)^{1/2}.$ (We recall also the fact that $\|A\|_F=\sqrt{\tr(A A^\dagger)}$.)\rasnote{Will we have defined/recalled all this earlier? Not sure}
\end{definition}

\begin{definition} [Riemannian distance] \label{def:riem}
Given two matrices $A,B \in \SOgp{N}$ the \emph{Riemannian distance} between $A$ and $B$, denoted $\dRie(A,B)$, is the length of a geodesic curve between $A$ and $B$ where the curve is constrained to lie in $\SOgp{N}$, i.e.~
\begin{equation} \label{eq:riemannian-distance}
\dRie(A,B) = \min
\left\{
\int_0^1 \|\gamma'(t)\|_F dt :  \gamma(0)=A, \gamma(1)=B, 
\gamma(t) \in \SOgp{N} \text{~for all~}t \in [0,1]
\right\}.
\end{equation}
\end{definition}

It is immediate that $d_{\Fro}(A,B) \leq \dRie(A,B)$ and hence for any distributions $\nu_1,\nu_2$ over $\SOgp{2^m}$ we have $W_{\Fro,2}(\nu_1,\nu_2) \leq W_{\Rie,2}(\nu_1,\nu_2).$ We further recall that for small distances, Frobenius and Riemannian distance are the same up to first order; more precisely, we recall Equation~(10) of \cite{Oliveira09}, which states that for all $Z,W \in \SOgp{2^m}$, we have
\begin{equation}
\label{eq:equation-10-of-Oliveira}
\dRie(W,Z) \leq d_{\Fro}(W,Z) + C_m d_{\Fro}(W,Z)^2,
\end{equation}
where $C_m$ is a constant depending only on $m.$
\medskip

\noindent {\bf Notation.}
We write $\nu^{\ast \ell}$ to denote the $\ell$-fold convolution of a distribution $\nu$ on $\SOgp{N}$, i.e.
\begin{equation} \label{eq:convolution}
\nu^{\ast \ell} = \int \delta_{O_1 \dots O_{\ell}} \nu(\mathrm{d}O_1) \dots \nu(\mathrm{d}O_\ell).
\end{equation}
We write $\delta_Z$ to denote  a point-mass distribution at $Z \in \SOgp{N}$. 
For conciseness we write ``$\nu^{\aux}_m$'' to denote the distribution $\SOgp{2^{m-1}} \times {[m] \choose m-1}.$

\subsection{High-level approach}

The general approach follows  \cite[Lem.~7]{HH21}, which in turn follows \cite{BHH16}. In this subsection we describe the main ingredients and explain how they yield \Cref{thm:small-m-restatement}.

\begin{enumerate}

\item We first prove a Frobenius distance statement about how two steps of the random walk $\nu^{\aux}_{m}$ can be coupled (this is where the bulk of the work takes place):

\begin{lemma} \label{lem:eq-90-BHH-analogue-Frobenius}
Let $\eps > 0$ and fix $X,Y \in \SOgp{2^m}$ that satisfy $\dRie(X,Y) = \eps.$
There is a coupling of the two $\SOgp{2^m}$-valued random variables $\bX',\bY'$, where $\bX'$ is individually distributed as $(\nu_{\aux})^{\ast 2} \ast \delta_X$
and
$\bY'$ is individually distributed as $(\nu_{\aux})^{\ast 2} \ast \delta_Y$,
that satisfies
\begin{equation}
\label{eq:equation-90-BHH-Frobenius}
\E[d_{\Fro}(\bX',\bY')^2]
\leq \eta' \cdot d_{\Fro}(X,Y)^2 + O_{m}(\eps^3),
\quad \quad \text{where~}\eta' := 1 - \left(1 - {\frac 1 m}\right){\frac {2^{2m-3}+2^{m-2} - 3}{2^{2m-1} + 2^m - 4} }.
\end{equation}
Moreover, for every outcome $X',Y'$ of the random variables $\bX',\bY',$ we have that
\begin{equation} \label{eq:max-Frobenius-distance}
d_{\Fro}(X',Y')^2 \leq \eps^2 + O_m(\eps^3).
\end{equation}
\end{lemma}

%
%
%

\item We use \Cref{lem:eq-90-BHH-analogue-Frobenius} to show that
taking two steps of the random walk on $\SOgp{2^m}$ associated with $\SOgp{2^{m-1}} \times {[m] \choose m-1}$ vis-a-vis $\rho^{k,k}_{2^m}$ is a ``local contraction'' for $L^2$-Wasserstein distance under the Riemannian distance. The formal statement is as follows:

\begin{lemma} \label{lem:lemma-8-HHJ-analogue}
We have 
\begin{equation} \label{eq:equation-85-BHH-analogue}
\limsup_{\eps \to 0} \sup_{X,Y \in \SOgp{2^m}} 
\left\{
{\frac {W_{\Rie,2}((\nu^{\aux}_m)^{\ast 2} \ast \delta_X,(\nu^{\aux}_m)^{\ast 2} \ast \delta_Y)}{\dRie(X,Y)}} : 
\dRie(X,Y) \leq \eps
\right\} 
\leq \eta'.
\end{equation}
\end{lemma}

\item \Cref{lem:lemma-8-HHJ-analogue} is useful for us by virtue of the following result of Oliveira, which shows that a Riemannian $L^2$-Wasserstein local contraction yields a global contraction:

\begin{theorem} [Theorem~3 of \cite{Oliveira09}] \label{thm:oliveira-old}
Let $\nu$ be a probability measure on $\SOgp{2^m}$ such that
\begin{equation} \label{eq:oliveira-hypothesis}
\limsup_{\eps \to 0} \sup_{X,Y \in \SOgp{2^m}} 
\left\{
{\frac {W_{\Rie,2}(\nu \ast \delta_X,\nu \ast \delta_Y)}{\dRie(X,Y)}} : 
\dRie(X,Y) \leq \eps
\right\} 
\leq \eta.
\end{equation}
Then for all distributions $\nu_1,\nu_2$ on $\SOgp{2^m}$, we have
\begin{equation} \label{eq:olveira-conclusion}
W_{\Rie,2}(\nu \ast \nu_1,\nu \ast \nu_2) \leq \eta \cdot W_{\Rie,2}(\nu_1,\nu_2).
\end{equation}
\end{theorem}

\end{enumerate}

With the above ingredients in hand we prove \Cref{thm:small-m-restatement} as follows:

\medskip

\noindent \emph{Proof of \Cref{thm:small-m-restatement} using \Cref{lem:lemma-8-HHJ-analogue}  and \Cref{thm:oliveira-old}.}
By \Cref{lem:lemma-8-HHJ-analogue} we may apply
 \Cref{thm:oliveira-old} to $\nu_1 = \delta_I$ (here $I$ is the identity element of $\SOgp{2^m}$) and $\nu_2 = \SOgp{2^m}$ (the Haar distribution on $\SOgp{2^m}$), and we get that 
\begin{equation} \label{eq:pickle}
W_{\Rie,2}((\nu^{\aux}_m)^{\ast 2} \ast \delta_I,(\nu^{\aux}_m)^{\ast 2} \ast \SOgp{2^m}) \leq \eta' \cdot W_{\Rie,2}(\delta_I,\SOgp{2^m}).
\end{equation}
Iteratively combining $\ell$ applications of  \Cref{thm:oliveira-old} in this way, we get that for any $\ell \geq 1$,
\begin{equation} \label{eq:radish}
W_{\Rie,2}((\nu^{\aux}_m)^{\ast 2\ell} \ast \delta_I,(\nu^{\aux}_m)^{\ast 2\ell} \ast \SOgp{2^m}) \leq \eta'^\ell \cdot W_{\Rie,2}(\delta_I,\SOgp{2^m}).
\end{equation}
Since $\SOgp{2^m}$ has diameter at most $\pi\cdot 2^{m/2}$ under $\dRie$ (see e.g.~the end of the proof of Theorem~1 of \cite{Oliveira09}; in fact any finite bound as a function of $m$ would be sufficient for our purposes), we have that $W_{\Rie,2}(\delta_I,\SOgp{2^m}) \leq \pi\cdot 2^{m/2}$. Since $(\nu^{\aux}_m)^{\ast 2\ell} \ast \delta_I=(\nu^{\aux}_m)^{\ast 2\ell}$ and
$(\nu^{\aux}_m)^{\ast 2\ell} \ast \SOgp{2^m}=\SOgp{2^m}$, we get that
\begin{equation} \label{eq:equation-51-HHJ-analogue}
W_{\Rie,2}((\nu^\aux_{m})^{\ast 2 \ell}, \SOgp{2^m})
\leq
\eta'^\ell \cdot \pi\cdot 2^{m/2}.
\end{equation}

Next, we recall the following result from \cite{BHH16} (where it is proved for the unitary group) which we will prove later:
\begin{lemma} [Lemma~20 from \cite{BHH16}]
\label{lem:equation-64-HHJ}
For any distribution $\nu$ over $\SOgp{2^m}$, we have that  
\begin{equation}
\label{eq:equation-64-HHJ}
        \forall k \in \N^+,  \quad \quad \quad 
 \opnorm{\E_{\nu}[\rho_{2^m}^{k,k}(\bg)] - \E_{\SOgp{2^m}}[\rho_{2^m}^{k,k}(\bg)]}
\leq
2k W_{\Rie,2}(\nu,\SOgp{2^m}).
\end{equation}
\end{lemma}

Applying \Cref{lem:equation-64-HHJ} to \Cref{eq:equation-51-HHJ-analogue}, we get that for any
$\ell \geq 1$ we have
\begin{equation} \label{eq:kidney}
        \forall k \in \N^+,  \quad \quad \quad 
 \opnorm{\E_{(\nu^\aux_{m})^{\ast 2 \ell}}[\rho_{2^m}^{k,k}(\bg)] - \E_{\SOgp{2^m}}[\rho_{2^m}^{k,k}(\bg)]}
 \leq 2k \eta'^\ell \cdot \pi\cdot 2^{m/2}.
\end{equation}
Since $\E_{\SOgp{2^m}}[\rho_{2^m}^{k,k}(\bg)]$ is a projector, we have that
\begin{equation} \label{eq:want-this}
  \forall k \in \N^+,  \quad \quad \quad 
 \opnorm{\E_{\nu^{\aux}_{m}}[\rho_{2^m}^{k,k}(\bg)] - \E_{\SOgp{2^m}}[\rho_{2^m}^{k,k}(\bg)]}^\ell 
 = 
 \opnorm{\E_{(\nu^\aux_{m})^{\ast 2 \ell}}[\rho_{2^m}^{k,k}(\bg)] - \E_{\SOgp{2^m}}[\rho_{2^m}^{k,k}(\bg)]},
\end{equation}
and hence we get that
\begin{equation} \label{eq:want-this-too}
  \forall k \in \N^+,  \quad \quad \quad 
 \opnorm{\E_{\nu^{\aux}_{m}}[\rho_{2^m}^{k,k}(\bg)] - \E_{\SOgp{2^m}}[\rho_{2^m}^{k,k}(\bg)]}^\ell 
 \leq 2k \eta'^\ell \cdot \pi\cdot 2^{m/2}.
\end{equation}
Taking the $\ell$-th root of both sides and letting $\ell \to \infty$, we get that
\begin{equation} \label{eq:small-m-tau-lower-bound-restatement}
        \forall k \in \N^+,  \quad \quad \quad  \opnorm{\E_{\nu^{\aux}_{m}}[\rho_{2^m}^{k,k}(\bg)] - \E_{\SOgp{2^m}}[\rho_{2^m}^{k,k}(\bg)]} \leq 
\eta'
 \end{equation}
 which gives \Cref{eq:small-m-tau-lower-bound} as desired and completes the proof of \Cref{thm:small-m-restatement} (modulo the proof of \Cref{lem:equation-64-HHJ}).
\qed

\medskip \noindent \emph{Proof of \Cref{lem:equation-64-HHJ}.} We follow the proof from \cite{BHH16} with trivial modifications.
Recalling that $d_{\Fro} \leq \dRie$, it is enough to establish
\begin{equation} \label{eq:equation-126-BHH}
  \forall k \in \N^+,  \quad \quad \quad 
 \opnorm{\E_{\nu}[\rho_{2^m}^{k,k}(\bg)] - \E_{\SOgp{2^m}}[\rho_{2^m}^{k,k}(\bg)]}
\leq
2k W_{\Fro,2}(\nu,\SOgp{2^m}),
\end{equation} 
which is what we do below.

We first recall the variational characterization of the operator norm, i.e.~the fact that for any $M$ we have\rasnote{Earlier on somewhere, record that $\|\cdot\|_1$ denotes the Schatten 1-norm, maybe mention that Schatten $\infty$-norm is the same as operator norm, etc.}
$\opnorm{M} = \max_{X: \|X\|_1 \leq 1} \tr(MX).$  Hence we may let $X$ be such that $\|X\|_1 \leq 1$ satisfies
\begin{equation} \label{eq:lemon}
\tr\left(\left( \E_{\nu}[\rho_{2^m}^{k,k}(\bg)] - \E_{\SOgp{2^m}}[\rho_{2^m}^{k,k}(\bg)]\right)X\right) 
=
\opnorm{\E_{\nu}[\rho_{2^m}^{k,k}(\bg)] - \E_{\SOgp{2^m}}[\rho_{2^m}^{k,k}(\bg)]},
\end{equation}
and we may define
\begin{equation} \label{eq:f}
f(U) := \tr(U^{\otimes 2k} X).
\end{equation}
Later we will prove:
\begin{claim} \label{claim:lipschitz}
The function $f$ is $2k$-Lipschitz with respect to Frobenius distance.
\end{claim}
Given \Cref{claim:lipschitz}, we have that $f/(2k)$ is 1-Lipschitz, and hence for any $k \in \N^+$,
\begin{align} 
\opnorm{\E_{\nu}[\rho_{2^m}^{k,k}(\bg)] - \E_{\SOgp{2^m}}[\rho_{2^m}^{k,k}(\bg)]}
= 2k \left|
\E_{\nu} [f(\bg)/(2k)] - 
\E_{\SOgp{2^m}}[f(\bg)/(2k)] \right | 
&\leq 2k W_{\Fro,1}(\nu,\SOgp{2^m})\label{eq:equation-129-BHH}\\
&\leq 2k W_{\Fro,2}(\nu,\SOgp{2^m}), \label{eq:equation-129-BHH-plus}
\end{align}
where \Cref{eq:equation-129-BHH-plus} holds because $W_{\Fro,1} \leq W_{\Fro,2}$ and \Cref{eq:equation-129-BHH} holds by the Kantorovich-Rubenstein-duality interpretation of Wasserstein $L^1$-distance (see \cite{KR58,Edwards11}), namely
\begin{equation}
\label{eq:kantorovich}
W_{\Fro,1}(\nu_1,\nu_2) =
\sup\left\{
\E_{\nu_1}[f(\bg)] - \E_{\nu_2}[f(\bg)] \ : \ f \text{~is~1-Lipschitz}
\right\}.
\end{equation}
This concludes the proof of \Cref{lem:equation-64-HHJ} modulo the proof of \Cref{claim:lipschitz}.
\qed

\medskip
\noindent \emph{Proof of \Cref{claim:lipschitz}.}
We observe that if $A,B,C,D \in \SOgp{2^m}$ then we have
\begin{align} 
\opnorm{A\otimes B - C \otimes D} &= 
\opnorm{A \otimes (B-D) + (A-C) \otimes D} \nonumber\\
&\leq \opnorm{A\otimes(B-D)} + \opnorm{(A-C) \otimes D} \nonumber\\
&\leq \opnorm{A} \opnorm{B-D} + \opnorm{D} \opnorm{A-C} \nonumber\\
&\leq \opnorm{B-D} + \opnorm{A-C}, \label{eq:pigeon}
\end{align}
where the last line used the fact that since $A,D \in \SOgp{2^m}$ they each have operator norm at most 1.
Iteratively applying \Cref{eq:pigeon}, we get that for $A,B \in \SOgp{2^m}$ we have
\begin{equation}
\label{eq:tensor-opnorm}
\opnorm{A^{\otimes 2k} - B^{\otimes 2k}} \leq 2k\opnorm{A-B}.
\end{equation}
Hence we have
\begin{align} 
|f(A)-f(B)| &=
|\tr((A^{\otimes 2k} - B^{\otimes 2k})X)| \tag{definition of $f$}\nonumber\\
&\leq \|X\|_1 \opnorm{A^{\otimes 2k} - B^{\otimes 2k}} \tag{using $\tr(A^\dagger B) \leq \|A\|_1 \|B\|_\infty$}\nonumber\\
&\leq \opnorm{A^{\otimes 2k} - B^{\otimes 2k}} \tag{since $\|X\|_1 \leq 1$}\nonumber\\
&\leq 2k \opnorm{A - B} \tag{by \Cref{eq:tensor-opnorm}} \nonumber\\
&\leq 2k \cdot d_{\Fro}(A,B). \tag{Frobenius norm dominates operator norm}
\end{align}
This conclude the proof of \Cref{claim:lipschitz}.
\qed

%
%
%
%

\subsubsection{Proof of \Cref{lem:lemma-8-HHJ-analogue} using \Cref{lem:eq-90-BHH-analogue-Frobenius}.}

Let $\eps>0$ and let $\bX',\bY'$ be the coupled random variables from \Cref{lem:eq-90-BHH-analogue-Frobenius}. Recall from \Cref{eq:max-Frobenius-distance} that for every outcome $(X',Y')$ of $(\bX',\bY')$, we have
\begin{equation}
\label{eq:jello}
d_{\Fro}(X',Y') \leq \eps + O_m(\eps^{3/2}).
\end{equation}
We have that
\begin{align} 
\E[\dRie(\bX',\bY')^2] &\leq \E[d_{\Fro}(\bX',\bY')^2 + 2C_m d_{\Fro}(\bX',\bY')^3
+ C_m^2 d_{\Fro}(\bX',\bY')^4] \label{eq:flood}\\
&\leq \E[d_{\Fro}(\bX',\bY')^2 (1 + O_m(\eps) + O_m(\eps^{3/2}) + O_m(\eps^2) + O_m(\eps^3))] \label{eq:deluge}\\
&\leq (1 + O_m(\eps)) \cdot \left(\eta' \cdot d_{\Fro}(X,Y)^2 + O_{m}(\eps^3)\right) \label{eq:downpour}\\
&\leq (1 + O_m(\eps)) \cdot \eta' \cdot \dRie(X,Y)^2 
+ O_{m}(\eps^3)\label{eq:hardrain},
\end{align}
where \Cref{eq:flood} is by \Cref{eq:equation-10-of-Oliveira}, \Cref{eq:deluge} is by \Cref{eq:max-Frobenius-distance} and \Cref{eq:jello}, and \Cref{eq:downpour} is by \Cref{eq:equation-90-BHH-Frobenius}, and 
\Cref{eq:hardrain} is because $d_{\Fro} \leq \dRie.$
Consequently we have
\begin{equation} \label{eq:gonnafall}
{\frac {\E[\dRie(\bX',\bY')^2]}{\dRie(X,Y)^2}}
\leq (1 + O_m(\eps)) \cdot \eta' + 
{\frac {O_{m}(\eps^3)}{\dRie(X,Y)^2}}
=
(1 + O_m(\eps)) \cdot \eta' + 
O_{m}(\eps),
\end{equation}
where the last equation is because $\dRie(X,Y)=\eps$. Letting $\eps \to 0$ we get \Cref{eq:equation-85-BHH-analogue}, and \Cref{lem:lemma-8-HHJ-analogue} is proved. \qed

\subsection{Proof of \Cref{lem:eq-90-BHH-analogue-Frobenius}.}

Recall that $X$ and $Y$ are two fixed elements of $\SOgp{2^m}$ that satisfy $\dRie(X,Y) = \eps>0$, and that \Cref{lem:eq-90-BHH-analogue-Frobenius} analyzes a coupled pair $(\bX',\bY')$ where the individual distribution of $\bX'$ corresponds to taking two steps of the $\nu_\aux$ random walk starting from $\delta_X$ (and likewise for the individual distribution of $\bY'$).

After two steps of the $\nu_\aux$ random walk starting from $\delta_X$, the distribution of possible outcomes is a uniform mixture of the $m^2$ distributions
\begin{equation}
\label{eq:two-steps-of-X}
\left\{
\tilde{\bO}_{[1,i-1] \cup [i+1,m]}
\bO_{[1,j-1] \cup [j+1,m]} X
\right\}_{i,j\in [m]},
\end{equation}
where $\tilde{\bO}$ and $\bO$ are independent draws from $\SOgp{2^{m-1}}$ and the notation $\tilde{\bO}_S,\bO_S$, for $S \in {[m] \choose m-1}$, is as defined in \Cref{not:gatepositioning}.
The same holds for $Y$:  after two steps of the $\nu_\aux$ random walk starting from $\delta_Y$, the distribution of possible outcomes is a uniform mixture of the $m^2$ distributions
\begin{equation}
\label{eq:two-steps-of-Y}
\left\{
\tilde{\bO}_{[1,i-1] \cup [i+1,m]}
\bO_{[1,j-1] \cup [j+1,m]} Y
\right\}_{i,j\in [m]}.
\end{equation}
To create a nontrivial coupling, we consider the following distribution of outcomes after two steps of the random walk from $\delta_X$:
\begin{equation}
\label{eq:two-steps-of-X-coupled}
\bX' \text{~is a uniform mixture of~}
\left\{
\tilde{\bO}_{[1,i-1] \cup [i+1,m]}
V^{i,j}_{[1,i-1] \cup [i+1,m]}
\bO_{[1,j-1] \cup [j+1,m]} X
\right\}_{i,j\in [m]}
\end{equation}
where $V^{i,j}$ belongs to $\SOgp{2^{m-1}}$ (so $V^{i,j}_{[1,i-1] \cup [i+1,m]}$ is an element of $\SOgp{2^m}$ that may act on all $m$ qubits except the $i$-th one) and may depend on $\bO_{[1,j-1] \cup [j+1,m]}$ (and on $X$ and $Y$).  The distribution $\bY'$ is as described above, i.e.~a uniform mixture of the $m^2$ distributions in \Cref{eq:two-steps-of-Y}. Because of the final multiplication by $\tilde{\bO}_{[1,i-1] \cup [i+1,m]}$ in \Cref{eq:two-steps-of-X-coupled}, the joint distribution defined in \Cref{eq:two-steps-of-Y} and \Cref{eq:two-steps-of-X-coupled} is a valid coupling for any $V^{i,j}_{[1,i-1] \cup [i+1,m]}$ as stated above; the crux of the argument is choosing and analyzing a suitable $V^{i,j}_{[1,i-1] \cup [i+1,m]}$, which we do below.

\subsubsection{ The choice of $V^{i,j}_{[1,i-1] \cup [i+1,m]}$.}

Recall from \Cref{eq:equation-90-BHH-Frobenius} that our goal is to upper bound the quantity $\E[d_{\Fro}(\bX',\bY')^2]=\E[\norm{\bX'-\bY'}_F^2]$. We have
\begin{align}
&\E[\norm{\bX'-\bY'}_F^2]\nonumber \\
&=
{\frac 1 {m^2}} \sum_{i,j \in [m]}
\E\left[
\norm{\tilde{\bO}_{[1,i-1] \cup [i+1,m]}V^{i,j}_{[1,i-1] \cup [i+1,m]} \bO_{[1,j-1] \cup [j+1,m]} X
-
\tilde{\bO}_{[1,i-1] \cup [i+1,m]} \bO_{[1,j-1] \cup [j+1,m]} Y}_F^2\right].
\label{eq:HHJ-55}
\end{align}
When $i=j$ we will take $V^{i,i}_{[1,i-1] \cup [i+1,m]}$ to be the identity matrix $\Id$,  so we get that $\E[\norm{\bX'-\bY'}_F^2]$ is equal to
\begin{align} 
&
{\frac 1 {m^2}}
\sum_{i\in [m]}
\E\left[
\norm{\tilde{\bO}_{[1,i-1] \cup [i+1,m]} \bO_{[1,j-1] \cup [j+1,m]} (X-Y)}_F^2\right] \nonumber\\
&+
{\frac 1 {m^2}} \sum_{i \neq j \in [m]}
\E\left[
\norm{\tilde{\bO}_{[1,i-1] \cup [i+1,m]}V^{i,j}_{[1,i-1] \cup [i+1,m]} \bO_{[1,j-1] \cup [j+1,m]} X
-
\tilde{\bO}_{[1,i-1] \cup [i+1,m]} \bO_{[1,j-1] \cup [j+1,m]} Y}_F^2\right].
\label{eq:falafel}
\end{align}
Since $\norm{M}_F^2  =  \norm{LM}_F^2$ for any $2^m \times 2^m$ matrix $M$ and any $L \in \SOgp{2^m}$, we have that \Cref{eq:falafel} is equal to
\begin{align} 
&
{\frac 1 m}
\norm{X-Y}_F^2 \nonumber\\
&+
{\frac 1 {m^2}} \sum_{i \neq j \in [m]}
\E\left[
\norm{V^{i,j}_{[1,i-1] \cup [i+1,m]} \bO_{[1,j-1] \cup [j+1,m]} X
-
 \bO_{[1,j-1] \cup [j+1,m]} Y}_F^2\right].
\label{eq:hummus}
\end{align}
Now we commit to our choice of $V^{i,j}\in \SOgp{2^{m-1}}$ when $i \neq j$: we take each such $V^{i,j}$ to be the element of $\SOgp{2^{m-1}}$ such that $V^{i,j}_{[1,i-1] \cup [i+1,m]}$ minimizes the quantity inside the expectation in \Cref{eq:hummus}. So \Cref{eq:hummus} is equal to
\begin{align}
&{\frac 1 m}
\norm{X-Y}_F^2 \nonumber \\
& +
{\frac 1 {m^2}} \sum_{i \neq j \in [m]}
\E\left[
\min_{V^{i,j} \in \SOgp{2^{m-1}}}
\norm{V^{i,j}_{[1,i-1] \cup [i+1,m]} \bO_{[1,j-1] \cup [j+1,m]} X
- 
 \bO_{[1,j-1] \cup [j+1,m]} Y}_F^2
\right]. \label{eq:tabbouleh}
\end{align}

\subsubsection{An explicit proxy for each minimizer.}
Let us focus in on one of the $m(m-1)$ summands in the above sum (the analysis will generalize straightforwardly to the other summands as well). It will be notationally easiest for us to consider $i=1,j=m$.  In this case the summand we would like to bound is

\begin{equation} \label{eq:i-is-1-j-is-m}
\E\left[
\min_{V^{1,m} \in \SOgp{2^{m-1}}}
\norm{V^{1,m}_{[2,m]} \bO_{[1,m-1]} X
-
 \bO_{[1,m-1]} Y}_F^2
\right].
\end{equation}
We note that for any two special orthogonal matrices $O_1, O_2$ we
have
\begin{equation} \label{eq:SO-factoid}
\norm{O_1 - O_2}^2_F = \tr \parens*{(O_1 - O_2)(O_1 - O_2)^T} = \tr \parens*{O_1O_1^T + O_2O_2^T - O_2O_1^T - O_1O_2^T} = 2(\tr(\Id) - \tr(O_1O_2^T)),
\end{equation}
where we used the fact that $O_iO_i^T = \Id$ and $\tr(O_2O_1^T) =
\tr((O_2O_1^T)^T) = \tr(O_1O_2^T)$. Applying this to \Cref{eq:i-is-1-j-is-m}, we get that for any fixed outcome $O_{[1,m-1]}$ of the random variable $\bO_{[1,m-1]}$, we have
\begin{equation}
\label{eq:shakshuka}
\min_{V^{1,m}}
\norm{V^{1,m}_{[2,m]} O_{[1,m-1]} X
-
 O_{[1,m-1]} Y}_F^2
 =
 2\parens*{\tr(\Id) - \min_{V^{1,m}} \tr\parens*{V^{1,m}_{[2,m]}O_{[1,m-1]}XY^TO_{[1,m-1]}^T}}.
 \end{equation}

We control \Cref{eq:shakshuka} by describing an explicit valid choice for $V^{1,m}_{[2,m]}$, which we denote\rasnote{Not sure yet how our notational dust will settle --- if this $V'$ notation will be confusing vis-a-vis transposing/conjugating/whatever, rename} $V'$, and using the fact that the minimizer does at least as well as our explicit choice. The following notation will help us describe the transformation $V'$:  for $Z \in \SOgp{2^m}$ and $O$ an element of $\SOgp{2^{m-1}}$, we write $\conj_O(Z)$ to denote
\begin{equation} \label{eq:OZ}
\conj_O(Z) := O_{[1,m-1]}ZO_{[1,m-1]}^T,
\end{equation}
and we write $T_O(Z)$ to denote
\begin{equation} \label{eq:TZ}
T_O(Z) := \tr_1(\conj_O(Z)),
\end{equation}
where $\tr_1$ is the partial trace where the first qubit is traced out.
We will also use the facts that (i) if $H$ is a skew-symmetric matrix (meaning $H^T=-H$) then $\exp(H)$ is a special orthogonal matrix, and conversely (ii) every special orthogonal matrix can be written as $\exp(H)$ for some
skew-symmetric matrix.

Let $R$ denote the special orthogonal matrix $R=XY^T$; from (ii) we may write $R=\exp(M)$ for some skew-symmetric matrix $M$.
We recall the fact (Equation~1.30 of \cite{Gan-thesis}) that for special orthogonal matrices $A, B$
that satisfy $\norm{B^{-1}A - \Id}_F < 1$ we have $\dRie(A, B) =
\norm{\log B^{-1}A}_F$.\footnote{Here $\log$ denotes the matrix logarithm:  for $A$ such that $\norm{A-\Id}_F<1$, $\log A$ is defined as
\[
\log A = \sum_{k=1}^\infty (-1)^{k+1} {\frac {(A-\Id)^k}{k}}.
\]
}
In our setting we have that 
$\norm{X^{-1} Y - \Id}_F \leq \norm{Y-X}_F \norm{X^{-1}}_F \leq \norm{Y-X}_F=d_{\Fro}(X,Y) \leq \dRie(X,Y) = \eps,$ so we may apply this fact to get that
\begin{equation} \label{eq:pancake}
\dRie(X,Y)=\norm{\log Y^{-1} X}_F=\norm{\log Y^T X}_F = \norm{M^T}_F = \norm{M}_F = \eps.
\end{equation}
Since $R$ is special orthogonal we may write it as $R=\exp(\eps H)$ where $H$ is skew-symmetric and $\eps H=M$ (since $R=\exp(M)$), and by \Cref{eq:pancake} we have that $\norm{H}_F = 1.$

Now we define
\begin{equation} \label{eq:def-of-S}
S := \exp(\eps c T_O(H)),
\end{equation}
where we will fix the (scalar) value $c$ later; observe that $T_O(H)$, and hence also $S$, acts only on qubits 2 through $m$.  We define $V'$ to be the transformation
\begin{equation}
\label{eq:def-of-Vprime}
V' := S_{[2,m]}.
\end{equation}
(i.e.~$V'$ is obtained by viewing $S$ as an element of $\SOgp{2^m}$ acting on qubits $2,\dots,m$).

In order for \Cref{eq:def-of-Vprime} to be a valid choice, we must show that it belongs to $\SOgp{2^m}$, which is equivalent to showing that $S$ is special orthogonal; this in turn is equivalent to showing that $T_O(H)$ is skew-symmetric. We have
\begin{align} 
  T_O(H)^T &= \parens*{\tr_1 \conj_O(H)}^T = \tr_1 \parens*{\parens*{O_{[1,m-1]}HO_{[1,m-1]}^T}^T} = \tr_1 \parens*{O_{[1,m-1]}H^TO_{[1,m-1]}^T} \nonumber \\
  &= -\tr_1 \parens*{O_{[1,m-1]}HO_{[1,m-1]}^T} = -T_O(H), \label{eq:T-of-H-is-skew-symmetric}
\end{align}
so $T_O(H)$ is indeed skew-symmetric.

\subsubsection{Bounding the explicit proxy for a minimizer.}

With $V'$ in hand, referring back to \Cref{eq:shakshuka}, we would like to analyze the quantity
\begin{equation} \label{eq:lasagna}
\tr\parens*{V' \conj_O(R)}.
\end{equation}
Let us write\rasnote{check if this notation is defined earlier} $\angles{A,  B}$ to denote $\tr(AB^T).$ We observe that
\begin{equation} \label{eq:cavatelli}
  \tr\parens*{V'\conj_O(R)} = \angles*{(S \otimes \Id_1)^T, \conj_O(R)} =
  \angles*{(\exp(-\eps c T_O(H)) \otimes \Id_1, \conj_O(R)} =
  \angles*{\exp(-\eps c T_O(H)), T_O(R)},
\end{equation}
where the middle equality uses $T_O(H)^T=-T_O(H)$ and the last equality follows from the definition of tracing out the first qubit.
By a Taylor expansion of the matrix exponential in the RHS of \Cref{eq:cavatelli}, we may write
\begin{equation} \label{eq:taylor-expansion-1}
\exp(-\eps c T_O(H)) =
\Id_{[2,m]} - \eps c T_O(H) + \frac{\eps^2}{2} (c T_O(H))^2 + O_m(\eps^3)
\end{equation}
where we abuse notation and write ``$O_m(\eps^3)$'' to denote some matrix $\Sigma$ that satisfies $\norm{\Sigma}_F = O_m(\eps^3)$.
We perform a similar Taylor expansion of $T_O(R)=T_O(\exp(\eps H))$ and get that
\begin{equation} \label{eq:taylor-expansion-2}
T_O(R) =
T_O(\Id)+ \eps T_O(H) + \frac{\eps^2}{2} T_O(H^2) + O_m(\eps^3) 
=
2 \Id_{[2,m]}+ \eps T_O(H) + \frac{\eps^2}{2} T_O(H^2) + O_m(\eps^3).
\end{equation}
Substituting \Cref{eq:taylor-expansion-1} and \Cref{eq:taylor-expansion-2} into \Cref{eq:cavatelli}, we get that
\begin{equation} \label{eq:gnocchi}
  \tr\parens*{V'\conj_O(R)} 
  =\angles*{
  \Id_{[2,m]} - \eps c T_O(H) + \frac{\eps^2}{2} (c T_O(H))^2 + O_m(\eps^3),
  2 \Id_{[2,m]}+ \eps T_O(H) + \frac{\eps^2}{2} T_O(H^2) + O_m(\eps^3)
  }.
  \end{equation}
  We expand the terms that arise in the above inner product as follows:
  
\begin{itemize}
\item $\angles*{\Id_{[2,m]}, 2\Id_{[2,m]}} = 2^m$;
\item $\angles*{\Id_{[2,m]}, \eps T_O(H)} = \eps \tr(H) = 0$ (since $H$ is
  skew-symmetric we  have $\tr(H)=-\tr(H)=0$);
\item $\angles*{\Id_{[2,m]}, \frac{\eps^2}{2} T_O(H^2)} =
\frac{\eps^2}{2} \tr(H^2) = 
-{\frac {\eps^2} 2} \tr(HH^T)= -{\frac {\eps^2} 2}\norm{H}_F^2=-{\frac {\eps^2} 2}$;
\item $\angles*{-\eps c T_O(H), 2\Id_{[2,m]}} = -2 \eps c \tr(H) = 0$;
\item 
$\angles*{-\eps c T_O(H), \eps T_O(H)} = -\eps^2 c \tr(T_O(H)^2)$;
\item $\angles*{\frac{\eps^2}{2} (c T_O(H))^2, 2\Id_{[2,m]}} = \eps^2 c^2 \tr(T_O(H)^2)$;
\end{itemize}
and all remaining terms can be collected in the $O_m(\eps^3)$ term. Combining these equalities, we have
\begin{equation} \label{eq:ravioli}
  \tr\parens*{V'\conj_O(R)} 
  =2^m - \frac{\eps^2}{2}  -\eps^2 c \tr(T_O(H)^2) + \eps^2 c^2 \tr(T_O(H)^2) + O_m(\eps^3).
  \end{equation}
  Since $T_O(H)$ is skew-symmetric $\tr(T_O(H)^2)$ is non-positive, so to maximize \Cref{eq:ravioli} we pick $c=1/2$, and we get that
\begin{equation} \label{eq:cavatappi}
  \tr\parens*{V'\conj_O(R)} 
  = 2^m - \frac{\eps^2}{2} \parens*{1 +  {\frac 1 2} \tr(T_O(H)^2)} + O_m(\eps^3).
  \end{equation}
Returning to \Cref{eq:shakshuka} and observing that the $\tr(\Id)$ in \Cref{eq:shakshuka} is $2^m$, this gives that
\begin{equation} 
\label{eq:fusilli}
(\ref{eq:shakshuka}) \leq
\eps^2 \left(1 +  {\frac 1 2} \tr(T_O(H)^2)\right) \pm O_m(\eps^3).
\end{equation}

Having a handle on \Cref{eq:shakshuka} for a fixed $O=O_{[1,m-1]}$ is nice, but recall that our actual goal is to upper bound \Cref{eq:i-is-1-j-is-m}, which is the expectation of \Cref{eq:shakshuka} over a random $\bO=\bO_{[1,m-1]}$.
To do this we will use the following claim, which we prove later:

\begin{claim} \label{claim:lollipop}
\begin{equation} \label{eq:lollipop}
\Ex_{\bO}[\tr(T_{\bO}(H)^2)] 
\leq -c' := -{\frac {2^{2m-2} + 2^{m-1} - 6}{2^{2m-1} + 2^m - 4}},
\end{equation}
where the expectation is over $\bO$ distributed as a Haar-random element of $\SOgp{2^{m-1}}$.
\end{claim}
\Cref{claim:lollipop} gives that
\begin{equation}
\label{eq:pappardelle}
(\ref{eq:i-is-1-j-is-m})
= \E_{\bO}[(\ref{eq:shakshuka})] \leq \eps^2 \left(1 - {\frac {c'} 2}\right) +O_m(\eps^3).
\end{equation}
Recalling \Cref{eq:HHJ-55} through \Cref{eq:tabbouleh}, we see that
\begin{align}
\E[\norm{\bX'-\bY'}_F^2]
&\leq
{\frac 1 m} \norm{X-Y}_F^2 + \left(1 - {\frac 1 m} \right)\eps^2\left(1  -  {\frac {c'} 2}\right) + O_m(\eps^3)\label{eq:manicotti}\\
&\leq
\eps^2 \left(1 - \left(1 - {\frac 1 m}\right){\frac {c'} 2} \right) + O_m(\eps^3)
\label{eq:orzo}\\
&= \eta' \cdot \eps^2  + O_m(\eps^3) \label{eq:fettuccine},
\end{align}
where \Cref{eq:orzo} uses $\norm{X-Y}_F^2 \leq \eps^2.$ We have almost reached \Cref{eq:equation-90-BHH-Frobenius}; it remains to observe that since
$d_{\Fro}(X,Y)=\norm{X-Y}_F  \leq \dRie(X,Y)=\eps,$ by \Cref{eq:equation-10-of-Oliveira} we have
$
\eps \leq \norm{X-Y}_F + C_m\eps^2$, so
\begin{equation}
\label{eq:gemelli}
\eps^2 \leq \left( \norm{X-Y}_F + C_m\eps^2\right)^2
= \norm{X-Y}_F^2 + O_m(\eps^3).
\end{equation}
Combining this with \Cref{eq:fettuccine} we get \Cref{eq:equation-90-BHH-Frobenius}.

It remains to establish \Cref{eq:max-Frobenius-distance}, namely that every outcome $(X',Y')$ of $(\bX',\bY')$ satisfies $\norm{X'- Y'}_F^2 \leq \eps^2 + O_m(\eps^3).$ Following \Cref{eq:HHJ-55} through \Cref{eq:tabbouleh}, we see that
\begin{equation} \label{eq:every-outcome}
\norm{X'-Y'}_F^2 = 
{\frac 1 m} \norm{X-Y}_F^2 + {\frac 1 {m^2}} \sum_{i \neq j \in [m]} \min_{V^{i,j}_{[1,i-1] \cup [i+1,m]}}
\norm{V^{i,j}_{[1,i-1] \cup [i+1,m]} O_{[1,j-1] \cup [j+1,m]} X
-
 O_{[1,j-1] \cup [j+1,m]} Y}_F^2.
\end{equation}
From \Cref{eq:fusilli} and the fact that $\tr(T_O(H)^2)$ is non-positive, we have that 
\begin{equation}
\label{eq:ziti}
(\ref{eq:shakshuka}) \leq 
\eps^2 \pm O_m(\eps^3),
\end{equation}
so substituting back into \Cref{eq:every-outcome} we get that
\begin{equation} \label{eq:spaghetti}
\norm{X'-Y'}_F^2 
\leq
{\frac 1 m} \norm{X-Y}_F^2 +
{\frac {m^2 - m}{m^2}} \left(\eps^2 \pm O_m(\eps^3)\right)
\leq \eps^2 \pm O_m(\eps^3),
\end{equation}
where for the last inequality we used that $\norm{X-Y}_F^2 = d_{\Fro}(X,Y)^2 \leq \dRie(X,Y)^2 = \eps^2.$
This concludes the proof of \Cref{lem:eq-90-BHH-analogue-Frobenius}, modulo the proof of \Cref{claim:lollipop}. \qed

\medskip

\subsubsection{Proof of \Cref{claim:lollipop}.}
Our goal is to analyze the expectation of
\begin{equation} \label{claim:lollipop-goal}
\tr\left(T_{\bO}(H)^2\right) =
\tr\left(\left(\tr_1\left(\bO_{[1,m-1]}H\bO_{[1,m-1]}^T\right)\right)^2\right).
\end{equation}
We begin by recalling (see Equation~(102) of \cite{BHH16}) that for any $m$-qubit operator $C \in \SOgp{2^m}$, we have\rasnote{Do we want to write a justification of this? There are notes on it from our 08-08 Slack, I can write a proof/justification if we think we should include one. A justification would be helpful for a reader like me who wanted to verify the whole argument in a more self-contained way, but it might also be ``obvious'' to more sophisticated readers (\cite{BHH16} stated the 3-qubit analogue without proof).}
\begin{equation} \label{eq:BHH-102}
\tr\left(\tr_1(C)^2\right) =
\tr
\left(
\left(C_{[1,m]} \otimes C_{[m+1,2m]}\right)
\left(\mathbb{F}_{[2,m]:[m+2,2m]}\right)
\right)
\end{equation}
where the trace on the right-hand size is over a $2m$-qubit system and $\mathbb{F}_{[2,m]:[m+2,2m]}$ is the $2m$-qubit operator which swaps qubits $[2,m]$ with qubits $[m+2,2m]$ (leaving qubits $1$ and $m+1$ unchanged). Taking $C=\bO_{[1,m-1]}H\bO_{[1,m-1]}$, we get that
\begin{equation}
\label{eq:towards-BHH-103}
\tr\left(T_{\bO}(H)^2\right)
=
\tr
\left(
\left(
\bO_{[1,m-1]}H_{[1,m]}\bO_{[1,m-1]}^T
\otimes 
\bO_{[m+1,2m-1]}H_{[m+1,2m]}\bO_{[m+1,2m-1]}^T
 \right)
\left(\mathbb{F}_{[2,m]:[m+2,2m]}\right)
\right),
\end{equation}
and by the cyclic property of the trace, we get that
\begin{equation}
\label{eq:BHH-103}
\tr\left(T_{\bO}(H)^2\right)
=
\tr
\left(
\left(
H_{[1,m]} \otimes H_{[m+1,2m]}
\right)
\left(
\bO_{[1,m-1]}^T \otimes \bO_{[m+1,2m-1]}^T
\right)
\left(\mathbb{F}_{[2,m]:[m+2,2m]}\right)
\left(
\bO_{[1,m-1]} \otimes \bO_{[m+1,2m-1]}
\right)
\right)
\end{equation}
The following quantum circuit diagram illustrates the $2m$-qubit system whose trace is being taken on the RHS of \Cref{eq:BHH-103}.

  \myfig{0.5}{qcircuit.pdf}{\label{figure:circuit}}{}

Regarding the expectation, from \Cref{eq:BHH-103} we have
\begin{align}
&\E_{\bO}[\tr\left(T_{\bO}(H)^2\right)]\nonumber \\
&=
\tr
\left(
\left(
H_{[1,m]} \otimes H_{[m+1,2m]}
\right)
\Ex_{\bO}\left[
\left(
\bO_{[1,m-1]}^T \otimes \bO_{[m+1,2m-1]}^T
\right)
\left(\mathbb{F}_{[2,m]:[m+2,2m]}\right)
\left(
\bO_{[1,m-1]} \otimes \bO_{[m+1,2m-1]}
\right)\right]
\right),\label{eq:apple}
\end{align}
so we would like to analyze the expectation on the RHS of \Cref{eq:apple}.
We tackle this by analyzing the following closely related expectation over a $(2m-2)$-qubit system with qubits $1,\dots,m-1,m+1,\dots,2m-1$ (i.e. omitting qubits $m$ and $2m$)
\begin{equation} \label{eq:BHH-104-analogue}
\Ex_{\bO \sim \SOgp{2^{m-1}}}\left[
\left(
\bO_{[1,m-1]}^T \otimes \bO_{[m+1,2m-1]}^T
\right)
\left(\mathbb{F}_{[2,m-1]:[m+2,2m-1]}\right)
\left(
\bO_{[1,m-1]} \otimes \bO_{[m+1,2m-1]}
\right)\right]
\end{equation}
(as \Cref{figure:circuit} suggests, it will be straightforward later to relate \Cref{eq:BHH-104-analogue} to the expectation in \Cref{eq:apple}).

To analyze \Cref{eq:BHH-104-analogue} we recall (a corollary of) Schur's lemma for the special orthogonal group (see \cite{Schur-Wikipedia}), which says that if $\rho$ is an irreducible representation of $\SOgp{2^{m-1}}$ on a vector space $V$ and $X$ is a linear mapping of $V$ into $V$, then
\begin{equation}
\label{eq:raw-Schur}
\E_{\bO \sim \SOgp{2^{m-1}}}
\left[\rho(\bO)^{-1} X \rho(\bO)\right] 
= {\frac {\tr(X)}{\dim(V)}} \Id.
\end{equation}
In the setting of \Cref{eq:BHH-104-analogue} the vector space $V$ is the space of real $2^{m-1} \times 2^{m-1}$ matrices. For $O \in \SOgp{2^{m-1}}$ we have that $\rho$ is the tensor representation\rasnote{Is this a good name/notation for this?} $\rho(O)=O_{[1,m-1]} \otimes O_{[m+1,2m-1]}$, which we view as a $2^{2m-2} \times 2^{2m-2}$ matrix which is a linear map from $V$ to $V$ (note that an element of $V$ is a $2^{2m-2}$-dimensional object); the linear map $\rho(O)$ acts on a $2^{m-1} \times 2^{m-1}$ matrix $\Sigma \in V$ by sending it to $O^T \Sigma O.$
The role of ``$X$'' is played by $\mathbb{F}_{[2,m-1]:[m+2,2m-1]}$, which, we remind the reader, in this setting is an operator on the $2m-2$ qubits $[2m] \setminus \{m,2m\}$ (i.e.~a $2^{2m-2} \times 2^{2m-2}$ matrix, which we view as a linear mapping from $V$ into itself) which happens to leave qubits $1$ and $m+1$ untouched. 

We recall (see e.g. Section~IV.1 of \cite{Zee16}) that the tensor representation $\rho$ is not irreducible; rather, it decomposes into three irreducible representations,
\begin{equation}
\label{eq:decomposition}
\rho = \rho_I \oplus \rho_S \oplus \rho_A,
\end{equation}
corresponding to the decomposition of the (real) vector space $V$ of $2^{m-1} \times 2^{m-1}$ matrices into three subspaces $I \oplus S \oplus A$, where $I$ is the 1-dimensional subspace of real multiples of the $2^{m-1} \times 2^{m-1}$ identity matrix, $S$ is the $(2^{m-1}(2^{m-1} + 1)/2 - 1)$-dimensional subspace of real symmetric traceless $2^{m-1} \times 2^{m-1}$ matrices, and $A$ is the $(2^{m-1}(2^{m-1} - 1)/2 - 1)$-dimensional subspace of real antisymmetric (hence also traceless) $2^{m-1} \times 2^{m-1}$ matrices.
(It is easy to verify that if $\Sigma \in V$ is a $2^{m-1} \times 2^{m-1}$ matrix that belongs to any one of these three subspaces, then the result of $\rho(O)$'s action on $\Sigma$, namely $O^T \Sigma O$, lies in the same subspace; for example, if $\Sigma = c \Id$, then we have $O^T \Sigma O = c O^T \Id O = c \Id = \Sigma.$) We write $\Pi_I, \Pi_S$ and $\Pi_A$ to denote the projectors onto these three subspaces.

Applying Schur's lemma to each of these three irreducible representations, we get that for a general $X$, 
\begin{equation} \label{eq:application-of-Schur}
(\ref{eq:raw-Schur})
= 
{\frac {\tr(X \Pi_I)}{1}} \Pi_I 
+
{\frac {\tr(X \Pi_S)}{2^{m-1}(2^{m-1} + 1)/2 - 1}} \Pi_S 
+
{\frac {\tr(X \Pi_A)}{2^{m-1}(2^{m-1} - 1)/2 }} \Pi_A.
\end{equation}
Viewed as a $2^{2m-2} \times 2^{2m-2}$ matrix, this is  a diagonal matrix with three diagonal  blocks of sizes 1, $2^{m-1}(2^{m-1} + 1)/2 - 1$ and $2^{m-1}(2^{m-1} - 1)/2$ respectively.
Below we first give a description of the projectors $\Pi_I,\Pi_S$ and $\Pi_A$, and then use this description to analyze \Cref{eq:application-of-Schur} for the specific $X=\mathbb{F}_{[2,m-1]:[m+2,2m-1]}$ of \Cref{eq:BHH-104-analogue}.

As stated above, the projections $\Pi_I, \Pi_S, \Pi_A$ correspond to decomposing a $2^{m-1} \times 2^{m-1}$ matrix $\Sigma$ into 
\begin{equation}
\label{eq:sigma-decomposition}
\Sigma = \Sigma_I + \Sigma_S + \Sigma_A
\end{equation} 
where $\Sigma_I$ is a scalar multiple of the identity matrix, $\Sigma_S$ is symmetric and traceless, and $\Sigma_A$ is antisymmetric and traceless.
Let us consider each of these projections in turn:

\begin{itemize}

\item The projector $\Pi_A$ that maps $\Sigma$ to $\Sigma_A$ is the mapping that takes $\Sigma$ to ${\frac 1 2} \left(\Sigma - \Sigma^T\right)$. Since the operator that takes $\Sigma$ to $\Sigma^T$ is $\mathbb{F}_{[1,m-1]:[m+1,2m-1]}$ (because swapping the first $m-1$ qubits with the second $m-1$ qubits corresponds to exchanging rows with columns), we have that $\Pi_A = {\frac 1 2} \left(\Id - \mathbb{F}_{[1,m-1]:[m+1,2m-1]}\right).$

\item The projector $\Pi_I$ that maps $\Sigma$ to $\Sigma_I$ is the mapping that takes $\Sigma$ to ${\frac {\tr(\Sigma)}{\tr(\Id)}} \Id$ (note that here ``$\Id$'' is the $2^{m-1} \times 2^{m-1}$ identity matrix).

\item Finally, the projector $\Pi_S$ that maps $\Sigma$ to $\Sigma_S$ is the mapping that takes $\Sigma$ to ${\frac 1 2} \left(\Sigma + \Sigma^T\right) - {\frac {\tr(\Sigma)}{\tr(\Id)}} \Id$, i.e. $\Pi_S = {\frac 1 2} \left(\Id + \mathbb{F}_{[1,m-1]:[m+1,2m-1]}\right) - \Pi_I.$

\end{itemize}

Now recall that the $X$ we are concerned with is $X=\mathbb{F}_{[2,m-1]:[m+2,2m-1]},$ and recall that this is a $2^{2m-2} \times 2^{2m-2}$ matrix. In the next three claims we give expressions for the three traces in the RHS of \Cref{eq:application-of-Schur}.

\begin{claim} \label{claim:trace-of-XPi_A}
$\tr(\mathbb{F}_{[2,m-1]:[m+2,2m-1]} \Pi_A)= 2^{m-1} - 2^{2m-4}$.
\end{claim}
\begin{proof}
We have that
\begin{align} 
\tr(\mathbb{F}_{[2,m-1]:[m+2,2m-1]} \Pi_A)&= {\frac 1 2} \left(\tr\left(\mathbb{F}_{[2,m-1]:[m+2,2m-1]} - 
\mathbb{F}_{[1,m-1]:[m+1,2m-1]} \mathbb{F}_{[2,m-1]:[m+2,2m-1]} \right)
\right) \nonumber \\
&= {\frac 1 2} \tr\left(\mathbb{F}_{[2,m-1]:[m+2,2m-1]}\right) - {\frac 1 2} \tr\left( 
\mathbb{F}_{1:m+1}\right) = 2^{m-1} - 2^{2m-4},
\label{eq:trace-XPi_A} 
\end{align}
where for the final equality we used the fact that any swap $\mathbb{F}_{\cdot,\cdot}$ that exchanges $0 \leq i \leq m-1$ of the first $m-1$ qubits with $i$ of the last $m-1$ qubits has trace $2^{2m-2-i}.$ (To see this, observe that given a $(2m-2)$-bit string $(a_1, \dots a_{m-1}, b_1,\dots,b_{m-1})$ indexing a row of an $i$-bit swap matrix such as $\mathbb{F}_{[1,i]:[m+1,m+i]}$, the $((a_1, \dots a_{m-1}, b_1,\dots,b_{m-1}),(a_1, \dots a_{m-1}, b_1,\dots,b_{m-1}))$-entry of $\mathbb{F}_{[1,i]:[m+1,m+i]}$ is 1 if and only if $a_j=b_j$ for $j \in [i]$; this holds for a $1/2^i$ fraction of all $2^{2m-2}$ strings $(a_1, \dots a_{m-1}, b_1,\dots,b_{m-1}).$)
\end{proof}

\begin{claim} \label{claim:trace-of-XPi_I}
$\tr(\mathbb{F}_{[2,m-1]:[m+2,2m-1]} \Pi_I)= 1.$
\end{claim}
\begin{proof}
The matrix $\mathbb{F}_{[2,m-1]:[m+2,2m-1]} \Pi_I$ is a $2^{2m-2} \times 2^{2m-2}$ matrix  whose rows are indexed by pairs $(\overline{a},\overline{b})$ where $\overline{a}=(a_1,\dots,a_{m-1}) \in \{0,1\}^{m-1}$ and $\overline{b}=(b_1,\dots,b_{m-1}) \in \{0,1\}^{m-1}$ (and the columns have the same indexing).
We have
\begin{equation} \label{eq:trace-of-XPi_I}
\tr(\mathbb{F}_{[2,m-1]:[m+2,2m-1]} \Pi_I)
=
\sum_{\overline{a},\overline{b}} \langle \overline{a},\overline{b}| \mathbb{F}_{[2,m-1]:[m+2,2m-1]} \Pi_I |\overline{a}, \overline{b} \rangle;
\end{equation}
viewing $|\overline{a},\overline{b}\rangle$ as a $2^{m-1} \times 2^{m-1}$ matrix (a matrix with a single 1-entry that is in row $\overline{a}$ and column $\overline{b}$) and recalling the definition of $\Pi_I$, we see that $\Pi_I$ acts on $|\overline{a},\overline{b}\rangle$ by sending it to 0 if $\overline{a} \neq \overline{b}$ and sending it to ${\frac 1 {2^{m-1}}} |\overline{a},\overline{a} \rangle$ if $\overline{a}=\overline{b}$. Hence we have that
\begin{equation}
\label{eq:martello}
\langle \overline{a},\overline{b}| \mathbb{F}_{[2,m-1]:[m+2,2m-1]} \Pi_I |\overline{a}, \overline{b} \rangle
= \delta_{\overline{a}=\overline{b}} \cdot {\frac 1 {2^{m-1}}} \cdot
\langle \overline{a},\overline{a}| \mathbb{F}_{[2,m-1]:[m+2,2m-1]} |\overline{a}, \overline{a} \rangle
\end{equation}
and hence
\begin{align}
(\ref{eq:trace-of-XPi_I}) &=
{\frac 1 {2^{m-1}}} \sum_{\overline{a}}
\langle \overline{a},\overline{a}| \mathbb{F}_{[2,m-1]:[m+2,2m-1]} |\overline{a}, \overline{a} \rangle
= 
{\frac 1 {2^{m-1}}} \sum_{\overline{a}}
\langle \overline{a},\overline{a}|\overline{a}, \overline{a} \rangle = 1,
\end{align}
where the penultimate equality is because swapping any subset of bits in the first copy of $\overline{a}$ with the corresponding subset of bits in the second copy of $\overline{a}$ leaves $|\overline{a},\overline{a}\rangle$ unchanged.
\end{proof}

\begin{claim} \label{claim:trace-of-XPi_S}
$\tr(\mathbb{F}_{[2,m-1]:[m+2,2m-1]} \Pi_S)= 2^{m-1} + 2^{2m-4} - 1.$
\end{claim}
\begin{proof}
Recalling that $\Pi_S = {\frac 1 2} \left(\Id + \mathbb{F}_{[1,m-1]:[m+1,2m-1]}\right) - \Pi_I$, we have
\begin{align}
\tr(\mathbb{F}_{[2,m-1]:[m+2,2m-1]} \Pi_S) &=
{\frac 1 2}\tr\left(\mathbb{F}_{[2,m-1]:[m+2,2m-1]}\right)
+ {\frac 1 2} \tr\left( \mathbb{F}_{1:m+1}\right)
- \tr\left(\mathbb{F}_{[2,m-1]:[m+2,2m-1]} \Pi_I \right) \nonumber\\
&= 2^{m-1} + 2^{2m-4} - 1,
\end{align}
where the last equality is by \Cref{eq:trace-XPi_A}  and \Cref{claim:trace-of-XPi_I}.
\end{proof}

For conciseness let us define $M:=2^{m-1}.$ Recalling \Cref{eq:BHH-104-analogue} and \Cref{eq:application-of-Schur} and putting the pieces together, by \Cref{claim:trace-of-XPi_A},  \Cref{claim:trace-of-XPi_I} and \Cref{claim:trace-of-XPi_S} we have that
\begin{align}
(\ref{eq:BHH-104-analogue}) 
 &=
 \Pi_I + {\frac {M + M^2/4 - 1}{M(M + 1)/2 - 1}} \Pi_S 
 + {\frac {M - M^2/4}{M(M - 1)/2 }} \Pi_A \nonumber\\
 &= \Pi_I + {\frac {M + M^2/4 - 1}{M(M + 1)/2 - 1}}\left({\frac 1 2} \left(\Id + \mathbb{F}_{[1,m-1]:[m+1,2m-1]}\right) - \Pi_I\right)
 +  {\frac {M - M^2/4}{M(M - 1)/2 }}\left({\frac 1 2} \left(\Id - \mathbb{F}_{[1,m-1]:[m+1,2m-1]}\right)\right) \nonumber\\
 &= {\frac {M^2 - 2M}{2M^2 + 2M - 4}} \Pi_I + 
 {\frac {3M+2}{2M^2 + 2M - 4}} \Id +
 {\frac {M^2 + M - 6}{2M^2 + 2M - 4}} \mathbb{F}_{[1,m-1]:[m+1,2m-1]}, \label{eq:plum}
 \end{align}
 and hence by \Cref{eq:apple} we have
 \begin{align}
 \E_{\bO}[\tr\left(T_{\bO}(H)^2\right)]
&= \tr
\left(
\left(
H_{[1,m]} \otimes H_{[m+1,2m]}
\right)
\left(
\mathbb{F}_{m:2m} \otimes (\ref{eq:plum})
\right)
\right) \nonumber \\
&=
{\frac {M^2 - 2M}{2M^2 + 2M - 4}} \tr \left(
\left(
H_{[1,m]} \otimes H_{[m+1,2m]}
\right) 
\left(\mathbb{F}_{m:2m} \otimes 
(\Pi_I)_{[1,m-1] \cup [m+1,2m-1]}
 \right)\right) \label{eq:pear1}\\
& \ \ \ + {\frac {3M+2}{2M^2 + 2M - 4}}  \tr
\left(
\left(
H_{[1,m]} \otimes H_{[m+1,2m]}
\right) 
\mathbb{F}_{m:2m}\right)  \label{eq:pear2} \\
& \ \ \ + {\frac {M^2 + M - 6}{2M^2 + 2M - 4}} \tr
\left(
\left(
H_{[1,m]} \otimes H_{[m+1,2m]}
\right) 
\mathbb{F}_{[1,m]:[m+1,2m]}\right)
 \label{eq:pear3}
\end{align}

We first show that $(\ref{eq:pear1}) \leq 0$; for notational convenience we write
$\Pi'_I$ to denote $(\Pi_I)_{[1,m-1] \cup [m+1,2m-1]}$.


\begin{align}
(\ref{eq:pear1}) 
&= \tr\left(\left(
H_{[1,m]} \otimes H_{[m+1,2m]}
\right) 
\left(\mathbb{F}_{m:2m} \otimes \Pi'_I \right)\right) \nonumber \\
&=\sum_{\overline{a},\overline{b} \in \zo^m} 
\langle \overline{a}, \overline{b}|
(H_{\overline{a}} \otimes H_{\overline{b}})
(\mathbb{F}_{m:2m} \otimes \Pi'_I) |\overline{a},\overline{b}\rangle \nonumber \\
&=\sum_{\overline{a},\overline{b} \in \zo^m} 
\langle \overline{a}, \overline{b}|
(H_{\overline{a}} \otimes H_{\overline{b}})
\Pi'_I |(a_1,\dots,a_{m-1},b_m),(b_1,\dots,b_{m-1},a_m)\rangle \label{eq:peach}.
\end{align}
Recalling that $\Pi'_I |(a_1,\dots,a_{m-1},b_m),(b_1,\dots,b_{m-1},a_m)\rangle= {\frac 1 {2^{m-1}}} \delta_{(a_1,\dots,a_{m-1})=(b_1,\dots,b_{m-1})}|(a_1,\dots,a_{m-1},b_m),\overline{a}\rangle$, we have that
(\ref{eq:peach}) is equal to
\begin{align}
& {\frac 1 {2^{m-1}}}\sum_{\overline{a} \in \zo^m,b_m \in \zo}
\langle\overline{a},(a_1,\dots,a_{m-1},b_m)|
(H_{\overline{a}} \otimes H_{\overline{b}})|(a_1,\dots,a_{m-1},b_m),\overline{a})\rangle \nonumber \\
&=
 {\frac 1 {2^{m-1}}}\sum_{\overline{a} \in \zo^m,b_m \in \zo}
\langle\overline{a}| H |(a_1,\dots,a_{m-1},b_m)\rangle \cdot
\langle (a_1,\dots,a_{m-1},b_m)|
H|\overline{a}\rangle \nonumber\\
&=
{\frac 1 {2^{m-1}}}\sum_{\overline{a} \in \zo^m,b_m \in \zo}
-\left(\langle\overline{a}| H |(a_1,\dots,a_{m-1},b_m)\rangle \right)^2 \leq 0,
\end{align}
 where the last equality is by skew-symmetry of $H$.
 
 Next we show that $(\ref{eq:pear2}) \leq 0.$ For brevity let us write $\overline{a}$ to denote $(a_1,\dots,a_{m-1})$ and likewise $\overline{b}$ to denote $(b_1,\dots,b_{m-1})$. We have 
\begin{align}
(\ref{eq:pear2}) &= \tr
\left(
\left(
H_{[1,m]} \otimes H_{[m+1,2m]}
\right) 
\mathbb{F}_{m:2m}\right)\nonumber \\
&=\sum_{\overline{a},\overline{b} \in \zo^{m-1},a,b \in \zo} 
\langle \overline{a},a, \overline{b},b|
(H_{\overline{a},a} \otimes H_{\overline{b},b})
\mathbb{F}_{m:2m} |\overline{a},a,\overline{b},b\rangle \nonumber \\
&=\sum_{\overline{a},\overline{b},a,b} 
\langle \overline{a},a, \overline{b},b|
(H_{\overline{a},a} \otimes H_{\overline{b},b})
 |\overline{a},b,\overline{b},a\rangle \nonumber\\
 &= \sum_{\overline{a},\overline{b},a,b}
 \langle  \overline{a},a| H |\overline{a},b\rangle \cdot
\langle \overline{b},b|
H|\overline{b},a\rangle \nonumber\\
&= -\sum_{\overline{a},\overline{b},a,b}
 \langle  \overline{a},a| H |\overline{a},b\rangle \cdot
\langle \overline{b},a|
H|\overline{b},b\rangle \tag{skew-symmetry of $H$}\nonumber\\
 &= -\sum_{\overline{a},\overline{b}}
\left( 
\langle  \overline{a},0| H |\overline{a},0\rangle \cdot \langle \overline{b},0| H|\overline{b},0\rangle  +
\langle  \overline{a},0| H |\overline{a},1\rangle \cdot \langle \overline{b},0| H|\overline{b},1\rangle  \right. + \nonumber\\
&  \  \  \  \ \ \ \ \ \ \ \ \ 
\left. \langle  \overline{a},1| H |\overline{a},0\rangle \cdot \langle \overline{b},1| H|\overline{b},0\rangle +
\langle  \overline{a},1| H |\overline{a},1\rangle \cdot \langle \overline{b},1| H|\overline{b},1\rangle 
\right).
 \end{align}
Recalling that diagonal entries of $H$ are 0 because $H$ is skew-symmetric, we see that the above equals
\begin{align}
& -\sum_{\overline{a},\overline{b}}
\langle  \overline{a},0| H |\overline{a},1\rangle \cdot \langle \overline{b},0| H|\overline{b},1\rangle  
+ \langle  \overline{a},1| H |\overline{a},0\rangle \cdot \langle \overline{b},1| H|\overline{b},0\rangle \nonumber\\
&=
-2 \sum_{\overline{a},\overline{b}}
\langle  \overline{a},0| H |\overline{a},1\rangle \cdot \langle \overline{b},0| H|\overline{b},1\rangle  \tag{skew-symmetry of $H$} \nonumber\\
&=
-2 \left(\sum_{\overline{a}}
\langle  \overline{a},0| H |\overline{a},1 \rangle\right)^2 \leq 0.
\end{align}
 
%
Finally we turn to $(\ref{eq:pear3}).$
  \begin{align}
\tr
\left(
\left(
H_{[1,m]} \otimes H_{[m+1,2m]}
\right) 
\mathbb{F}_{[1,m]:[m+1,2m]}\right)
&=\sum_{\overline{a},\overline{b} \in \zo^m} 
\langle \overline{a}, \overline{b}|
(H_{\overline{a}} \otimes H_{\overline{b}})
\mathbb{F}_{[1,m]:[m+1,2m]} |\overline{a},\overline{b}\rangle \nonumber \\
&=\sum_{\overline{a},\overline{b} \in \zo^m} 
\langle \overline{a}, \overline{b}|
(H_{\overline{a}} \otimes H_{\overline{b}})
 |\overline{b},\overline{a}\rangle \nonumber\\
 &= \sum_{\overline{a},\overline{b} \in \zo^m}
 \langle\overline{a}| H |\overline{b}\rangle \cdot
\langle \overline{b}|
H|\overline{a}\rangle \nonumber\\
&=  \sum_{\overline{a},\overline{b} \in \zo^m} - \left(\langle \overline{a}|H|\overline{b}\rangle\right)^2 = -\|H\|_F^2 = -1. 
\label{eq:watermelon}
 \end{align}
 Combining \Cref{eq:pear1} through \Cref{eq:watermelon} and the fact that $m \geq 3$, this completes the proof of \Cref{claim:lollipop}. \qed
 
 }


\section{Pseudorandom products of operators} \label{sec:pseudorandom-walks}
In this section we generalize the ``derandomized squaring'' technique of Rozenman and Vadhan~\cite{RV05} so that it may be applied to random walks on groups, where the goal is to show rapid mixing of a particular representation.
We remark that the proofs are not really different from those in~\cite{RV05}, and that a similar generalization appeared recently in~\cite{JMRW22}.

\begin{notation}
    Throughout we will be considering noncommutative polynomials, with real coefficients, over symbols $u_1, \dots, u_{\c}, u_1^\dagger, \dots, u_{\c}^\dagger$.
    (These symbols will eventually be substituted by square matrices.)
    If $p$ is such a polynomial, its \emph{adjoint} $p^\dagger$ is formed in the natural way (i.e., $(u_i^\dagger)^\dagger = u_i$ and  $(u_i u_j)^\dagger = u_j^\dagger u_i^\dagger$, etc.), and we call $p$ \emph{self-adjoint} if $p^\dagger = p$.
\end{notation}
\begin{notation}
    We will also consider \emph{polynomial sequences}~$S = (s_1, \dots, s_m)$, where each $s_i$ is a polynomial in the $u_j$'s.  (Usually $s_i$ will in fact be a \emph{monomial}.)
\end{notation}
\begin{notation}
    If $\calU = (U_1, \dots, U_{\c})$ is a sequence of matrices, we write $S(\calU) = (s_1(\calU), \dots, s_m(\calU))$, where $s_j(\calU)$ is the matrix resulting from substituting $u_i = U_i$ for each~$i \in [\c]$.  
\end{notation}
\begin{notation}
    Given a polynomial sequence~$S$ we write $\avg(S)$, or $\avg \circ S$, for the polynomial $\frac{1}{m} \sum_{j=1}^m s_j$.
\end{notation}
\begin{definition}
    If $p$ is a polynomial over $u_1, \dots, u_{\c}$, we define
    \begin{equation}
        \|p\| = \sup_r \{\opnorm{p(\calU)} : \calU = (U_1, \dots, U_{\c}),\ U_j \in \C^{r \times r},\ \opnorm{U_j} \leq 1\ \forall j\},
    \end{equation}
    the largest operator norm that  $p$ can achieve when $u_1, \dots, u_{\c}$ are substituted with square matrices of bounded operator norm.
    More generally, if $S = (s_1, \dots, s_m)$ is a sequence of  polynomials we write $\|S\| = \max(\|s_1\|, \dots, \|s_m\|)$.
\end{definition}
\begin{definition}
    A \emph{directed graph} $G = (V,E)$ will consist of a finite \emph{sequence} of vertices~$V$, and a finite \emph{sequence} of edges~$E$ from $V \times V$ (so parallel edges and self-loops are allowed).
    Such a graph is \emph{undirected} if $E$ can be partitioned into pairs of the form $\{(i,j), (j,i)\}$. 
    We say $G$ is \emph{$d$-out-regular} if for each $i \in V$ we have exactly $d$ elements of the form $(i,j)$ in~$E$; one can analogously define in-regularity, and the two concepts are the same for undirected graphs.  
    Note that if $G$ is an undirected $d$-regular graph,  then $|E| = d|V|$ (contrary to usual convention, as $E$ is still composed of directed edges).
\end{definition}
\begin{definition}  \label{def:q}
    Given a graph $G = (V,E)$, where $V = (1, 2, \dots, m)$, and given a polynomial sequence $S = (s_1, \dots, s_m)$, we define $q_G \circ S$ to be the polynomial sequence\footnote{One would have hoped for the more natural-looking ordering $s_i^\dagger s_j$, but alas we are forced to follow standard conventions: the length-$2$ path $(i,j)$ means ``first~$i$, then~$j$''; but, with operators acting on the left, $s_i^\dagger s_j$ means ``first do $s_j$, then do $s_i^\dagger$''.}
    \begin{equation}
            (s_j^\dagger s_i)_{(i,j) \in E}.
    \end{equation}
\end{definition}
\begin{remark}
    If $G$ is undirected, the polynomial $\avg(q_G \circ S)$ is self-adjoint.
\end{remark}
\begin{fact}    \label{fact:quadratic}
    We always have $\|q_G \circ S\| \leq \|S\|^2$, and hence $\|S\| \leq 1 \implies \|q_G \circ S\| \leq 1$.
\end{fact}
\begin{definition}
    If $G = (V,E)$ is a $d$-out-regular directed graph with $V = (1, 2, \dots, \c)$, then the normalized adjacency matrix of~$G$ is
    \begin{equation}    \label{eqn:adj}
        A_G \coloneqq \frac{1}{d} \sum_{(i,j) \in E} \ket{j}\!\bra{i} = \c \cdot \avg(q_G \circ \calU), \quad \text{where } \calU = (\bra{1}, \dots, \bra{\c}).
    \end{equation}
\end{definition}
\begin{fact}    \label{fact:fun}
    Let $G = (V,E)$ be an out-regular directed graph with $V = (1, 2, \dots, m)$ and let 
    $\calW = (W_1, \dots W_m)$ be a sequence of matrices from $\C^{r \times r'}$.\footnote{We only really care about $r' = r$, but we allow $r' \neq r$ for the sake of comparison with \Cref{eqn:adj}, where $r = 1$ and $r' = c$.}
    Then
    \begin{equation}    \label{eqn:funfact}
        \avg(q_G(\calW)) = \frac1m \calW^\dagger (A_G \otimes \Id_{r \times r}) \calW,
    \end{equation}
    where we identify $\calW$ with $\sum_{j=1}^m \ket{j} \otimes W_j$, the $mr \times r'$ matrix formed by stacking the $W_j$'s into a column.
    (For $r'=r$, this identity is essentially the formula $w^\dagger A w = \sum_{ij} w_i^\dagger A_{ij} w_j$, but with entries from the ring~$\C^{r \times r}$.)
\end{fact}
\begin{notation}
    We let $\K_m$ denote the complete (regular) undirected graph with self-loops on~$m$ vertices, which has $V = (1, 2, \dots, m)$ and $E = ((1,1), (1,2), \dots, (1,m), (2,1), \dots, (m,m))$.
    We may write $\K$ in place of $\K_m$ if the context is clear.
\end{notation}
\begin{fact}    \label{fact:sq}
    If $S = (s_1, \dots, s_m)$ is a polynomial sequence,
    \begin{equation}
        \avg(q_{\K_m} \circ S) = \avg(S)^\dagger \avg(S),
    \end{equation}
    the Hermitian-square of $\avg(S)$.  
    Hence if $\calU$ is a sequence of matrices, $\opnorm{\avg(q_{\K} \circ S(\calU))} = \opnorm{\avg(S(\calU))}^2$.
\end{fact}
\begin{definition}
    Recall that a regular undirected graph $G$ is said to be a \emph{(2-sided) $\mu$-expander} if $\opnorm{A_G - A_\K} \leq \mu$.
\end{definition}
\begin{fact}    \label{fact:expand}
    Since $A_{\K}$ is the projection onto the $1$-dimensional subspace spanned by $\sum_j \ket{j}$, and since~$A_G$ also fixes this subspace, $G$ being a $\mu$-expander is  equivalent to $\opnorm{A_G - (1-\mu)A_{\K}} \leq \mu$.
\end{fact}

The following result is essentially the same as \cite[Thm.~4.4]{RV05}:
\begin{proposition}
    Let $G$ be a $\mu$-expander on vertex set $V = (1, 2, \dots, m)$, let $S = (s_1, \dots, s_m)$ be a polynomial sequence with $\|S\| \leq 1$, and let $\calU = (U_1, \dots, U_{\c})$ be a sequence of matrices in~$\C^{r \times r}$ with $\opnorm{U_j} \leq 1$ for all~$j$.
    Then
    \begin{equation}
        \opnorm{\avg(q_G \circ S(\calU))} \leq (1-\mu) \opnorm{\avg(S(\calU))}^2 + \mu.
    \end{equation}
\end{proposition}
\begin{proof}
    Write $\calW = S(\calU) = (W_1, \dots, W_m)$, so $\opnorm{W_j} \leq 1$ for all~$j$.  
    Using \Cref{fact:fun} twice, we derive
    \begin{equation}
        \opnorm{\avg(q_G(\calW)) - (1-\mu)\avg(q_{\K}(\calW))} = \frac1m \opnorm{\calW^\dagger(\Delta \otimes \Id_{r \times r}) \calW}, \quad \text{where } \Delta = A_G - (1-\mu)A_{K}.
    \end{equation}
    We have $\opnorm{\Delta} \leq \mu$ by \Cref{fact:expand}, and $\opnorm{\calW} \leq \sqrt{\littlesum_j \opnorm{W_j}} \leq \sqrt{m}$.
    So by submultiplicativity of operator norm, the right-hand side above is at most~$\mu$, and the proof is complete from \Cref{fact:sq} and the triangle inequality.
\end{proof}
Iterating this, and using \Cref{fact:quadratic} to conclude that $\norm{q_{G_t} \circ \cdots \circ q_{G_1} \circ S} \leq 1$ whenever $\norm{S} \leq 1$, yields:
\begin{proposition} \label{prop:itera}
    Let $S = (s_1, \dots, s_m)$ be a polynomial sequence with $\|S\| \leq 1$ and let $\calU = (U_1, \dots, U_{\c})$ be a matrix sequence with $\opnorm{U_j} \leq 1$ for all~$j$.
    Moreover, let $G_1, G_2, \dots, G_t$ be a sequence of regular graphs, where $G_i = (V_i, E_i)$ is a $\mu_i$-expander  with $V_{i+1} = E_i$ (and $V_1 = (1, 2, \dots, m)$).
    Then
    \begin{equation}
        \opnorm{\avg(q_{G_t} \circ q_{G_{t-1}} \circ \cdots \circ q_{G_1} \circ S(\calU))} \leq f_{\mu_t} \circ f_{\mu_{t-1}} \circ \cdots \circ f_{\mu_1}(\opnorm{\avg(S(\calU))}),
    \end{equation}
    where $f_{\mu}(\lambda) = (1-\mu)\lambda^2 + \mu$.
    In particular, if $m = \c$, $S = (u_1, \dots, u_{\c})$, and we write $Q = q_{G_t} \circ \cdots \circ q_{G_1}$ and $F_{(\mu_1, \dots, \mu_t)} = f_{\mu_t} \circ  \cdots \circ f_{\mu_1}$, then
    \begin{equation}
        \opnorm{\avg(Q \circ \calU)} \leq F_{(\mu_1, \dots, \mu_t)}(\opnorm{\avg(\calU)}).
    \end{equation}
\end{proposition}

The work~\cite{RV05} also contains calculations very similar to the following (wherein the special number~$.11$ is chosen due to certain explicit expander constructions):
\begin{proposition} \label{prop:calcs}
    For $0 < \delta, \eps \leq 1$, we have $F_{\vec{\mu}}(1-\delta) \leq \eps$ for any sequence $\vec{\mu}$ that entrywise satisfies
    \begin{equation}
        (0, \dots, 0) \leq \vec{\mu} \leq (\vec{\mu}^{(1)}, \vec{\mu}^{(2)}), \qquad \vec{\mu}^{(1)} \coloneqq (\underbrace{.11, \dots, .11}_{\text{$\ell_1$ times}}), \qquad \vec{\mu}^{(2)} \coloneqq \tfrac14 (2^{-2}, 2^{-4},  2^{-8}, \dots, 2^{-2^{\ell_2}}),
    \end{equation}
    where $\ell_1 \geq \log_{2^{.8}}(1/\delta) +3$ (note: $2^{.8} \approx 1.74$) and $\ell_2 \geq \log_2 \log_2 (1/\eps)$.
\end{proposition}
\begin{proof}
    Since $f_{\mu}(\lambda)$ is nondecreasing on $[0,1]$ for both $\mu$ and $\lambda$, it suffices to analyze all upper bounds as if they were equalities.
    It is easy to check that $f_{.11}(1-\delta) \leq 1-2^{.8}\delta$ for all $0 \leq \delta \leq .03$, and hence
    \begin{equation}
        \ell \geq  \log_{2^{.8}}(1/\delta) - 6 \quad\implies\quad f_{.11}^{\circ \ell}(1-\delta) \leq 1 - .03/1.75 \leq .985.
    \end{equation}
    Also, $f_{.11}^{\circ 9}(.985) \leq 1/4$, and hence $F_{\vec{\mu}^{(1)}}(1-\delta) \leq 1/4$.
    The proof is now complete by observing that $F_{\vec{\mu}^{(2)}}(1/4) \leq \frac12 2^{-2^{-\ell_2}}$.
\end{proof}

Regarding explicit construction of expander graphs, taking $p = 29$ and $509$ in \cite[Thm.~1.2]{Alo21} and adding a self-loop to every vertex yields:
\begin{theorem} \label{thm:alon}
    For $(d,\mu) = (32,.45)$ and also $(d,\mu) = (512,.11)$, there is a strongly explicit algorithm for constructing $n$-vertex, $d$-regular, $\mu$-expander graphs (for all sufficiently large~$n$).
\end{theorem}
By repeatedly squaring the $32$-regular graphs above, one can also conclude the following (in which it is possible that $d = d(n) > n$):
\begin{corollary}   \label{cor:expanders}
    For any easy-to-compute $j = j(n) \in \N$, there is a strongly explicit ($\polylog(n, d)$ time) algorithm for constructing $n$-vertex, $d$-regular, $\mu$-expander graphs (for all sufficiently large~$n$) where, for $k = 2^j$, we have $d = 32^k$ and $\mu = \mu(n) = .45^k \leq \frac14 2^{-k} = \frac14 d^{-1/5}$ (the inequality holding provided $j \geq 4$).
\end{corollary}

Putting together \Cref{cor:expanders},  \Cref{prop:calcs}, and \Cref{prop:itera} yields the following:
\begin{theorem} \label{thm:derando-walks}
    There is a strongly explicit, space-minimal algorithm with the following behavior on inputs~$c$ and $0 < \delta, \eps < 1$ (where we assume $c = 2^{i_1}$, $\delta = 16^{-i_2}$, and $\eps = 2^{-2^{i_3}}$ for  some \mbox{$i_1,i_2,i_3 \in \N$} sufficiently large).
    The algorithm outputs a sequence~$Q$ of $N = O(\c/(\delta^{11.25} \eps^{10}))$ monomials over symbols $u_1, \dots, u_{\c}$ and $u_1^\dagger, \dots, u_{\c}^\dagger$, each of length $L = 8\log_2(1/\eps)/\delta^{1.25}$, with the following property:

    For any sequence $\calU = (U_1, \dots, U_{\c})$ of matrices in $\C^{r \times r}$ satisfying $\opnorm{U_i} \leq 1$ for all~$i$ and $\opnorm{\avg(\calU)} \leq 1-\delta$, it holds that $\opnorm{\avg(Q \circ \calU)} \leq \eps$.

    Here ``strongly explicit and space-minimal'' means that, given a monomial index~$i \in [N]$  and a monomial position index $j \in [L]$, the algorithm runs in deterministic $\polylog(\c/\delta \eps)$ time and $O(\log(\c/\delta \eps))$ space and outputs the~$j$th symbol of the $i$th monomial in~$Q$.
\end{theorem}
\begin{proof}
    Given $c, \delta, \eps$, the desired $Q$ is $q_{G_t} \circ \cdots \circ q_{G_1} \circ (u_1, \dots, u_{\c})$, where $G_1, \dots, G_t$ is a sequence as in \Cref{prop:itera}, with:
    \begin{itemize}
        \item $\ell_1 = \log_{2^{.8}}(1/\delta) + 3 = \frac54 \log_2(1/\delta) + 3$, $\ell_2 = \log_2 \log_2 (1/\eps)$, and $t = \ell_1 + \ell_2$;
        \item $G_1, \dots, G_{\ell_1}$ are $512$-regular $.11$-expanders, with $G_j$ on $512^{j-1} {\c}$ vertices, as in \Cref{thm:alon};
        \item $G_{\ell_1 + 1}, \dots, G_{\ell_1 + \ell_2}$ are as in \Cref{cor:expanders}, with $G_{\ell_1 + j}$ being a $32^k$-regular, $\tfrac14 2^{-k}$-expander (for $k = 2^{\min(j,4)}$) on $32^{k+32} N_0$ vertices (once $j \geq 4$), where $N_0 = 512^{\ell_1} \c$ is $|E(G_{\ell_1})|$.
    \end{itemize}
    The length of~$Q$ is 
    \begin{equation}
        N = |E(G_{t})| = 32^{2^{\ell_2+1}+32}N_0 = 2^{160} \cdot 2^{5 \cdot \log_2(1/\eps) \cdot 2} \cdot 2^{9(\log_{2^{4/5}}(1/\delta) + 3)} = 2^{187} \cdot \c/\delta^{11.25}\eps^{10},
    \end{equation}
    and each monomial in~$Q$ has length $2^t = 8\log_2(1/\eps)/\delta^{1.25}$.
    The desired bound $\opnorm{\avg(Q(\calU))} \leq \eps$ follows from \Cref{prop:itera,prop:calcs}.  
    Finally, the time and space bounds are easy to verify, as computation of the $j$th symbol of the $i$th monomial of~$Q$ simply amounts to determining  the $i$th edge of~$G_t$, and then following a path down a binary tree of height~$t$, where at each node one has to compute the the~$a$th edge of a particular~$G_b$.
\end{proof}
\begin{remark}
    As in \cite[Thm.~5.8]{RV05}, if $\delta$ is not small but is rather already of the form~$\delta = 1-\lambda$ for small~$\lambda$, one can retain only the last $\ell_2 - \log_2 \log_2(1/\lambda)$ or so expanders and obtain \mbox{$L = O(\log(1/\eps)/\log(1/\lambda))$}; we omit details.
\end{remark}

When using \Cref{thm:derando-walks}, we will often want to disregard a certain ``trivial'' subspace; we will then 
employ the following simple observation:
\begin{fact}    \label{fact:triv}
    In the setting of \Cref{thm:derando-walks},
    say each $U_j$ may be written as $U_j = R_j \oplus U'_j$, where $R_j$ acts on subspace~$T$ and $U'_j$ acts on its orthogonal complement $T^\bot$ in~$\C^r$.  Then $\avg(Q(\calU)) = \avg(Q(\calR)) \oplus \avg(Q(\calU'))$, where $\calR = (R_1, \dots, R_{\c})$ and $\calU' = (U'_1, \dots, U'_{\c})$.
\end{fact}

For example, suppose $G = (V,E)$ is a $d$-regular undirected graph on $V = (1, 2, \dots, n)$ with normalized adjacency matrix expressed as
\begin{equation}    \label{eqn:decomp}
    A_G = \avg(P_1, \dots,  P_d),
\end{equation}
where  $P_1, \dots, P_d$ are $n \times n$ permutation matrices.
Each $P_i$ and $P_i^\dagger$ has operator norm~$1$ and fixes the one-dimensional space $T = \spn\{\ket{1} + \cdots + \ket{n}\}$.
If we write $P_i = \mathrm{proj}_T \oplus U'_i$ where $U'_i$ is the action of $P_i$ on $T^\bot$, then
\begin{equation}
    A_G = \mathrm{proj}_T \oplus \avg(U'_1, \dots, U'_d) 
\end{equation}
and we are in a position to apply \Cref{fact:triv} and \Cref{thm:derando-walks} together.  
The result is a sequence~$Q$ of ``walks'', each of the form $P_{i_L}^\dagger P_{i_{L-1}} \cdots P_{i_2}^\dagger P_{i_1}$.
Applying one such walk to any starting vertex~$\ket{v}$ leads to a valid walk of length~$L$ in~$G$ (with the steps $P_i^\dagger$ being valid since~$G$ is undirected). 
If we write~$\wt{G}$ for the $|Q|$-regular undirected graph on~$V$ wherein each $v \in V$ has an edge to all its walk outcomes, the result is that
\begin{equation}
    A_{\wt{G}} = \avg(Q \circ (P_1, \dots, P_d)) = \mathrm{proj}_T \oplus \avg(Q(U'_1, \dots, U'_d)).
\end{equation}
Hence if $G$ is a $(1-\delta)$-expander, we obtain that $\wt{G}$ is an~$\eps$-expander with $|Q| = O(d/(\delta \eps)^{O(1)})$ and walks of length $O(\log(1/\eps)/\delta^{O(1)})$.
As shown in \cite{Rei08,RV05}, given any simple, connected, $n$-vertex, undirected graph, there is a very simple transformation preserving connectivity that produces a~$4$-regular undirected graph (together with the associated $P_1, \dots, P_4$ as in \Cref{eqn:decomp}) that has $\delta \geq 1/\poly(n)$;  by taking $\eps = 1/\poly(n)$, one can use these pseudorandom walks to establish Reingold's Theorem~$\mathsf{SL} = \mathsf{L}$~\cite{Rei08}.

\bibliographystyle{alpha}
\bibliography{some-papers}

\newcommand{\etalchar}[1]{$^{#1}$}
\begin{thebibliography}{BCHJ{\etalchar{+}}21}

\bibitem[ABI86]{ABI85}
Noga Alon, L\'{a}szl\'{o} Babai, and Alon Itai.
\newblock A fast and simple randomized parallel algorithm for the maximal
  independent set problem.
\newblock {\em Journal of Algorithms}, 7(4):567--583, 1986.

\bibitem[AGHP92]{AGHP92}
Noga Alon, Oded Goldreich, Johan H{\aa}stad, and Ren\'{e} Peralta.
\newblock Simple constructions of almost {$k$}-wise independent random
  variables.
\newblock {\em Random Structures \& Algorithms}, 3(3):289--304, 1992.

\bibitem[AK63]{ArnoldKrylov62}
Vladimir Arnol'd and Alexander Krylov.
\newblock Uniform distribution of points on a sphere and certain ergodic
  properties of solutions of linear ordinary differential equations in a
  complex domain.
\newblock {\em Doklady Akademii Nauk SSSR}, 148:9--12, 1963.

\bibitem[Alo21]{Alo21}
Noga Alon.
\newblock Explicit expanders of every degree and size.
\newblock {\em Combinatorica}, 41(4):447--463, 2021.

\bibitem[BC20]{BC20}
Charles Bordenave and Beno{\^i}t Collins.
\newblock Strong asymptotic freeness for independent uniform variables on
  compact groups associated to non-trivial representations.
\newblock Technical Report 2012.08759, arXiv, 2020.

\bibitem[BCHJ{\etalchar{+}}21]{BCHKP21}
Fernando Brand{\~a}o, Wissam Chemissany, Nicholas Hunter-Jones, Richard Kueng,
  and John Preskill.
\newblock Models of quantum complexity growth.
\newblock {\em PRX Quantum}, 2(3):030316, 2021.

\bibitem[BdS16]{BS16}
Yves Benoist and Nicolas de~Saxc\'{e}.
\newblock A spectral gap theorem in simple {L}ie groups.
\newblock {\em Inventiones Mathematicae}, 205(2):337--361, 2016.

\bibitem[BG12]{BG12}
Jean Bourgain and Alex Gamburd.
\newblock A spectral gap theorem in {${\rm SU}(d)$}.
\newblock {\em Journal of the European Mathematical Society (JEMS)},
  14(5):1455--1511, 2012.

\bibitem[BH99]{Bridson1999}
Martin Bridson and Andr{\'e} Haefliger.
\newblock {\em Length Spaces}, pages 32--46.
\newblock Springer Berlin Heidelberg, 1999.

\bibitem[BH08]{BH08}
Alex Brodsky and Shlomo Hoory.
\newblock Simple permutations mix even better.
\newblock {\em Random Structures \& Algorithms}, 32(3):274--289, 2008.

\bibitem[BHH16]{BHH16}
Fernando Brand{\~a}o, Aram Harrow, and Micha{\l} Horodecki.
\newblock Local random quantum circuits are approximate polynomial-designs.
\newblock {\em Communications in Mathematical Physics}, 346(2):397--434, 2016.

\bibitem[BNZZ19]{BNZZ19}
Eiichi Bannai, Mikio Nakahara, Da~Zhao, and Yan Zhu.
\newblock On the explicit constructions of certain unitary {$t$}-designs.
\newblock {\em Journal of Physics. A.}, 52(49):495301, 17, 2019.

\bibitem[Bra37]{Bra37}
Richard Brauer.
\newblock On algebras which are connected with the semisimple continuous
  groups.
\newblock {\em Annals of Mathematics}, 38(4):857--872, 1937.

\bibitem[DCEL09]{DCEL09}
Christoph Dankert, Richard Cleve, Joseph Emerson, and Etera Livine.
\newblock Exact and approximate unitary 2-designs and their application to
  fidelity estimation.
\newblock {\em Physical Review A}, 80(1):012304, 2009.

\bibitem[DLT02]{DLT02}
David DiVincenzo, Debbie Leung, and Barbara Terhal.
\newblock Quantum data hiding.
\newblock {\em Transactions on Information Theory}, 48(3):580--598, 2002.

\bibitem[FPY15]{FPY15}
Hilary Finucane, Ron Peled, and Yariv Yaari.
\newblock A recursive construction of {$t$}-wise uniform permutations.
\newblock {\em Random Structures \& Algorithms}, 46(3):531--540, 2015.

\bibitem[Gow96]{Gow96}
W.~Timothy Gowers.
\newblock An almost {$m$}-wise independent random permutation of the cube.
\newblock {\em Combinatorics, Probability and Computing}, 5(2):119--130, 1996.

\bibitem[Gro99]{Gro99}
Cheryl Grood.
\newblock Brauer algebras and centralizer algebras for {${\rm SO}(2n,{\bf
  C})$}.
\newblock {\em Journal of Algebra}, 222(2):678--707, 1999.

\bibitem[Haf22]{Haf22}
Jonas Haferkamp.
\newblock Random quantum circuits are approximate unitary $t$-designs in depth
  $o(nt^{5+o(1)})$.
\newblock {\em Quantum}, 6:795, 2022.

\bibitem[HFGW18]{HFGW18}
Anna-Lena Hashagen, Steven Flammia, David Gross, and Joel Wallman.
\newblock Real randomized benchmarking.
\newblock {\em Quantum}, 2:85, 2018.

\bibitem[HHJ21]{HH21}
Jonas Haferkamp and Nicholas Hunter-Jones.
\newblock Improved spectral gaps for random quantum circuits: large local
  dimensions and all-to-all interactions.
\newblock {\em Physical Review A}, 104(2):Paper No.~022417, 18, 2021.

\bibitem[HKOT23]{HKOT23}
Jeongwan Haah, Robin Kothari, Ryan O'Donnell, and Ewin Tang.
\newblock Query-optimal estimation of unitary channels in diamond distance.
\newblock Technical Report 2302.14066, arXiv, 2023.

\bibitem[HL09]{HL09}
Aram Harrow and Richard Low.
\newblock Efficient quantum tensor product expanders and $k$-designs.
\newblock In {\em Proceedings of the 2009 International Workshop on
  Approximation, Randomization, and Combinatorial Optimization (APPROX)}, pages
  548--561. Springer, 2009.

\bibitem[HMMR05]{HMMR05}
Shlomo Hoory, Avner Magen, Steven Myers, and Charles Rackoff.
\newblock Simple permutations mix well.
\newblock {\em Theoretical Computer Science}, 348(2-3):251--261, 2005.

\bibitem[JMRW22]{JMRW22}
Fernando~Granha Jeronimo, Tushant Mittal, Sourya Roy, and Avi Wigderson.
\newblock Almost {R}amanujan expanders from arbitrary expanders via operator
  amplification.
\newblock In {\em 2022 {IEEE} 63rd {A}nnual {S}ymposium on {F}oundations of
  {C}omputer {S}cience---{FOCS} 2022}, pages 378--388. IEEE Computer Soc., Los
  Alamitos, CA, [2022] \copyright 2022.

\bibitem[Kas07]{Kas07}
Martin Kassabov.
\newblock Symmetric groups and expander graphs.
\newblock {\em Inventiones Mathematicae}, 170(2):327--354, 2007.

\bibitem[KM15]{KM15}
Pravesh Kothari and Raghu Meka.
\newblock Almost optimal pseudorandom generators for spherical caps.
\newblock In {\em Proceedings of the 2015 Symposium on the Theory of Computing
  (STOC)}, pages 247--256. ACM, 2015.

\bibitem[KN14]{KN14}
Daniel Kane and Jelani Nelson.
\newblock Sparser {J}ohnson--{L}indenstrauss transforms.
\newblock {\em Journal of the ACM}, 61(1):Art. 4, 23, 2014.

\bibitem[KNR09]{KNR09}
Eyal Kaplan, Moni Naor, and Omer Reingold.
\newblock Derandomized constructions of \emph{k}-wise (almost) independent
  permutations.
\newblock {\em Algorithmica}, 55(1):113--133, 2009.

\bibitem[{Kot}22]{Kothari22}
{Kothari, Pravesh}.
\newblock {Personal communication}.
\newblock 2022.

\bibitem[LK10]{LK10}
Ping Li and Christian K{\"o}nig.
\newblock b-{B}it minwise hashing.
\newblock In {\em Proceedings of the 19th Annual International Conference on
  World Wide Web}, pages 671--680, 2010.

\bibitem[Mec19]{Mec19}
Elizabeth Meckes.
\newblock {\em The random matrix theory of the classical compact groups},
  volume 218.
\newblock Cambridge University Press, 2019.

\bibitem[MOP22]{MOP22}
Sidhanth Mohanty, Ryan O'Donnell, and Pedro Paredes.
\newblock Explicit near-{R}amanujan graphs of every degree.
\newblock {\em SIAM Journal on Computing}, 51(3):STOC20--1--STOC20--23, 2022.

\bibitem[NC10]{NC10}
Michael Nielsen and Isaac Chuang.
\newblock {\em Quantum Computation and Quantum Information}.
\newblock Cambridge University Press, 10th anniversary edition edition, 2010.

\bibitem[NN93]{NN93}
Joseph Naor and Moni Naor.
\newblock Small-bias probability spaces: efficient constructions and
  applications.
\newblock {\em SIAM Journal on Computing}, 22(4):838--856, 1993.

\bibitem[Oli09]{Oliveira09}
Roberto~Imbuzeiro Oliveira.
\newblock {On the convergence to equilibrium of Kac's random walk on matrices}.
\newblock {\em Ann. Appl. Probab.}, 19(3):1200--1231, 2009.

\bibitem[Rei08]{Rei08}
Omer Reingold.
\newblock Undirected connectivity in log-space.
\newblock {\em Journal of the ACM}, 55(4):Art. 17, 24, 2008.

\bibitem[RTV06]{RTV06}
Omer Reingold, Luca Trevisan, and Salil Vadhan.
\newblock Pseudorandom walks on regular digraphs and the {$\mathsf{RL}$} vs.
  {$\mathsf{L}$} problem.
\newblock In {\em Proceedings of the 2006 Symposium on the Theory of Computing
  (STOC)}, pages 457--466. ACM, 2006.

\bibitem[RV05]{RV05}
Eyal Rozenman and Salil Vadhan.
\newblock Derandomized squaring of graphs.
\newblock Electronic Colloquium on Computational Complexity, TR05-092, 2005.

\bibitem[RY17]{RY17}
Daniel Roberts and Beni Yoshida.
\newblock Chaos and complexity by design.
\newblock {\em Journal of High Energy Physics}, 2017(4):1--64, 2017.

\bibitem[Sch01]{Sch01}
Issai Schur.
\newblock {\em {\"U}ber eine {K}lasse von {M}atrizen, die sich einer gegebenen
  {M}atrix zuordnen lassen}.
\newblock PhD thesis, Universit{\"a}t Berlin, 1901.

\bibitem[Sen18]{Sen18}
Pranab Sen.
\newblock Efficient quantum tensor product expanders and unitary $t$-designs
  via the zigzag product.
\newblock Technical Report 1808.10521, arXiv, 2018.

\bibitem[Shi02]{Shi02}
Yaoyun Shi.
\newblock Both {T}offoli and {C}ontrolled-{NOT} need little help to do
  universal quantum computation.
\newblock Technical Report quant-ph/0205115, arXiv, 2002.

\bibitem[Tao14]{Tao14}
Terence Tao.
\newblock {\em Hilbert's fifth problem and related topics}, volume 153.
\newblock American Mathematical Society, 2014.

\bibitem[Web16]{Web16}
Zak Webb.
\newblock The {C}lifford group forms a unitary 3-design.
\newblock {\em Quantum Information \& Computation}, 16(15-16):1379--1400, 2016.

\bibitem[Wey39]{Wey39}
Hermann Weyl.
\newblock {\em The {C}lassical {G}roups. {T}heir {I}nvariants and
  {R}epresentations}.
\newblock Princeton University Press, 1939.

\bibitem[{Wik}23]{Wiki:projectiveunitarygroup}
{Wikipedia contributors}.
\newblock {Projective unitary group}.
\newblock
  \href{https://en.wikipedia.org/wiki/Projective\_unitary\_group}{https://en.wikipedia.org/wiki/Projective\_unitary\_group},
  accessed September 16, 2023.

\bibitem[Yua12]{Yua12}
Qiaochu Yuan.
\newblock Four flavors of {S}chur--{W}eyl duality, 2012.
\newblock
  \url{https://qchu.wordpress.com/2012/11/13/four-flavors-of-schur-weyl-duality/}.

\bibitem[ZZP17]{ZZP17}
Linxi Zhang, Chuanghua Zhu, and Changxing Pei.
\newblock Randomized benchmarking using unitary $t$-design for average fidelity
  estimation of practical quantum circuit.
\newblock Technical Report 1711.08098, arXiv, 2017.

\end{thebibliography}

\end{document}